\documentclass[10pt]{article}

\usepackage{grffile}

\usepackage{amssymb}
\usepackage{amsmath}
\usepackage{mathrsfs}
\usepackage[dvips]{epsfig}
\usepackage[footnotesize]{caption}
\usepackage{graphicx}
\usepackage[all]{xy}

\usepackage{tikz}
\usepackage{ctable}

\usepackage{pgfplots}
\usepackage{amsthm}

\newtheorem{Lemma}{Lemma}[section]
\newtheorem{Theorem}{Theorem}
\newtheorem{Proposition}[Lemma]{Proposition}
\newtheorem{Corollary}[Lemma]{Corollary}
\newtheorem{Remark}[Lemma]{Remark}

 {\begin{trivlist} \item[]{\bf Proof. }}%
 {\hspace*{\fill}$\rule{.4\baselineskip}{.4\baselineskip}$\end{trivlist}}

 {\begin{trivlist}\item[]\textbf{Acknowledgments.}}{\end{trivlist}}

\makeatletter\@addtoreset{figure}{section}\makeatother

\makeatletter \@addtoreset{equation}{section} \makeatother

\newcommand{\R}{\mathbb{R}}

\def\Re{\mathop{\mathrm{Re}}}

\newcommand{\rmO}{\mathrm{O}}
\newcommand{\rmo}{\mathrm{o}}

\newcommand{\rme}{\mathrm{e}}
\newcommand{\rmi}{\mathrm{i}}

\newcommand{\eps}{\varepsilon}

\newcommand{\Su}{\mathcal{S}_U}
\newcommand{\Sd}{\mathcal{S}_D}
\makeatletter
\newcommand*{\defeq}{\mathrel{\rlap{%
                     \raisebox{0.3ex}{$\m@th\cdot$}}%
                     \raisebox{-0.3ex}{$\m@th\cdot$}}%
                     =}
\makeatother

\newcommand{\lindex}{\ell}

\makeatletter
\newsavebox{\@brx}
\newcommand{\llangle}[1][]{\savebox{\@brx}{\(\m@th{#1\langle}\)}%
  \mathopen{\copy\@brx\kern-0.5\wd\@brx\usebox{\@brx}}}
\newcommand{\rrangle}[1][]{\savebox{\@brx}{\(\m@th{#1\rangle}\)}%
  \mathclose{\copy\@brx\kern-0.5\wd\@brx\usebox{\@brx}}}
\makeatother

\definecolor{Green}{rgb}{0.,0.4,0.}

\renewcommand{\leq}{\leqslant}
\renewcommand{\geq}{\geqslant}

\newcommand{\scflux}{\tilde{\theta}}

\newcommand{\Rmnum}[1]{\uppercase\expandafter{\romannumeral #1\relax}}

\def\XXint#1#2#3{{\setbox0=\hbox{$#1{#2#3}{\int}$}
     \vcenter{\hbox{$#2#3$}}\kern-.5\wd0}}

\newfam\bifam
\font\tenbi=cmmib10 scaled \magstep1 \font\sevenbi=cmmib10 at 11pt
\font\fivebi=cmmib10 at 6pt \textfont\bifam = \tenbi
\scriptfont\bifam = \sevenbi \scriptscriptfont\bifam= \fivebi

\ifx\pdfoutput\undefined
   \pdffalse
\else
   \pdfoutput=1
   \pdftrue
\fi
\ifpdf
   \usepackage{graphicx}
   \usepackage{epstopdf}
\else
   \usepackage{graphicx}
\fi

\begin{document}

\begin{center}

{\fontsize{17}{17}\fontfamily{cmr}\fontseries{b}\selectfont{    
Advection and autocatalysis as organizing principles for banded vegetation patterns
}}\\[0.2in]
Richard Samuelson$\,^1$, Zachary Singer$\,^2$, Jasper Weinburd$\,^3$, Arnd Scheel$\,^3$\\
\textit{\footnotesize $\,^1$Trinity College,  300 Summit St, Hartford, CT 06106 , USA}\\
\textit{\footnotesize $\,^2$Department of Mathematical Sciences, Wean Hall 6113, Carnegie Mellon University, Pittsburgh, PA 15213  ,USA}\\
\textit{\footnotesize 
$\,^3$University of Minnesota, School of Mathematics,   206 Church St. S.E., Minneapolis, MN 55455, USA}\\

\date{\small \today} 
\end{center}

\begin{abstract}
\noindent 
We motivate and analyze a simple model for the formation of banded vegetation patterns. The model incorporates a minimal number of ingredients for vegetation growth in semi-arid landscapes. It allows for comprehensive analysis and sheds new light onto phenomena such as the migration of vegetation bands and the interplay between their upper and lower edges. The key ingredient is the formulation as a closed reaction-diffusion system, thus introducing a conservation law that both allows for  analysis and provides ready intuition and understanding through analogies with characteristic speeds of propagation and shock waves. 
\end{abstract}
%

\section{Banded vegetation patterns --- phenomena, questions, and a simple model}

The formation of banded vegetation patterns has been understood as a self-organizing mechanism that allows vegetation to cope with scarcity of resources by leveraging beneficial effects of high density soil occupation; see for instance the reviews \cite{borgogno,meron2015nonlinear} and references therein. Present in semi-arid and some arid climates, where one might expect sensitive dependence of vegetation patterns on climate variations,  these patterns are surprisingly robust. Recent analyses of satellite images and aerial photographs show very little variation in the patterns over time spans as long as 50 years.  Modeling efforts in this context are particularly difficult, not only because of the intrinsic complexity of ecological systems, but because of the scarcity of time-dependent data that could be used to validate models; see however the recent study \cite{silber3}. On the other hand, the inherent fragility of vegetation, the danger and irreversibility of desertification, and the difficulty of controlled experiments, make it highly desirable to predict dynamics theoretically, in particular the dependence of vegetation densities on parameters and the possibility of tipping points. 

Our interest here is in a class of macroscopic models in the literature that track vegetation densities, nutrients, and water, possibly accounting for different roles of surface water and subsurface water. A common ingredient to many models is a growth rate for vegetation densities that increases with the density, encoding a symbiotic effect of plant growth due to a variety of factors such as reduced soil erosion, water binding, and protection from sunlight. In the Klausmeier model, this autocatalytic effect is reflected in  kinetic growth rates $b^2w$, where $b$ is a vegetation density (biomass) and $w$ measures water densities \cite{klaus}. In most models, autocatalytic growth is supplemented with linear death rates $-b$, and a variety of source and transport terms, modeling rain fall, water evaporation, etc. Models for spatial transport vary from simple advective transport and diffusion of water paired with diffusive spread of vegetation, to modeling porous media flow and nonlocal dispersal of plant seeds; see for instance \cite{rietkerk,meron,rietkerk2}. From a mathematical point of view, the analysis of such models is often focused on Turing-type linear stability analysis, predicting the spontaneous formation of periodic structures with a wavenumber given through a linear maximal growth calculation. More refined methods then allow one to study transitions between vegetation bands, spots, and gaps in uniform vegetation; see for instance \cite{godwa1,silber,lejeune,meron2}. Here, one envisions that the small amplitude variations in vegetation density found from a weakly nonlinear analysis near certain thresholds predict well the dynamics and patterns far from this threshold, with possibly large variations of densities and steep gradients, as often observed in nature. A technically complementary analysis focuses on separation of spatial scales as a means of systematically building and understanding spatial patterns, exploiting for instance disparities in spatio-temporal scales for water transport versus biomass evolution; see for instance \cite{sewalt,siero}. In a different direction, the actual formation process of banded patterns may well have a crucial role in the selection of banded patterns: colonization through spreading rather than emergence after a spontaneous and uniform change in the environment can produce quite different resulting patterns; see for instance \cite{sherratt1,sherratt2}.

From this technical point of view, our effort here can be seen as providing a different building block for the analysis of patterns in such systems, seeking the simplest model that can yet reproduce many of the complex patterns observed. Starting with the understanding of such a bare-bones model, we hope that one can more systematically argue for the relevance of more complex processes for the phenomena observed. Our focus is on the formation of banded patterns in a uniformly sloped environment, in contrast to many of the above studies. We will comment only briefly in our discussion on the equivalent analysis in the absence of advection. 

\paragraph{Our model.}
To set up our model, we encode the state of the system through two variables, $b$ and $w$, that represent biomass bound to the soil and nutrients dissolved in water, hence subject to advection. Transport of biomass $b$ is diffusive with rate $d_b$, while nutrients  are simply advected with constant speed $c$ determined in particular by surface slope. Kinetics are as simple as possible, with a single rate function $r(b,w)$ modeling the conversion of nutrients in water $w$ to biomass $b$ bound to the soil, 
\[
r(b,w)=\alpha b^2w-\mu b,
\]
with positive rate constants $\alpha$ for the autocatalytic effect on growth and $\mu$ for the mortality. The resulting system of partial differential equations then is
\begin{align}
 b_t&=d_b \Delta b + \alpha b^2 w - \mu b,\notag\\
 w_t&=cw_x -\alpha b^2 w + \mu b,\label{e:rd}
\end{align}
where $b=b(t,x)$, $w=w(t,x)$, and subscripts denote partial derivatives.
Note that we posited constant speed of advection, corresponding to an idealized terrain with constant slope, where water is being advected towards negative $x$. 

The main difference to the Klausmeier model \cite{klaus} is the absence of source terms and the introduction of a conservation law. Specifically, Klausmeier's model adds a source term $A$ for rain fall and an evaporation term $-Bw$ into the $w$-equation, but does not take reinsertion of nutrients after decay, the term  $+\mu b$, into account. We model nutrients and biomass, which we presume are conserved, either as vegetation bound to the soil, or as nutrients advected with water. As a consequence, we obtain  the conservation law 
\begin{equation}\label{e:rdcl}
 \partial_t \int_\Omega (b+w)= \int_{\partial\Omega} \left(d_b\partial_\nu b  + c\,\mathrm{sign}\, (\nu\cdot e_x)\right),
\end{equation}
that is, the sum of nutrients and biomass changes only due to diffusive and advective fluxes.

Scaling time, space, and $(b,w)$, one can readily obtain $d_b=1$, $\alpha=1$, and $\mu =1$, arriving at
\begin{align}
 b_t&=\Delta b +  b^2 w -  b,\notag\\
 w_t&=cw_x -b^2 w + b.\label{e:rds}
\end{align}
We emphasize that we do not claim that water is conserved on time scales relevant for the evolution of vegetation patterns --- evaporation and rain fall clearly play significant rolls in the dynamics. We rather think of $w$ as the concentration of certain nutrients contained in water, and released back upon plant decay, with the autocatalytic plant growth as a key but clearly not the sole ingredient to the ecological dynamics. Somewhat more generally, the equations describe simple autocatalytic mass-action kinetics $2B+W\to 3B$ with rate $wb^2$, in a reactor where the reactant $W$ is supplied through advection in a liquid or gaseous phase and the product $B$ is insoluble, subject to (slower) diffusion. As opposed to general reaction-diffusion models, this model describes a closed reactor, where reactants are supplied through an explicitly modeled flow process. We comment briefly on related study of such closed reaction processes in biology, ecology, and physical chemistry in the discussion section \cite{gohmesuroscheel2,gohmesuroscheel,morita,kotzagiannidis,morita2,morikeshet,champ2,champneys}.

\paragraph{Main contributions.} Our main results exemplify two features of \eqref{e:rds}. First, the core  part of this paper contains a comprehensive analysis of \emph{traveling waves} to \eqref{e:rds}. Despite its rather simple structure with few parameters and the constraint of a conservation law, the model allows for interesting complexity. Simultaneously, it is amenable to an almost complete analytical description and therefore quite explicit predictions. Second, we add interpretation to the traveling-wave analysis and the more general dynamics of the equation by providing a partly rigorous, partly formal analogy to the dynamics of \emph{scalar viscous conservation laws}, relating patterns observed here to Riemann problems, shocks, and rarefaction waves. 

It turns out that, due to an additional scaling in the traveling-wave equation, one can roughly characterize traveling waves in terms of a total flux of biomass and nutrients, only.  This flux is typically equivalent to a prescribed uphill concentration $w_+$. Our main results can be informally summarized as follows. 

\begin{enumerate}
 \item \emph{Small disturbances} of vegetation zones move uphill with positive group velocity; disturbances of vege\-tation-free zones are advected downhill; see Figure \ref{f:1a}.
 \item \emph{Vegetation zones at low densities} are unstable against sideband instabilities, leading to spatially disorganized patterns; see Figure \ref{f:d} and Figure \ref{f:1r}.
\item \emph{Upper edge:} a uniform vegetation zone can spread uphill into a vegetation-free zone; see Figure \ref{f:1b}.
 \item \emph{Lower edge:} a vegetation-free zone can spread uphill into a vegetation zone; see Figure \ref{f:1b}.
 \item \emph{Passive edges:}  upper and lower edges move slowly, with the group velocity of the vegetation state, for  high $w_+<w_\mathrm{u}^*$ and sufficiently high $w_+>w_{\lindex}^*$, respectively.
 \item \emph{Single bands} can move uphill for sufficiently large $w_+>w_\mathrm{s}^*$ (with speeds significantly lower than upper edges); see Figure \ref{f:d}.
 \item \emph{Single gaps} can spread uphill for intermediate ranges of $w_+$.
 \item \emph{Periodic vegetation bands:}  exist in a parameter region slightly larger than single bands; see Figure \ref{f:d}.
\end{enumerate}

From the point of view of scalar viscous conservation laws, lower and upper edges are \emph{undercompressive shocks} that act as organizing centers in Riemann problems, possibly with a glancing mode. Many more complex structures can be understood as bound states between these undercompressive shocks and simpler Lax shocks. We also note that our results present the possibly simplest explanation of the somewhat counterintuitive uphill motion of vegetation bands, quantified recently in \cite{silber3}, against the direction of advective transport, by relating the transport to to a simple calculation of group velocities.

\begin{figure}
 \includegraphics[width=0.3\textwidth]{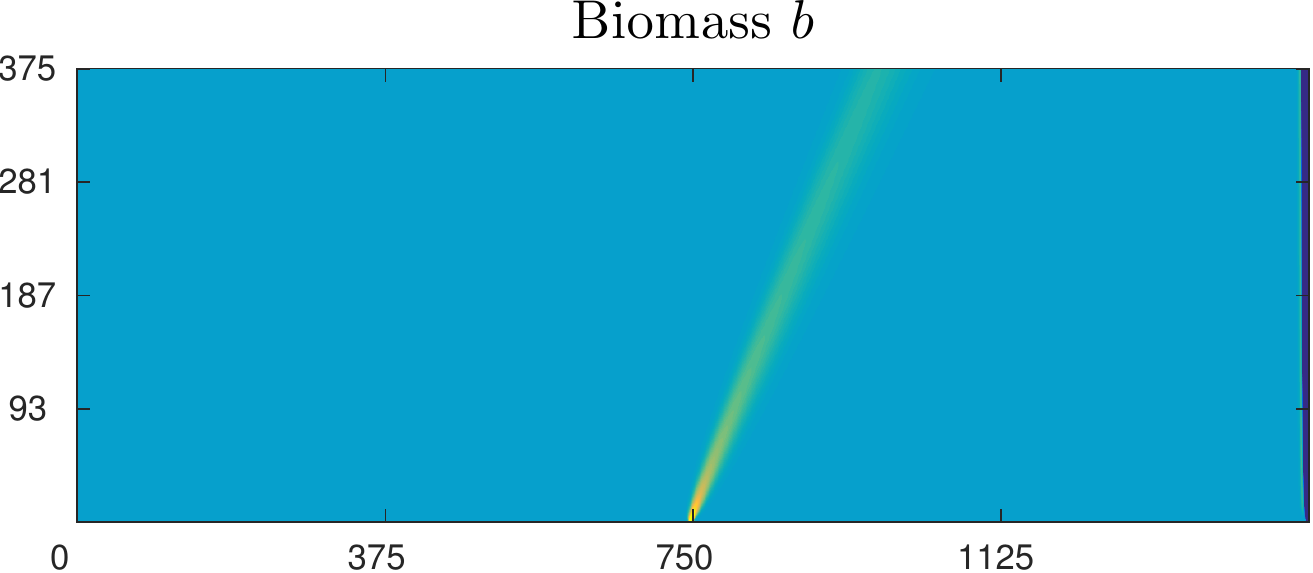}\hfill
\includegraphics[width=0.34\textwidth]{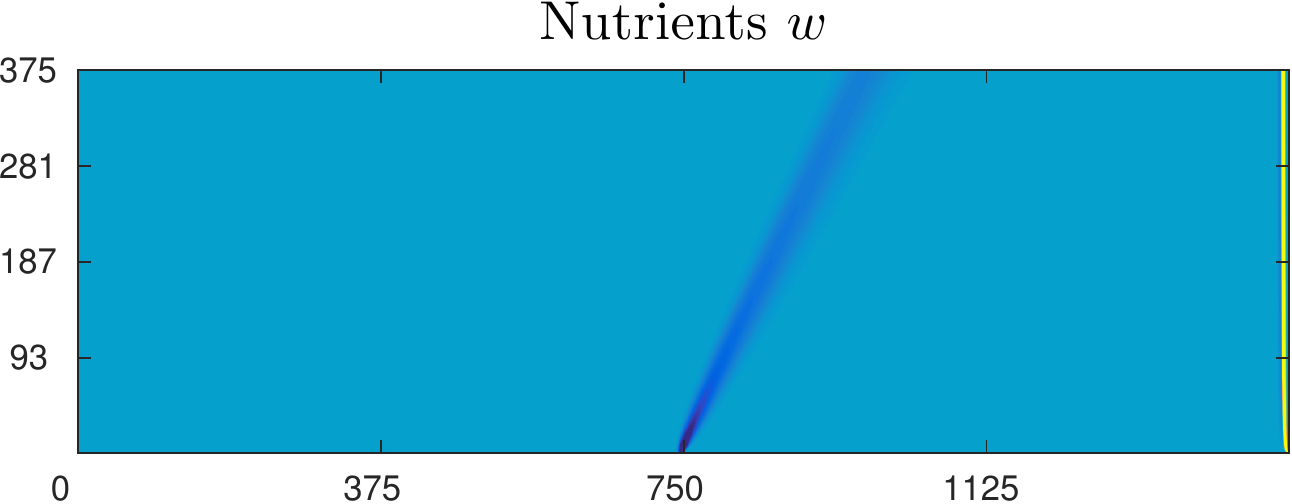}\hfill
 \includegraphics[width=0.3\textwidth]{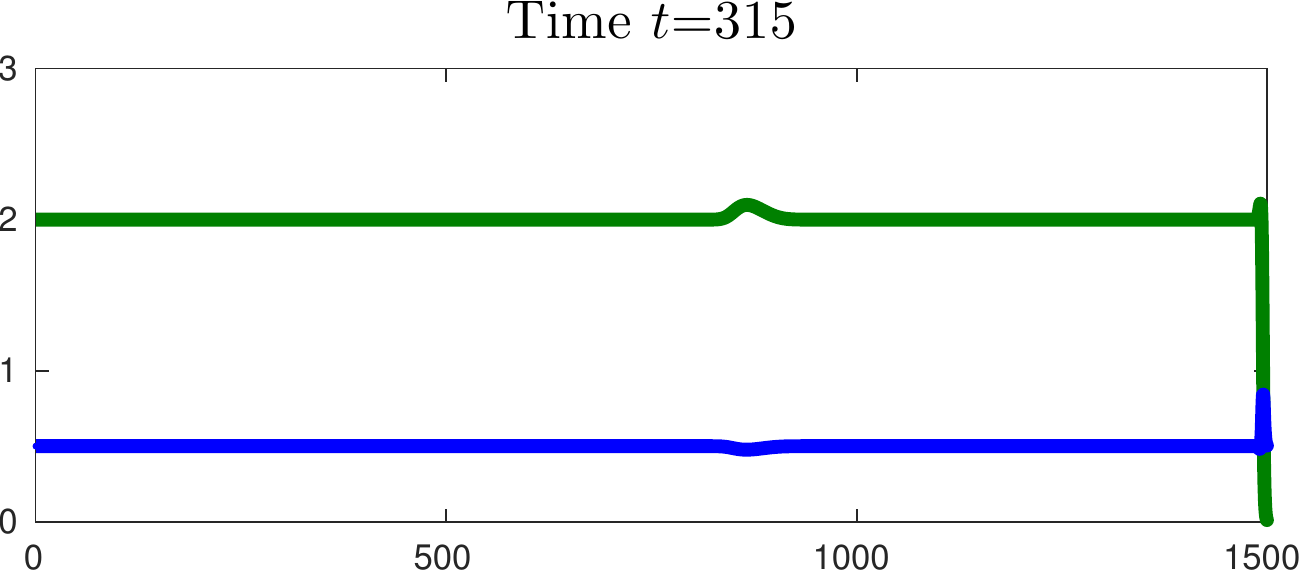}\\
 \includegraphics[width=0.3\textwidth]{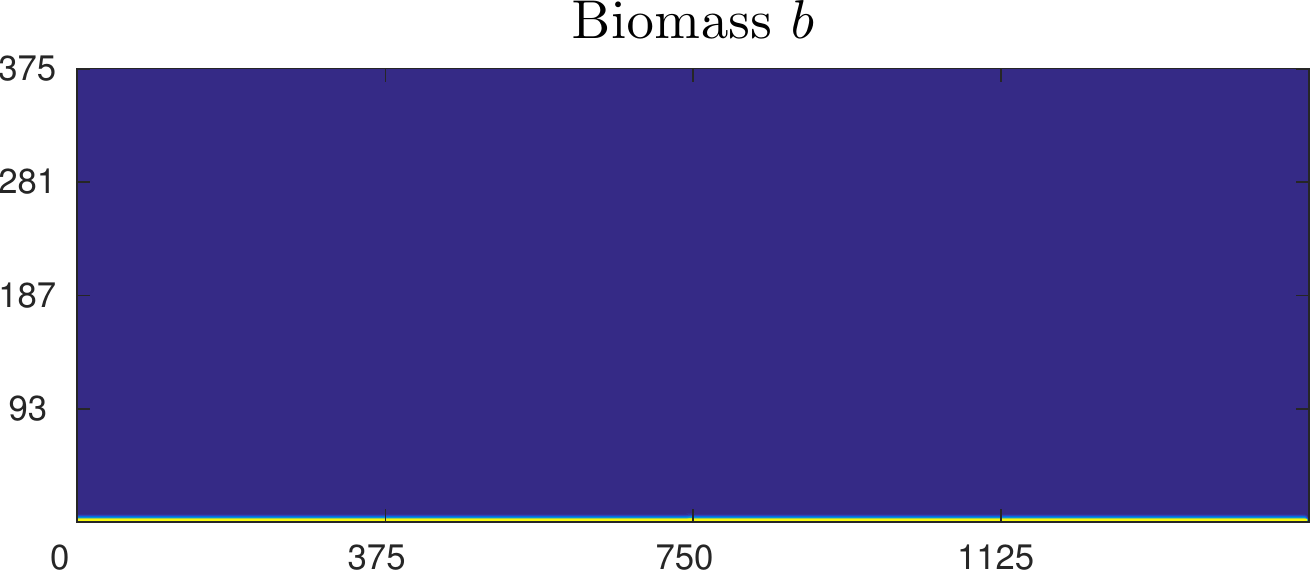}\hfill
\includegraphics[width=0.34\textwidth]{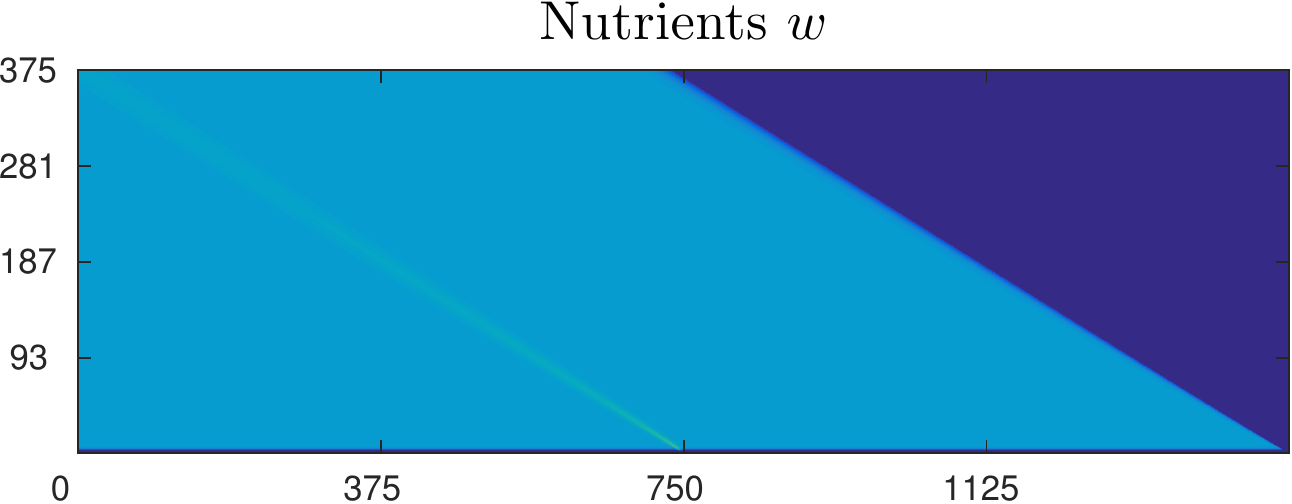}\hfill
 \includegraphics[width=0.3\textwidth]{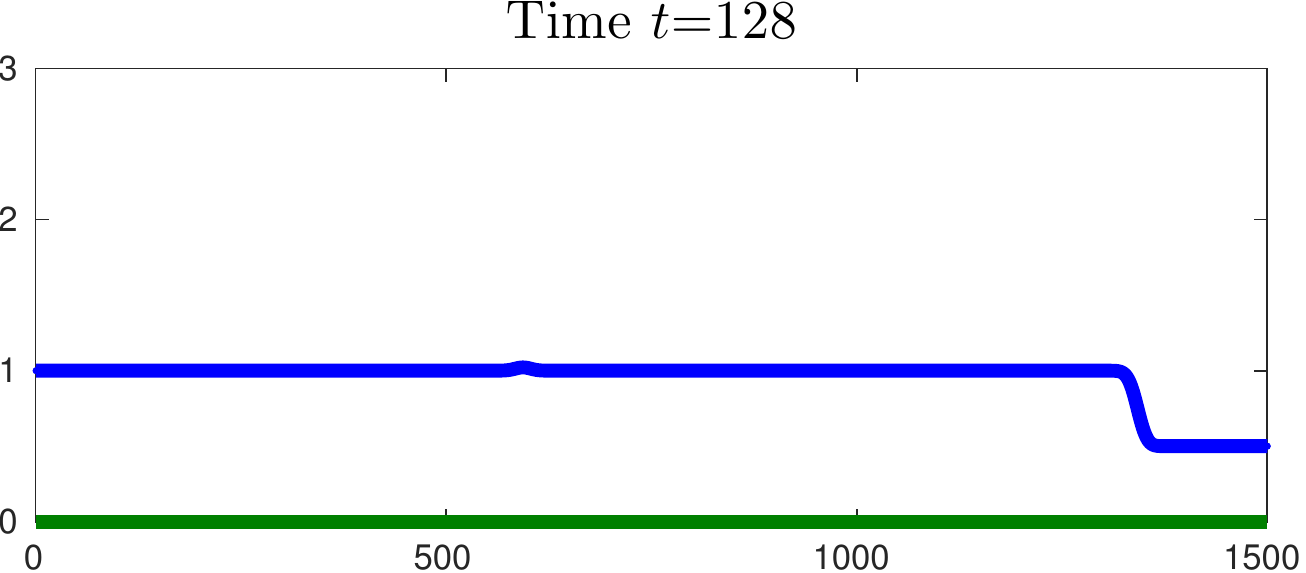}
 \caption{Upstream transport in vegetation zones (top)  and downstream transport in desertified zones (bottom), illustrating (i) and (ii). All numerical simulations carried out on a fixed grid with $dx=0.1$, using upwind first-order  discretization of the advection term, second-order finite differences for the diffusion, and \textsc{Matlab}'s \textsc{ode15s} for time integration. Throughout, densities from dark blue (minimum) to yellow (maximum) in space-time plots (time pointing upward, advection of $w$ pointing to the left); snapshots on the right, with biomass (green) and nutrients (blue).  }\label{f:1a}
\end{figure}

\begin{figure}[b]
 \includegraphics[width=0.3\textwidth]{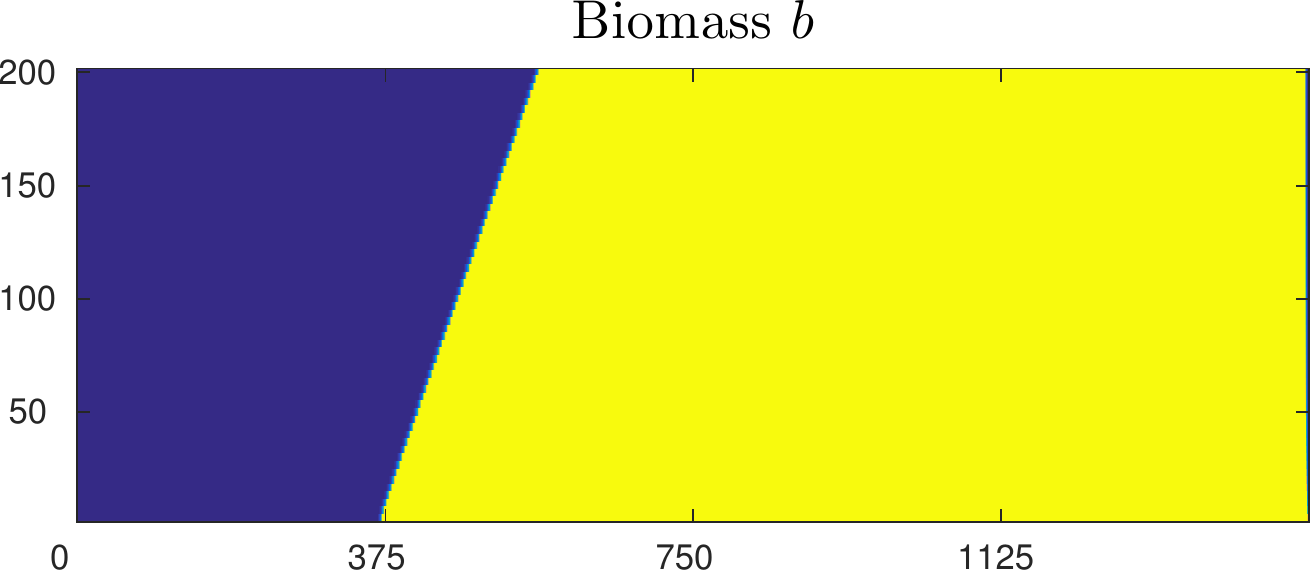}\hfill
\includegraphics[width=0.34\textwidth]{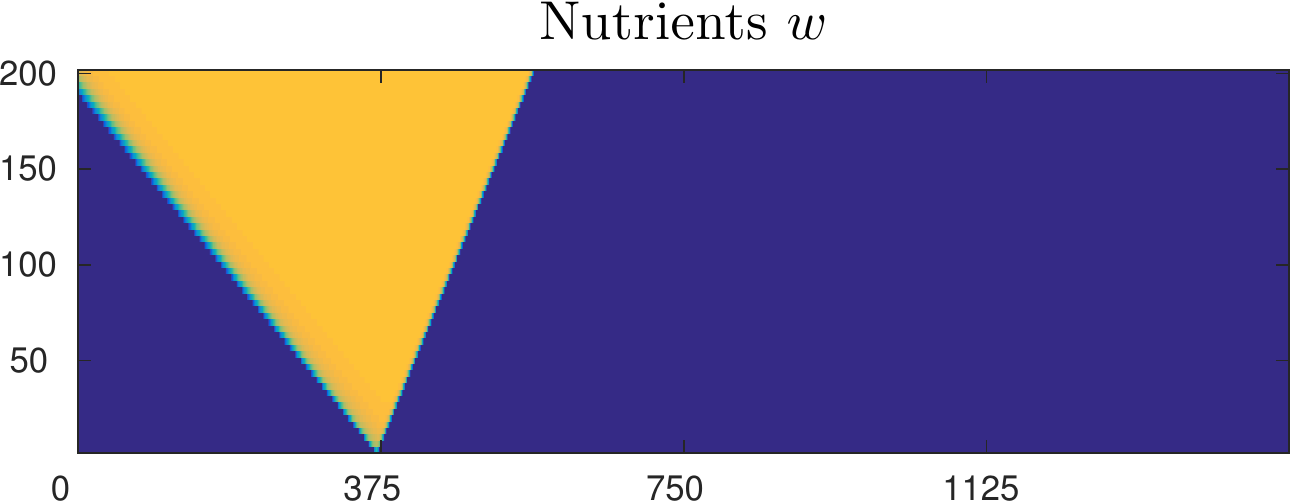}\hfill
 \includegraphics[width=0.3\textwidth]{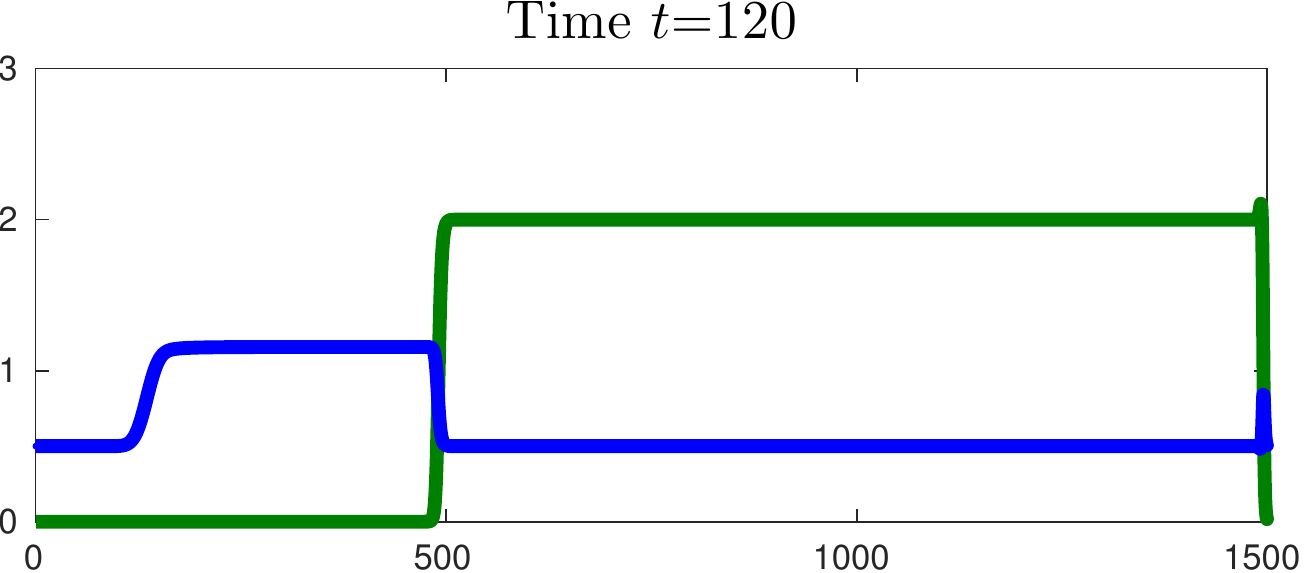}\\
 \includegraphics[width=0.3\textwidth]{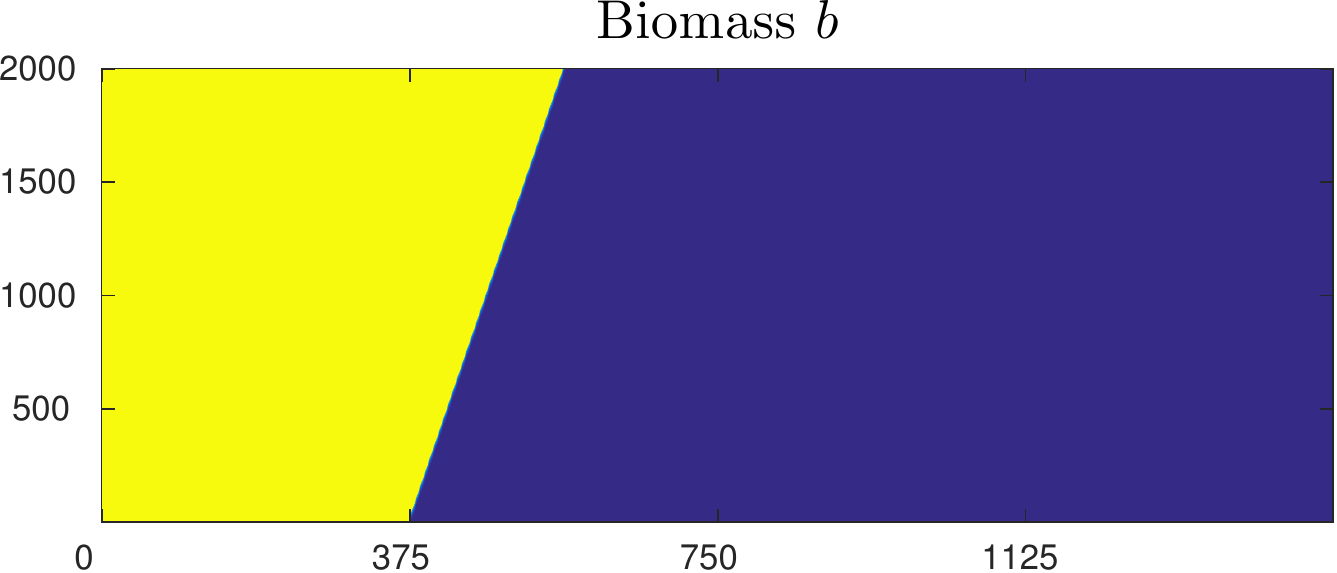}\hfill
\includegraphics[width=0.34\textwidth]{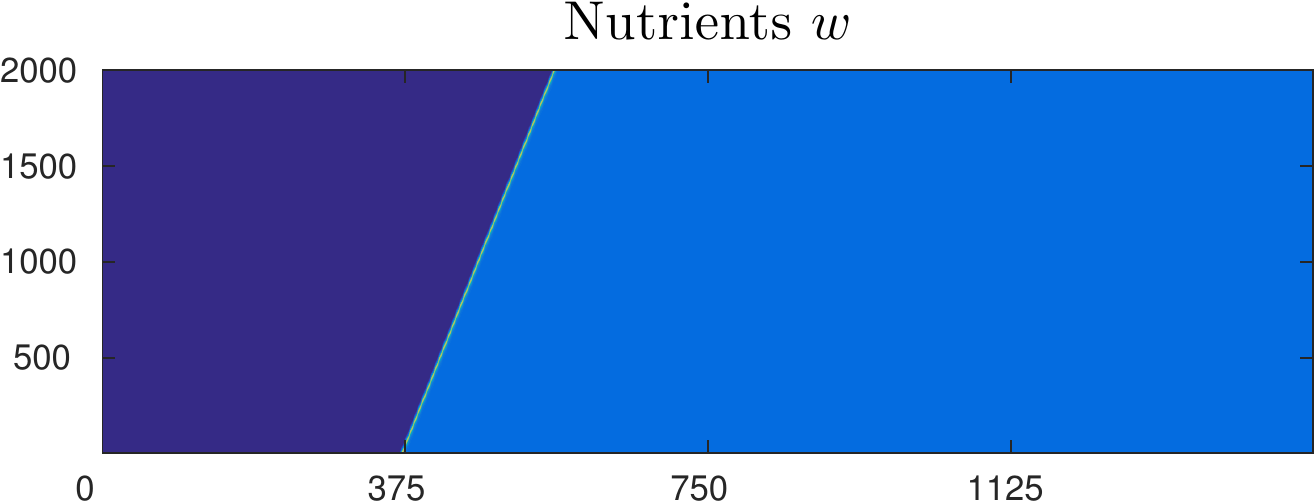}\hfill
 \includegraphics[width=0.3\textwidth]{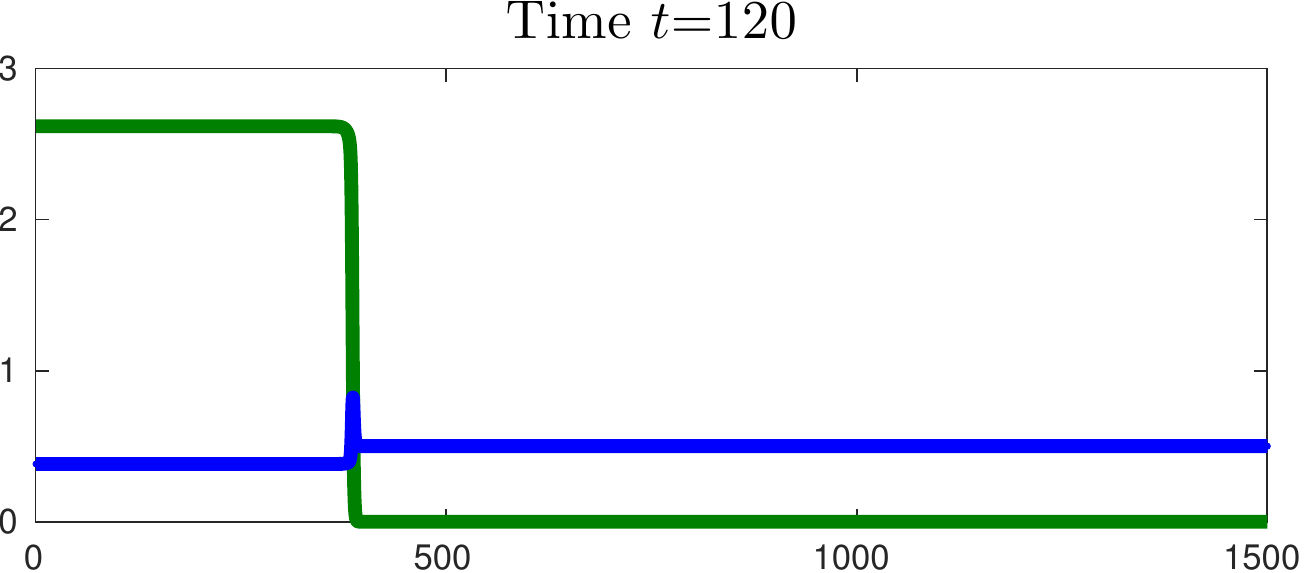}
 \caption{Upper (top)  and lower (bottom) edges of vegetation zones propagating uphill as described in (iii) and (iv).  }\label{f:1b}
\end{figure}

\begin{figure}
 \includegraphics[width=0.3\textwidth]{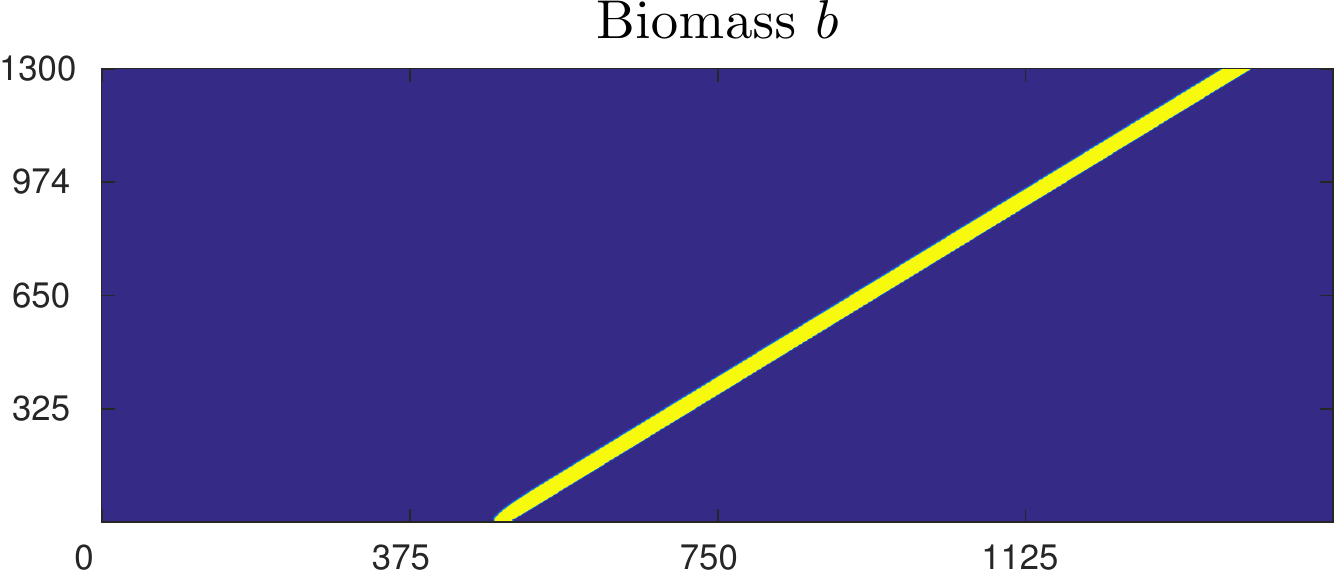}\hfill
\includegraphics[width=0.34\textwidth]{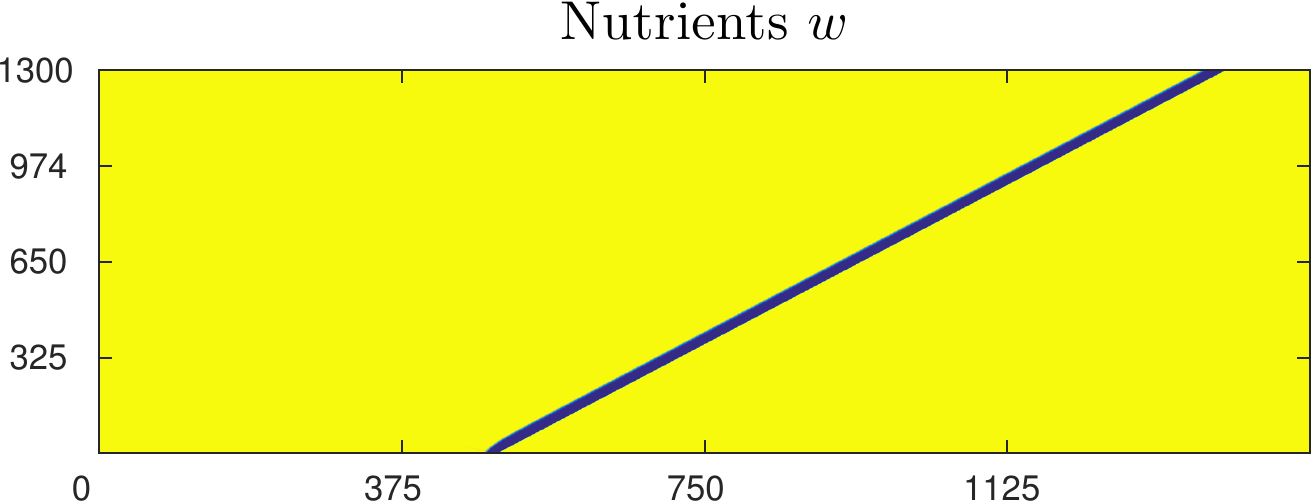}\hfill
 \includegraphics[width=0.3\textwidth]{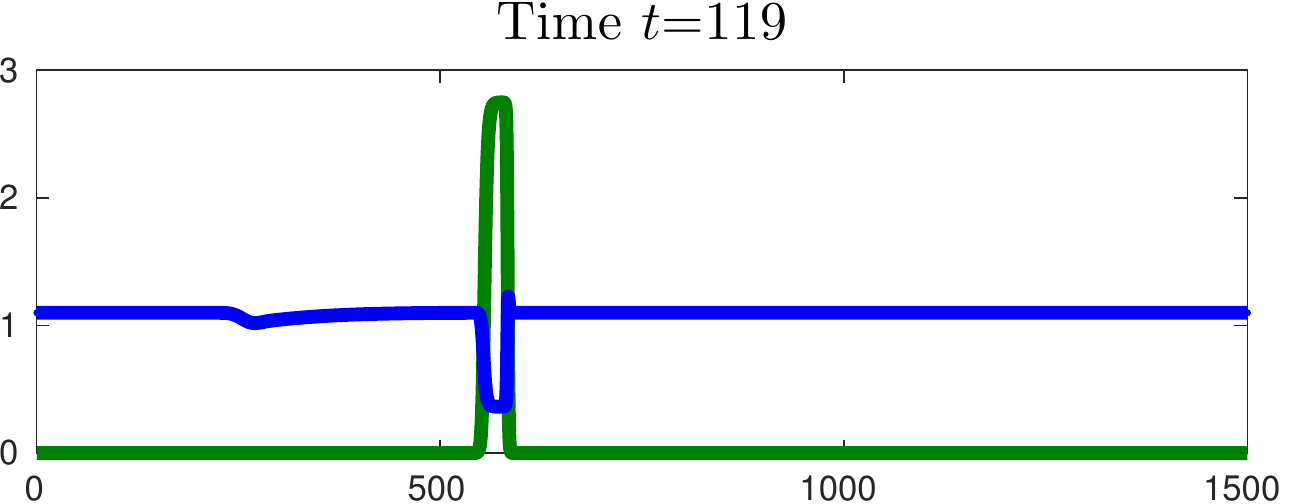}\\
 \includegraphics[width=0.3\textwidth]{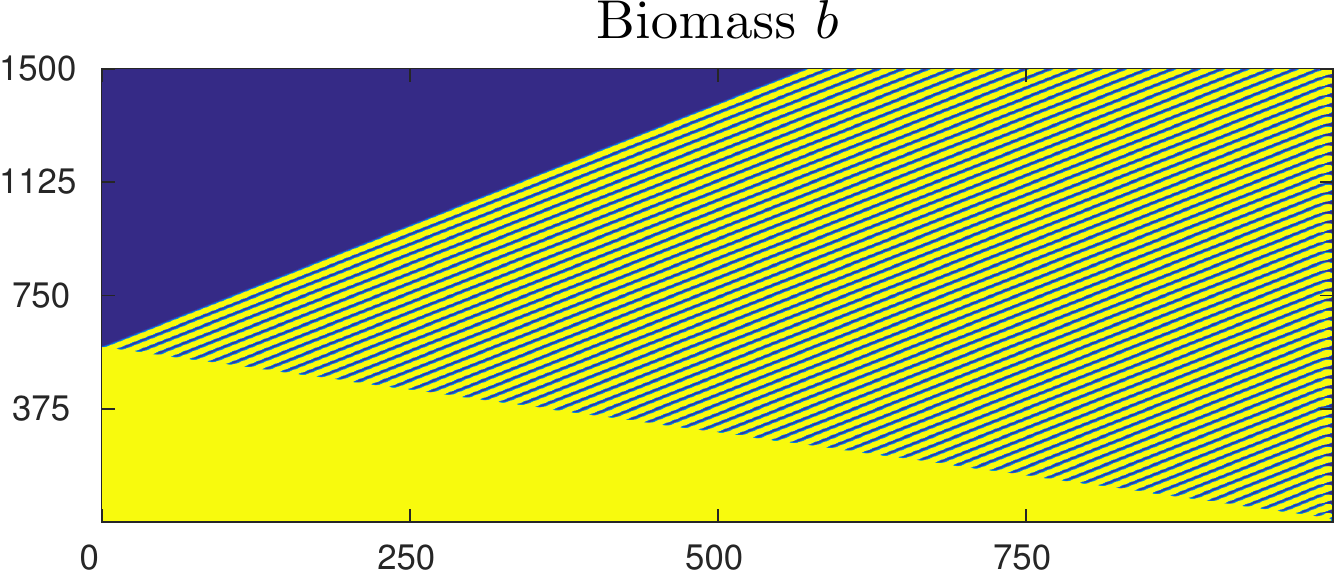}\hfill
\includegraphics[width=0.34\textwidth]{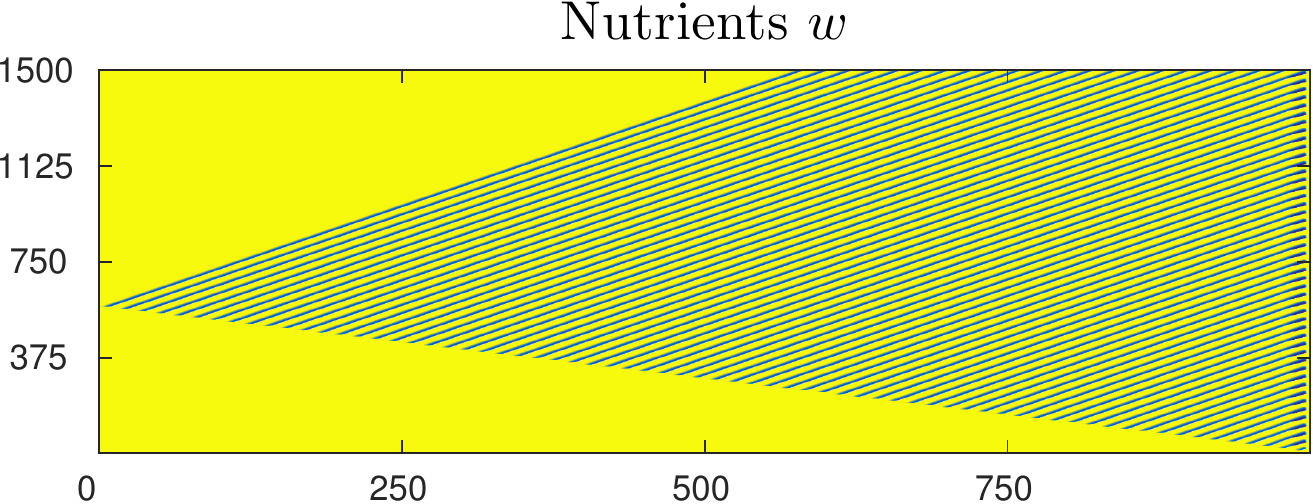}\hfill
 \includegraphics[width=0.3\textwidth]{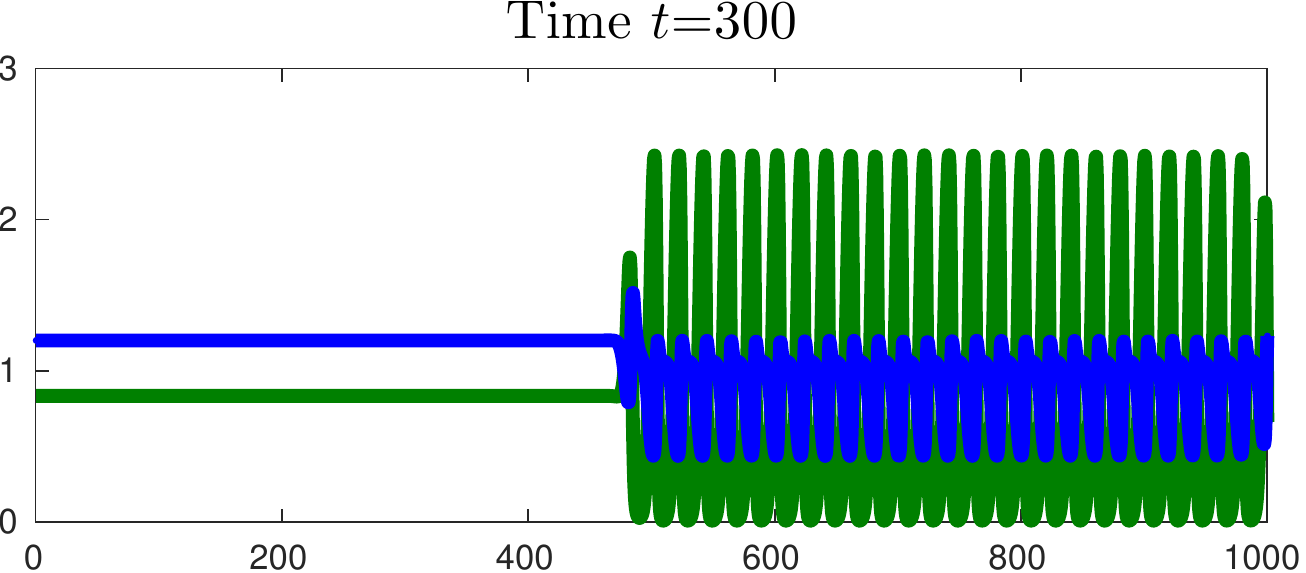} 
 \caption{Single vegetation bands, propagating uphill (top), (vi). 
Formation of regular periodic banded patterns, originating from a sideband unstable state through a perturbation at the right boundary (ii), (viii).}\label{f:d}
\end{figure}
%

%
%
%
%
%
%
%
%

\begin{Remark}[Conservation laws --- terminology]\label{r:cl} It is important to distinguish between the fact that our system \eqref{e:rd} possesses a ``conservation law'' \eqref{e:rdcl}, and the fact that we will view effective dynamics of \eqref{e:rd} as being conjugate in an appropriate sense to the dynamics of scalar viscous ``conservation laws''. In the latter sense, the term conservation law refers to the more narrow class of equations
\begin{equation}\label{e:scl}
 u_t=(d(u)u_x)_x -f(u)_x,\qquad u\in\R,\ x\in\R,
\end{equation}
with flux $f$ and viscosity $d$, whereas in the former sense, we are simply referring to the presence of a conserved quantity. Throughout, we will use the term conservation law to refer to the latter narrow class of equations. 
\end{Remark}

\paragraph{Outline.} We discuss the dynamics of the ODE and the associated PDE stability in Section \ref{s:ode}. Section \ref{s:cl} describes the connection with scalar viscous conservation laws. 
Section \ref{s:tw} contains our main results on traveling waves. We prove existence of heteroclinic orbits in Section \ref{s:p} and conclude with a discussion, Section \ref{s:d}.

\paragraph{Acknowledgment.} This work was supported through grant NSF DMS--1311740. Most of the analysis was carried out during an NSF-funded REU project on Complex Systems at the University of Minnesota in Summer 2017. The authors gratefully acknowledge conversations with Arjen Doelman and Punit Gandhi, who pointed to many of the referecnces included here and provided  many helpful comments and suggestions on an early version of the manuscript.

\section{Spatially constant equilibria and linear properties}\label{s:ode}

Spatially constant solutions satisfy
\[
 b_t=b^2w-b,\qquad w_t=-b^2w+b,
\]
with two curves of equilibria $\Gamma_0=\{b=0,w\geq 0\}$ and $\Gamma_1=\{bw=1,\,w>0\}$. The equilibria in $\Gamma_0$ and in $\Gamma_1\cap \{b>1\}$ are stable, equilibria with $0<b<1$ are unstable. For PDE stability, we consider the linearized equation at an equilibrium $(b_*,w_*)$ after Fourier-Laplace transform, $b,w\sim \rme^{\lambda t + \rmi (kx+\ell y)}$,
\begin{align}
 \lambda b&=-(k^2+\ell^2) b + (2b_*w_*-1)b + b_*^2 w,\notag\\
 \lambda w&= c\rmi k w - (2b_*w_*-1)b - b_*^2 w. \label{e:spec}
\end{align}
On the two stable branches, one finds two (explicit) eigenvalues $\lambda_{1/2}(k,\ell)$, with $\Re\lambda_2(k,\ell)<0$ for all $k$, and $\lambda_1(0,0)=0$, reflecting mass conservation as a neutral eigenvalue. One verifies that $\Re\lambda_1(k,\ell)\leq 0$ for all $k$ if and only if $\Re\lambda_{1,kk}(0,0)\leq 0$, that is, the long-wavelength expansion around the neutral mode determines stability. Expanding $\lambda_1$, one finds
\[
  \lambda_1(k)=-c_\mathrm{g}\rmi k - d_{\mathrm{eff},x} k^2  - d_{\mathrm{eff},y}\ell^2+\rmO(|k|^3+|\ell|^3),
 \]
 where $c_\mathrm{g}=-c$ and $d_\mathrm{eff}=0$ for $b_*=0$, and 
 \begin{equation}
  c_\mathrm{g}= 
  \frac{c}{b_*^2-1},\qquad  d_{\mathrm{eff},x}=
 \frac{b_*^2}{(b_*^2-1)^3} \left( (b_*^2-1)^2-c^2\right),\qquad  d_{\mathrm{eff},y}=b_*^2-1,
 \label{e:cg}
\end{equation}
for $b_*>1$. The notation $c_\mathrm{g}$ and $d_\mathrm{eff}$ refers to group velocity and effective diffusivities of long-wavelength modulations of total mass, and can be understood as coefficients in the reverse Fourier transform of the long-wavelength expansion, $\Phi_t=-c_\mathrm{g}\Phi+d_{\mathrm{eff},x}\Phi_{xx}+d_{\mathrm{eff},y}\Phi_{yy}$. We emphasize that from this simple calculation, we conclude that disturbances of vegetation-free states are advected ``downhill'' with speed $c$, as expected, while disturbances of the vegetation state $b_*>1$ are transported ``uphill'' with speed $c_g>0$. 

\begin{figure}[b]
\centering\includegraphics[width=0.7\textwidth]{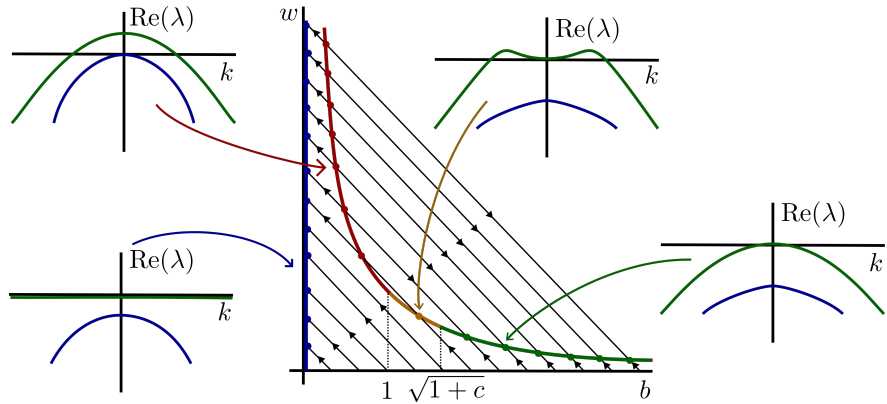}
\caption{Dynamics of the kinetics in the $b$-$w$-phase plane with equilibria and spectra of the linearization from \eqref{e:spec}. }\label{f:lin}
\end{figure}

As a consequence, spatially constant states are linearly (marginally) stable when 
\begin{equation}\label{e:st}
 d_{\mathrm{eff},x}>0, \quad \text{ that is, } \quad b>\sqrt{1+c};
\end{equation}
see Figure \ref{f:lin} for schematic plots of eigenvalues and Figure \ref{f:1r} for the evolution of instabilities. Equivalently, we see that vegetation states destabilize for large advection speeds.

We note that the group velocity diverges as $b_*\searrow 1$. The ``unphysical'' large group velocities are irrelevant since states with $c_\mathrm{g}>1$ are unstable; see \cite[\S 3.2.2.b]{Coullet2000} and \cite{rs} for a discussion of this phenomenon in the context of wave trains

In particular, for low levels of water flow, vegetation patterns are unstable and break up into disorganized localized states. Linear transport by group velocities is illustrated in Figure \ref{f:1a}, the sideband instability is shown in  Figure \ref{f:d} and Figure \ref{f:1r}. 

\begin{Remark}[Geometry and $c_\mathrm{g}$]\label{r:geom}
 One can more generally compute $c_\mathrm{g}$ from simple properties of the phase portrait of the kinetics in Figure \ref{f:lin}. For more general kinetics, $b_t=f(b,w),w_t=-f(b,w)$, suppose that there is a curve of equilibria $f((\gamma_b,\gamma_w)(\tau))=0$. One then readily computes, expanding the neutral eigenvalue in the linearization, 
 \[
  c_\mathrm{g}=-c\frac{\gamma_w'}{\gamma_b'+\gamma_w'}=-c\frac{\gamma'\cdot\left(\begin{array}{c}0\\1\end{array}\right)}{ \gamma'\cdot\left(\begin{array}{c}1\\1\end{array}\right)},
 \]
which is positive for vectors $\gamma'$ in the sector of width $\pi/4$ bordered by $\left(\begin{array}{c}1\\0\end{array}\right) $ and $\left(\begin{array}{c}1\\-1\end{array}\right)$ (and in the opposite sector, where however $d_\mathrm{eff}<0$). More directly, equilibria with uphill transport are characterized by null clines $w=h(b)$ with $0>h'>-1$. In words, the somewhat counterintuitive uphill migration of vegetation stems, in this sense of group velocities, from a somewhat counterintuitive inverse (but not too strongly so) equilibrium relation between nutrient supply and biomass: equilibrium states with higher biomass concentration correspond to smaller (free) nutrient concentrations. 
\end{Remark}

\begin{figure}[b]
 \includegraphics[width=0.3\textwidth]{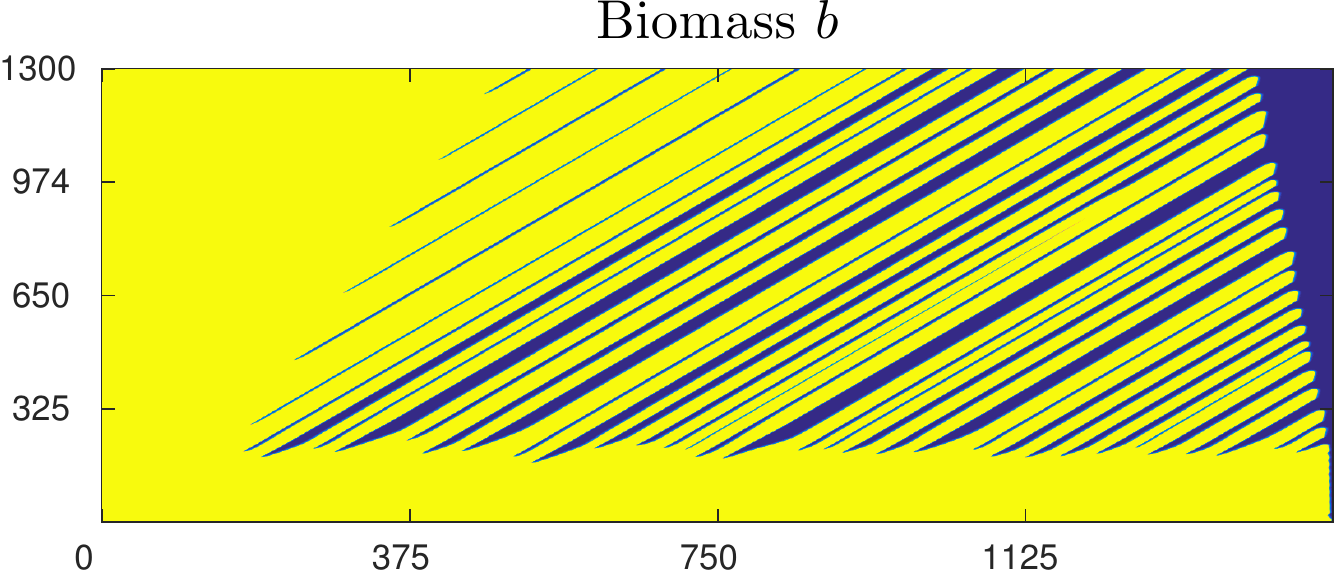}\hfill
\includegraphics[width=0.34\textwidth]{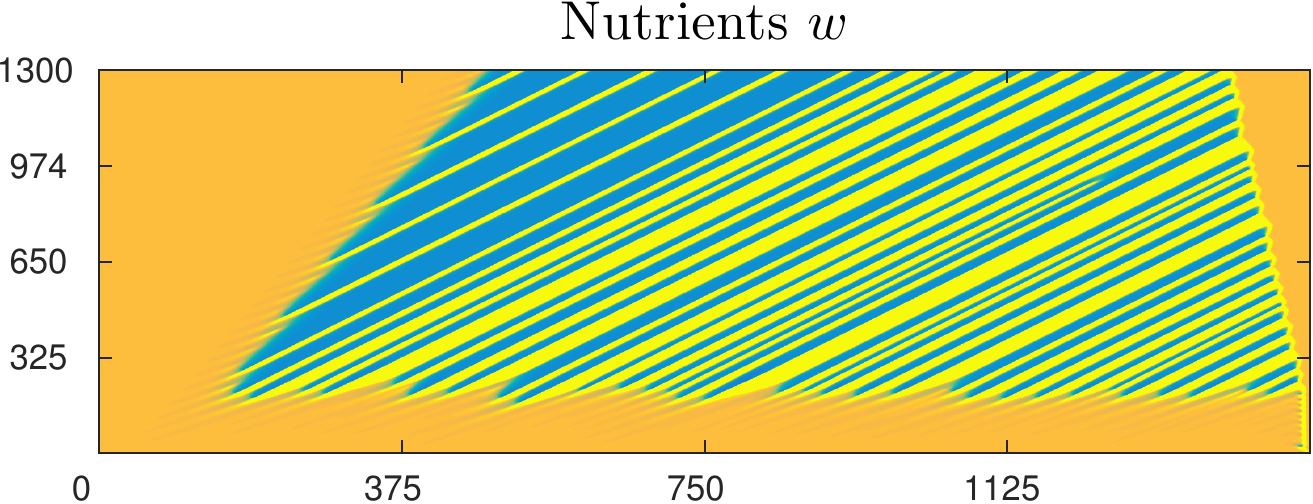}\hfill
 \includegraphics[width=0.3\textwidth]{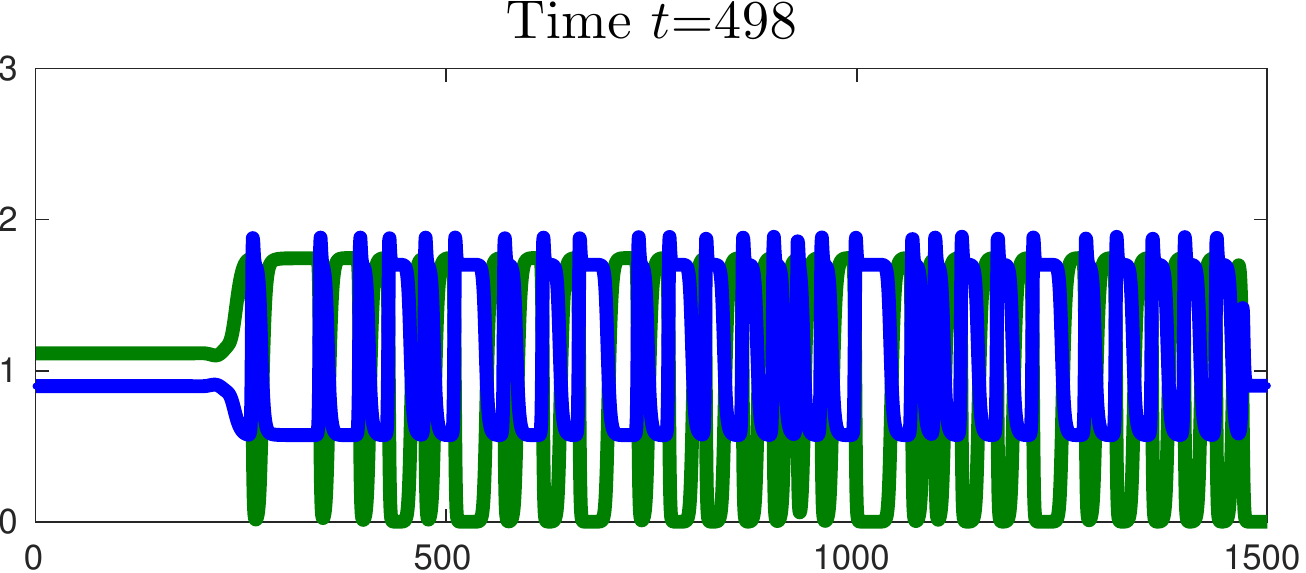}\\
 \includegraphics[width=0.3\textwidth]{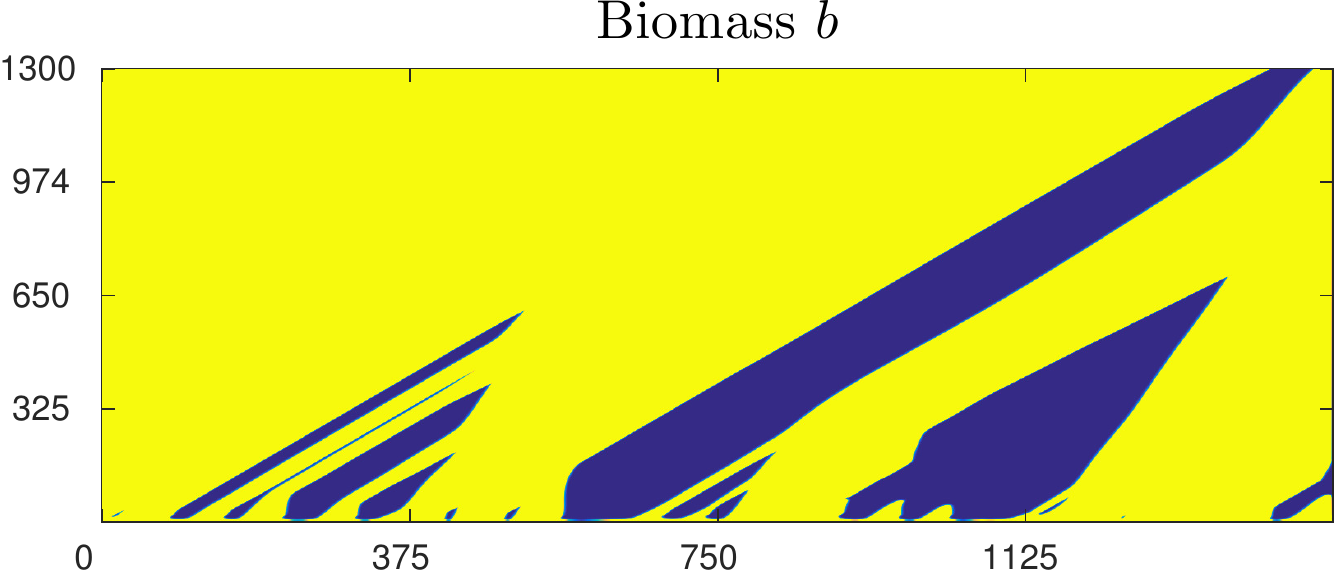}\hfill
\includegraphics[width=0.34\textwidth]{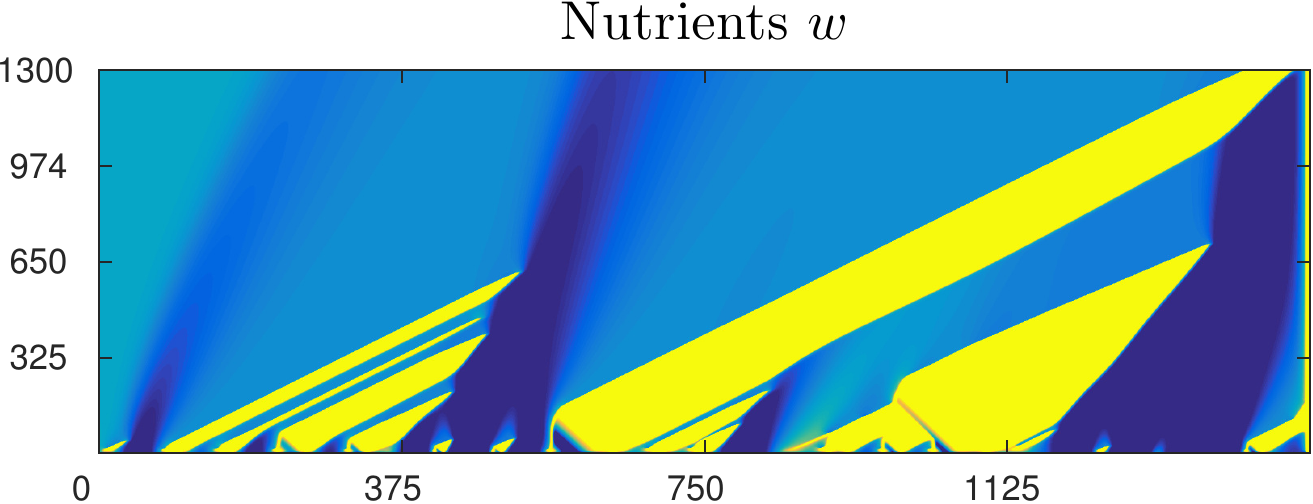}\hfill
 \includegraphics[width=0.3\textwidth]{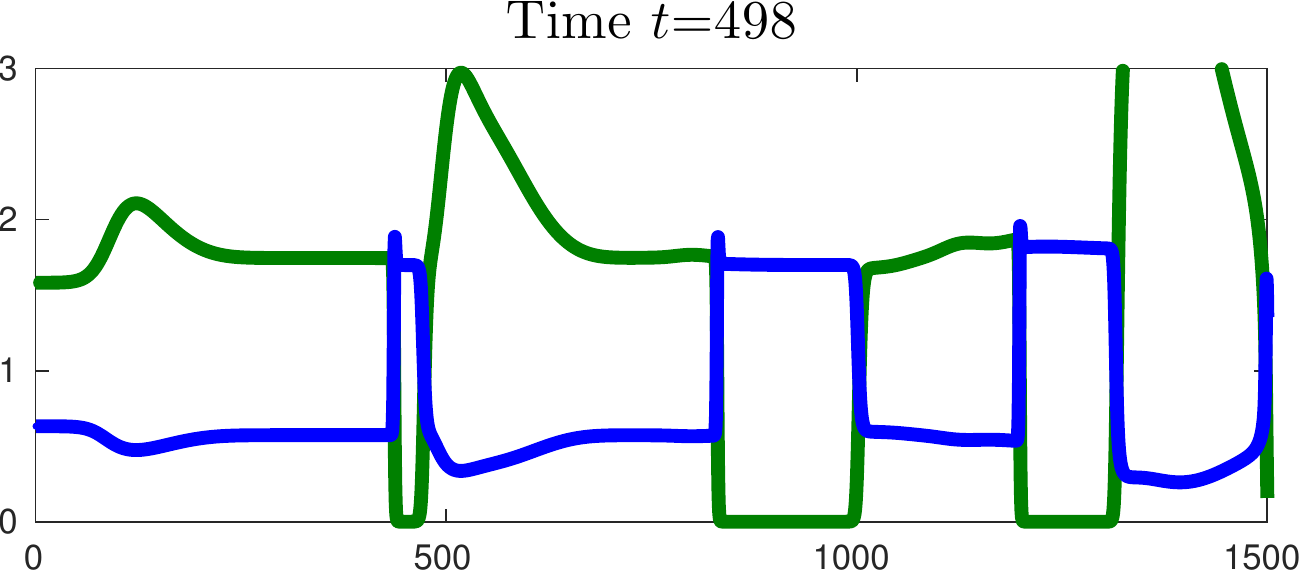} 
 \caption{Evolution of sideband instabilities (top) and typical patterns evolving from random initial data (bottom).}\label{f:1r}
\end{figure}
\section{The conservation law formalism}\label{s:cl}
The analysis so far can be reviewed from the point of view of viscous scalar conservation laws; see Remark \ref{r:cl}. Our goal now is to  explain how this analogy can be constructed formally and, to some extent rigorously. 

We first briefly recall features of scalar conservation laws that are mimicked in our system. We then show how to derive conservation law dynamics using a modulation approach, locally, and a formal reduction, globally. 

\paragraph{Scalar conservation laws.}

Dynamics of scalar conservation laws of the form \eqref{e:scl} can be most easily understood in terms of small disturbances of a constant state, $u(x)=u_0+\eps v_0(x)$, where $v(x)$ satisfies at leading order the convection-diffusion equation $v_t=d(u_0)v_{xx}-f'(u_0)v_x$. The localized initial condition $v_0(x)$ experiences linear transport with characteristic speed $f'(u_0)$ and diffusive decay with effective diffusivity. One can show that this convective-diffusive decay is preserved when taking into account higher-order terms in $\eps$. 

Beyond small amplitude, localized data, one describes dynamics in terms of Riemann problems, with initial conditions $u(x)=u_-$ for $x<0$, $u(x)=u_+$ for $u>0$. The values $u_\pm$ are propagated with characteristic speed $f'(u_\pm)$. In the case $f(u_-)>f(u_+)$, one typically observes a (unique) Lax shock, propagating with speed given through the Rankine-Hugoniot condition 
\[
 s=\frac{f(u_+)-f(u_-)}{u_+-u_-},
\]
which can be readily obtained by substituting a traveling-wave ansatz $u=u(x-st)$  into \eqref{e:scl} and integrating once. In the case $f(u_+)<f(u_-)$, solutions typically evolve into rarefaction waves, explicit in the case $d=0$ through an ansatz $u=u(x/t)$. 

Also, from the traveling-wave equation, one finds that all traveling waves are necessarily of (possibly degenerate) Lax type, that is, characteristics $x=f'(u_\pm)t$ enter the shock location $x=st$. Equivalently, we always have $f'(u_-)\geq s\geq f'(u_+)$ for viscous shocks, simply by inspecting stability properties of equilibria in the scalar traveling-wave ODE, and assuming well-posedness $d>0$. 

More generally, beyond the scalar setting here, one would classify shocks according to the number of characteristics entering and leaving the shock, respectively. In particular, in the scalar setting, shocks where characteristics leave the shock line, for example $f'(u_-)<s$,  would be \emph{undercompressive}. 

We refer to Figure \ref{f:char} for an illustration of these ideas, in the present context of vegetation patterns.

\paragraph{Local reduction --- modulation of vegetation densities.}
We describe solutions in a vicinity of a vegetation state, with an ansatz 
\begin{equation}\label{e:lwl}
b=b_0+\eps b_1(\eps^2 t, \eps (x-s t)),\qquad w=\frac{1}{b_0}-\eps \frac{1}{b_0^2}b_1(\eps^2 t, \eps (x-s t))+\eps^2 w_2(\eps^2 t, \eps (x-s t))+\rmO(\eps^3), 
\end{equation}
with error terms depending on the scaled variables $\tau=\eps^2 t$ and $\xi=\eps(x-st)$. Substituting into \eqref{e:rds} and collecting terms at order $\rmO(\eps^2)$, we find
\[
 -sb_{1,\xi}=b_0^2\left(w_2-\frac{b_1^2}{b_0^3}\right),\qquad \frac{s}{b_0^2}b_{1,\xi}=-\frac{c}{b_0^2}b_{1,\xi}-b_0^2\left(w_2-\frac{b_1^2}{b_0^3}\right).
\]
Adding the two equations gives $(s-c_\mathrm{g})b_{1,\xi}=0$, with $c_\mathrm{g}=\frac{c}{b_0^2-1}$, as expected. We also collect 
\begin{equation}\label{e:w2}
 w_2=-\frac{s}{b_0^2}b_{1,\xi}+\frac{b_1^2}{b_0^3}.
\end{equation}
At order $\rmO(\eps^3)$, after adding the equations for $b$ and $w$, using the expression for $w_2$, and using the ensuing equation for $w_{2,\xi}$, we find
\[
 \left(1-\frac{1}{b_0^2}\right)b_{1,\tau}=\left(1-\frac{s(c+s)}{b_0^2}\right)b_{1,\xi\xi}+(c+s)\left(\frac{b_1^2}{b_0^3}\right)_\xi.
\]
After some short algebra, we see that this equation is equivalent to Burgers' equation
\begin{equation}\label{e:burg}
 b_{1,\tau}=d_\mathrm{eff} b_{1,\xi\xi}-c_\mathrm{g}'b_1 b_{1,\xi}.
\end{equation}
Derivations of this type are ubiquitous in the literature. One would hope that error terms can be rigorously controlled using methods as in \cite{dsss}. Figure \ref{f:1c} exemplifies the presence of both Lax shocks and rarefaction waves in our system.

\begin{figure}
 \includegraphics[width=0.3\textwidth]{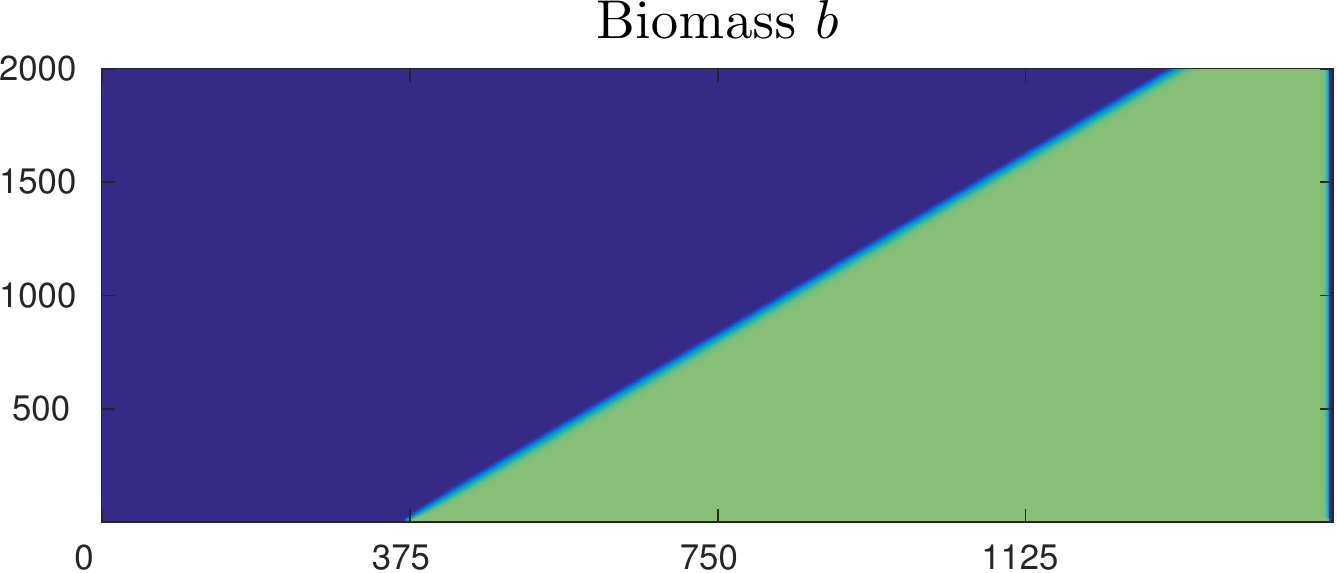}\hfill
\includegraphics[width=0.34\textwidth]{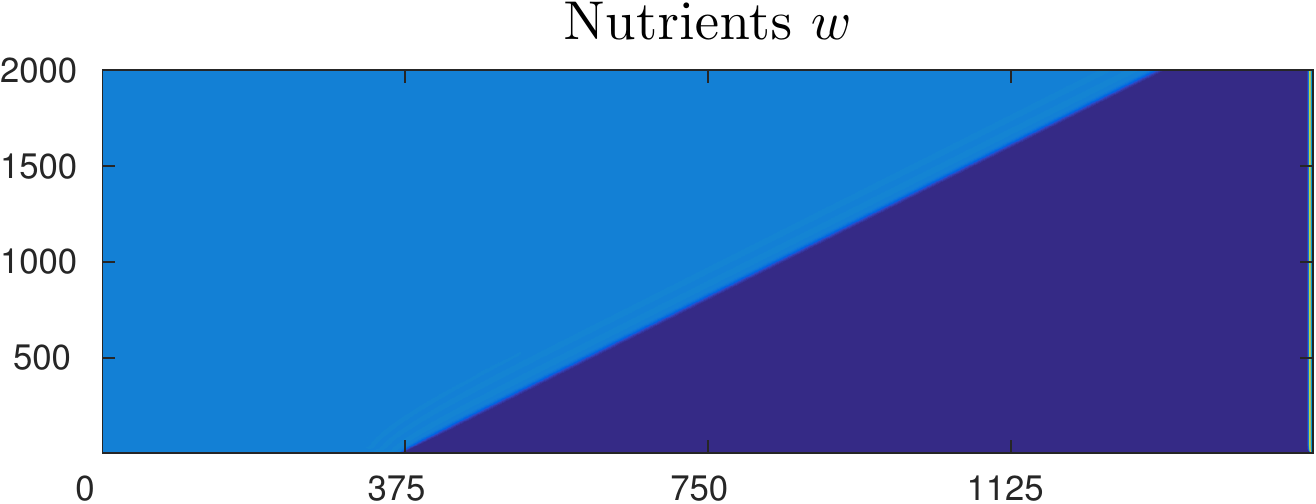}\hfill
 \includegraphics[width=0.3\textwidth]{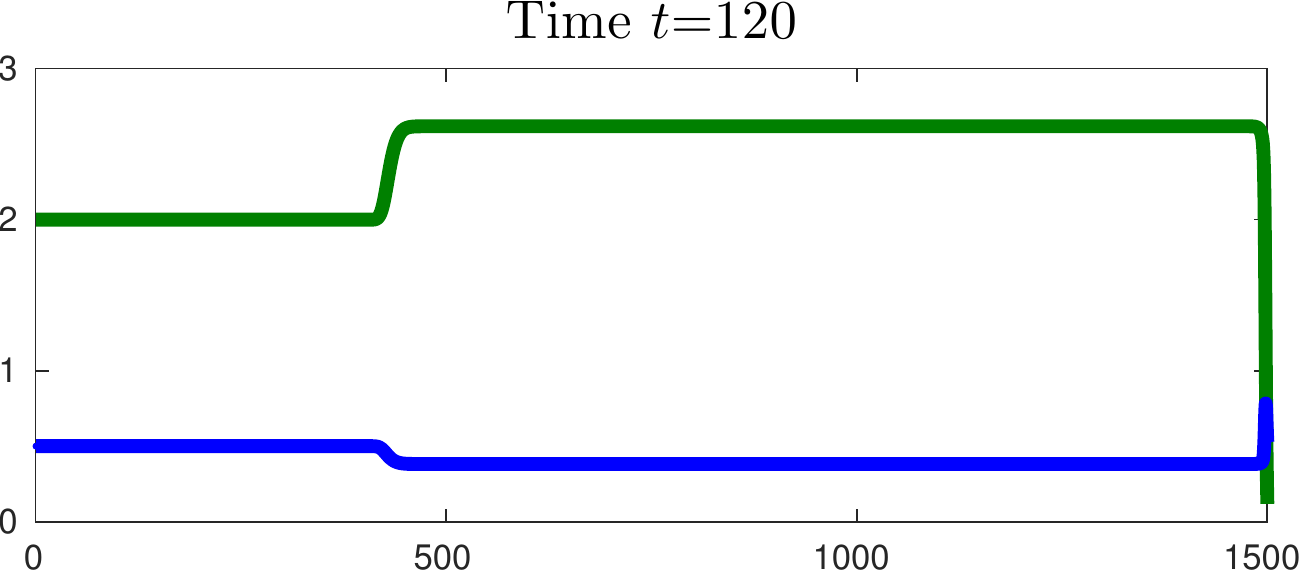}\\
 \includegraphics[width=0.3\textwidth]{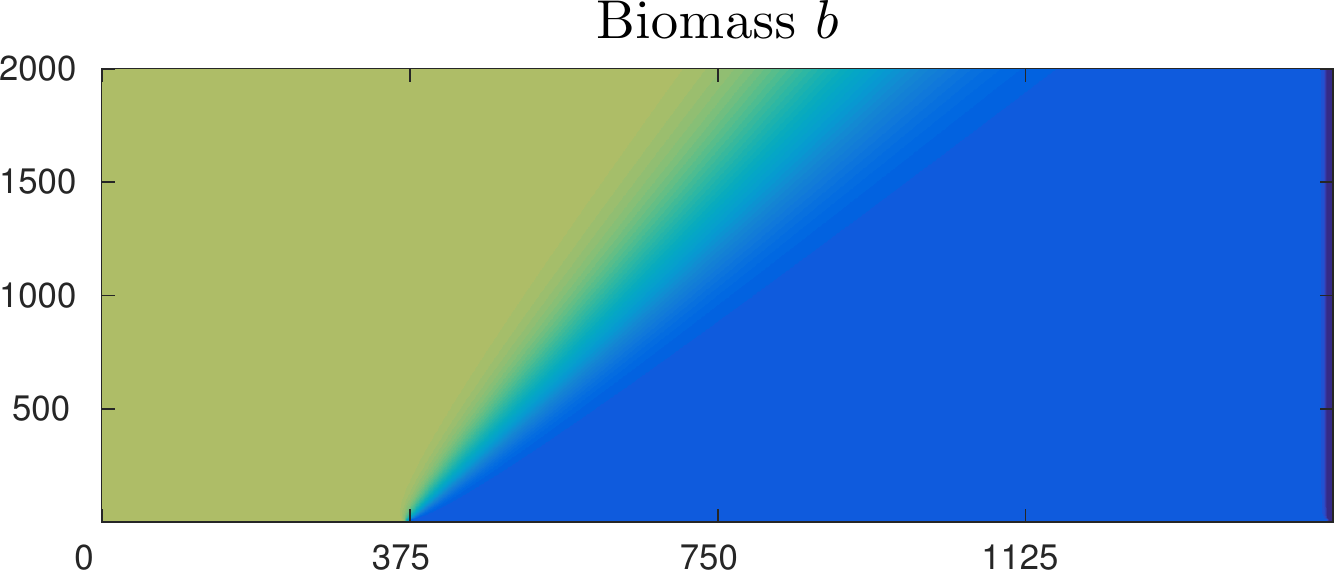}\hfill
\includegraphics[width=0.34\textwidth]{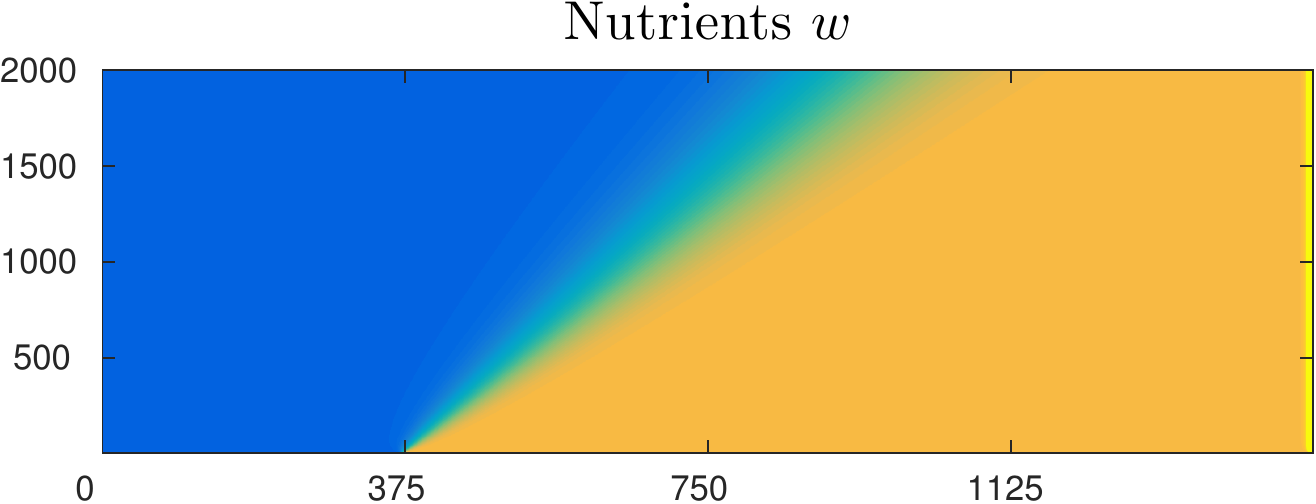}\hfill
 \includegraphics[width=0.3\textwidth]{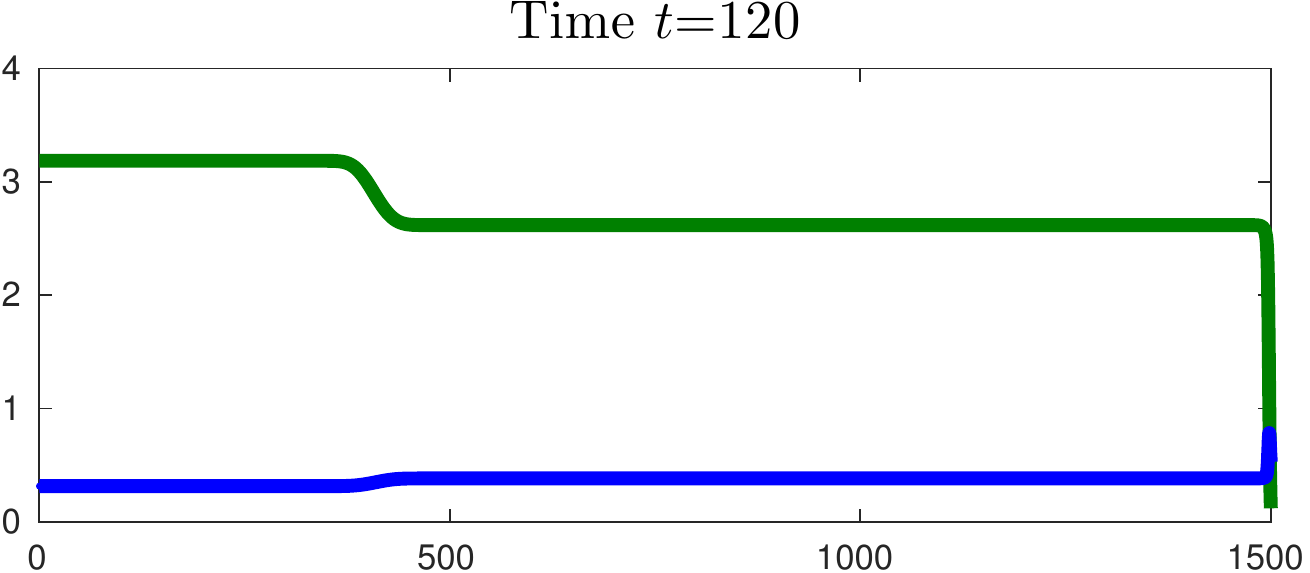}
 \caption{Plots of Lax shock (top)  and rarefaction wave (bottom) within a vegetation zone.  }\label{f:1c}
\end{figure}

\paragraph{Global reduction --- transport and viscosity.}
We notice that \eqref{e:burg} could be derived heuristically from simple linear information, the effective viscosity $d_\mathrm{eff}$ and the linear transport $c_\mathrm{g}(b)$. Noticing that the quantity $z=b+w$ solves a conservation law $z^+_t=F_x$, we could postulate the form
\begin{equation}\label{e:la}
 z^+_t=(d_\mathrm{eff}(b)z^+_x)_x-c_\mathrm{g}(b)z^+_x,\qquad b+1/b=z^+,
\end{equation}
which, among the general forms of viscous scalar conservation laws is determined by properties of the linearization at constants $b\equiv b_0$. 
While it is not clear how one would describe how \eqref{e:la} approximates 
\begin{equation}\label{e:csl}
b_t+w_t=b_{xx}+cw_x,
\end{equation}
we would like to pursue the idea of a global conservation law somewhat further. We could change variables to $z^\pm=b\pm w$, and assume that $z_-=\sqrt{(z^+)^2-4}$, at least for the stable branch. Substituting into \eqref{e:csl} gives, after a short computation,
\[
 z^+_t=(d(z^+))_{xx}-f(z^+)_x,
\]
with $f'(z^+)=c/(b^2-1)=c_\mathrm{g}$. The effective diffusivity $d'=b^2/(b^2-1)\neq d_\mathrm{eff}$ is incorrect at this order of approximation, and an equation 
\begin{equation}\label{e:redcl}
 z^+_t=(d_\mathrm{eff}(z^+_x))_x-f(z^+)_x,
\end{equation}
with $d'_\mathrm{eff}=d'-c^2b^2/(b^2-1)^3$, thus matching \eqref{e:cg}, appears to be more accurate.
%
%
%
%
The Rankine-Hugoniot condition for the speed of heteroclinic traveling-wave solutions is readily obtained from the Ansatz $z^+=z^+(x-st)$, $z^+\to z^+_\pm$ for $\xi\to\pm\infty$, 
\[
 s=\frac{f(z^+_+)-f(z^+_-)}{z^+_+-z^+_-},
\]
and corresponds to \eqref{e:twcl}. Lax shocks now correspond to shocks with $f'(z^+_-)>s>f'(z^+_+)$. This condition is generally satisfied for heteroclinic orbits connecting the two nontrivial equilibria, arising for instance as small heteroclinics in the saddle-node bifurcation described in the next section; see Figure \ref{f:1c}.  It is however never satisfied for upper and lower edges, since $s>0$ and $f'(z^+_-)<0$. In the traveling-wave analysis of the next section, one can check that $f'(z^+_+)>s$ for the heteroclinics connecting to the largest equilibrium, $f'(z^+_+)=0$ when connecting to the saddle-node equilibrium, and $f'(z^+_+)<s$ when connecting to the middle equilibrium. In this respect, upper and lower edges are \emph{undercompressive shocks}. That is, they are connecting states with characteristics emanating from the shock on the left, and with characteristics either emanating from or with the same speed as the shock on the right; see also Figure \ref{f:char}, below, for a schematic illustration. 

As previously mentioned, it maybe quite difficult to describe precisely how long-time dynamics of the reaction-advection-diffusion system \eqref{e:rd} are approximated by a scalar conservation law. We notice however that the concepts discussed here all translate immediately from the scalar conservation law to \eqref{e:rd}, replacing characteristic speeds by group velocities, and the existence of shocks with the nonlinear analysis from Section \ref{s:tw}.

\section{A phase diagram for traveling waves}\label{s:tw}

We look for traveling waves $\tilde{b}(x-st),\tilde{w}(x-st)$, which gives
\begin{equation}
 -s\tilde{b}'=\tilde{b}'' +  \tilde{b}^2 \tilde{w} -  \tilde{b},\qquad \qquad
 -s\tilde{w}'=c\tilde{w}' -\tilde{b}^2 \tilde{w} + \tilde{b}.\label{e:rdtw}
\end{equation}
Adding the equations and integrating once gives 
\begin{equation}
 s\tilde{b}+(s+c)\tilde{w}+\tilde{b}'=(s+c)\scflux,\label{e:twcl}
\end{equation}
for some constant of integration $\scflux$. Inspecting the original equation \eqref{e:rd} in a co-moving frame $\xi=x-st$, we see that $(s+c)\scflux$ is the flux for the total mass,
\[
 (b+w)_t=(s+c)\scflux_\xi.
\]
The fact that $\scflux$ is conserved for the traveling-wave equation can then be viewed as a Rankine-Hugoniot type constraint on the speed of fronts, given asymptotic states where $\tilde{b}'=0$; see Section \ref{s:int} for more details on the conservation law point of view\,\footnote{Compare also with \cite[(1.8)]{ssdefect} where such conditions were derived in reaction-diffusion systems when the underlying conserved quantity is the phase of an oscillation rather than an explicit variable.}. Note that, for solutions that limit on a vegetation-free state $b_\infty=0$, $\scflux=w_\infty$ encodes the amount of nutrient at infinity.
Solving for $\tilde{w}$ and substituting into the first equation in \eqref{e:rdtw} gives 
\[
 \tilde{b}''+s\tilde{b}'+\tilde{b}^2(\scflux - \frac{s}{s+c}\tilde{b}-\frac{1}{s+c}\tilde{b}')-\tilde{b}=0.
\]
Upon scaling 
\begin{equation}\label{e:sc}
 b=\tilde{b}(s+c)^{-1/2}, \  \theta= \scflux(s+c)^{1/2},
\end{equation}
we find
\begin{align}
 b'&=v-sb,\notag\\
 v'&=-b^2(\theta-v)+b.\label{e:tw}
\end{align}
Our main theoretical results characterize bounded solutions to this planar ODE \eqref{e:tw}. We first list theoretical results, Section \ref{s:th}, and then show numerically computed bifurcation diagrams, Section \ref{s:num}. We conclude the section with interpretations of our results, Section \ref{s:int}.

\subsection{Theoretical existence results}\label{s:th}

We first collect some elementary facts on equilibria and their bifurcations. We then state global results on the existence of heteroclinic orbits, as well as local results on the existence of periodic and homoclinic orbits. 

\paragraph{Steady-state bifurcations.} For each $s>0$, \eqref{e:tw} possesses either a unique equilibrium $b=v=0$, or two additional equilibria (counted with multiplicity), 
\[
 b_\pm=\frac{\theta}{2s}\pm\sqrt{\frac{\theta^2}{4s^2}-\frac{1}{s}}, \qquad v_\pm=sb_\pm.
\]
The discriminant vanishes, and the equilibria disappear in a saddle-node bifurcation, when 
\begin{flalign}\label{e:sn}
 \text{(SN)} & \qquad\qquad  \qquad \theta^2=4s, \qquad \text{or} \qquad s=c_\mathrm{g},&
\end{flalign}
where the latter equality can be readily found by undoing the scaling $\tilde{b}=\sqrt{\frac{s+c}{s}}$ and using \eqref{e:cg}; see also \cite{ssdefect} for the relation between zero group velocities and saddle-node bifurcations in traveling-wave equations. Before the saddle-node, $s<\theta^2/4$, there are two equilibria in addition to and compatible with a trivial equilibrium $b=0$, $w=w_+$, with respect to the conserved quantity $\theta$. Past the saddle-node, for large speeds $s>\theta^2/4$, vegetation states are not compatible with a trivial equilibrium. 

The saddle-node curve passes through a Bogdanov-Takens point, with algebraically double zero eigenvalue, at 
\begin{flalign}
 \text{(BT)} & \qquad\qquad\qquad\qquad  \qquad \theta=2, \qquad s=1.&
\end{flalign}
For $\theta<2$, the saddle-node equilibrium $b=1$ possesses a negative, stable eigenvalue in addition to the zero eigenvalue, for $\theta>2$ the additional eigenvalue is positive.

In the PDE linearization, $\theta=2$, $s=1$ corresponds to $b=1$ and $\tilde{b}=\sqrt{1+c}$, that is, the Bogdanov-Takens point coincides with the onset of the sideband instability at zero group velocity; see \cite[\S 4.2]{hs} and \cite[\S 6]{hs2} for a similar scenario in a different context. One can verify that the sideband instability, $\tilde{b}<\sqrt{1+c}$  and the homogeneous instability, $\tilde{b}<1$, only occur for $b_-$. 

From the Bogdanov-Takens curve emerges a branch of Hopf bifurcation curves for $b_-$, 
\begin{flalign}\label{e:hopf}
 \text{(Hopf)} & \qquad\qquad\qquad\qquad  \qquad \theta=\frac{s^2+1}{\sqrt{s}}, \qquad s<1.&
\end{flalign}

One can verify that both Bogdanov-Takens and Hopf bifurcation are generically unfolded. Across the  Hopf bifurcation, stability equilibria destabilize with increasing $\theta$ and $s$ fixed. The branching is towards decreasing $\theta$, hence subcritical as a bifurcation in $\theta$, for $\theta>4\cdot3^{-3/4}=1.7547653\ldots$ and $s<3^{-1/2}=0.5773502\ldots$. It is supercritical otherwise. We computed the cubic Hopf coefficient $a$ in the basis $(\sqrt(1-s^2),0), (s,1)$ using computer algebra and found 
\[
 \Re a(s)= \frac{1 - 3 s^2}{8 s^2 - 8}.
\]
The location of the degenerate Hopf bifurcation coincides well with the end point of the continuation of the periodic saddle-node in \textsc{AUTO07p}. For large $\theta$, $s\sim 0$, the real part of the cubic Hopf coefficient converges to $1/8$. In the limit $s\to 1$, the cubic coefficient diverges to $-\infty$ as expected near the codimension-two point. 

The Bogdanov-Takens point, $b=v=1$ at $\theta=2$, $s=1$, is generically unfolded. Quadratic terms are, in the notation of \cite[Thm 8.4]{kuzi}, 
    $a_{20} = 2$, $b_{20} = 2$, and $b_{11} = 0$, such that the non-degeneracy conditions $a_{20} + b_{11} \neq 0$ and $b_{20} \neq 0$ hold. 
One also readily verifies the versal unfolding in the parameters $\theta$ and $s$, based on the non-degeneracy of the saddle-node and the derivatives in trace and determinant.

\paragraph{Heteroclinic orbits --- lower and upper edges of vegetation bands.}

We now state our main existence results on heteroclinic orbits in the traveling-wave equation \eqref{e:tw}. Throughout, we denote by $0<w_-\leq w_+$ the three equilibria of \eqref{e:tw}. 

\begin{Theorem}[Upper edge]\label{t:u}
 There exists a unique, continuous curve $\{s_\mathrm{u}(\theta), \theta\in (0,\infty)\}$, such that for parameter values on this curve there exists a heteroclinic orbit connecting $w=0$ to $w=w_+$, thus describing the upper edge of a vegetation zone. Moreover, 
 \begin{itemize}
  \item $s_\mathrm{u}$ is non-decreasing;
  \item $s_\mathrm{u}(\theta)\to 0$ for $\theta\to 0$;
  \item $s_\mathrm{u}(\theta)=(s_\infty+\rmo(1))\theta^{2/3}$ for $\theta\to\infty$;
  \item $s_\mathrm{u}(\theta)=\theta^2/4$ for $\theta$ sufficiently small.
 \end{itemize}
Moreover, the heteroclinic diverges, with $w_+\to\infty$ for $\theta\to 0$ or $\theta\to\infty$.
\end{Theorem}
We derive more precise asymptotics for the heteroclinics in the proof. Numerically, $s_\infty=0.9055$.
\begin{Theorem}[Lower edge]\label{t:l}
 There exists a unique, continuous curve $\{s_{\lindex}(\theta), \theta\in (\theta_0,\infty)\}$, such that for parameter values on this curve there exists a heteroclinic orbit connecting $w=0$ to $w=w_+$, thus describing the lower edge of a vegetation zone. Moreover, 
 \begin{itemize}
  \item $s_{\lindex}$ is non-decreasing;
  \item $s_{\lindex}(\theta)\to 0$ for $\theta\to \theta_0$;
  \item $s_{\lindex}(\theta)\to \infty$ for $\theta\to\infty$;
  \item $s_{\lindex}(\theta)=\theta^2/4$ for $\theta$ sufficiently large.
 \end{itemize}
Moreover, the heteroclinic diverges, with $w_+\to\infty$ for $\theta\to \theta_0$ or $\theta\to\infty$.
\end{Theorem}
Again, we find more precise asymptotics in the proof. Numerically, $\theta_0\sim 1.389\ldots$.

\begin{Corollary}[Maxwell point]\label{c:M}
 There exists a unique $(\theta_\mathrm{M},s_\mathrm{M})$, such that there exists a heteroclinic loop, that is, both upper and lower edge heteroclinics exist simultaneously.
\end{Corollary}
Numerically, $(\theta_\mathrm{M},s_\mathrm{M})=(0.7689,1.8465)$. In this regime, the trace of the linearization of $w=0$ is negative and the trace at $w=w_+$ is positive. The heteroclinic loop bifurcation is therefore a somewhat non-standard bifurcation discussed in \cite{shash,shil}. The two families of homoclinics emerging from the loop, asymptotic to each of the two equilibria in the loop respectively, bifurcate tangent to one of the heteroclinic orbit branches (in this case the lower edge homoclinic). We do however not attempt to analytically verify the generic unfolding conditions of the loop to rigorously establish the bifurcation diagram near the heteroclinic loop. However, we do note that the numerically computed bifurcation diagram agrees well with the corresponding diagram \cite[Fig. 13.7.20]{shil}.

\paragraph{Homoclinic and periodic orbits --- vegetation bands and gaps.}

Here we describe results on existence of homoclinic orbits. While it is conceivable to obtain somewhat more global existence and monotonicity results for the curves of existence in parameter space, mimicking the methods employed to prove Theorems \ref{t:u} and \ref{t:l}, instead we concentrate on end points of the numerically computed bifurcation curves. We refer to homoclinic orbits asymptotic to $w_+$ as vegetation gaps, and to homoclinic orbits asymptotic to $0$ as vegetation bands. We find
\begin{enumerate}
 \item a branch of vegetation gaps bifurcates from the BT point in the direction of decreasing $\theta$ and $s$;
 \item the branch of vegetation gaps terminates on the Maxwell point from Corollary \ref{c:M} and continues  from there as a vegetation band;
 \item the branch of vegetation bands has speed $s\sim \frac{6}{7}\theta^{-2}$ as $\theta\to\infty$, while the amplitude of the vegetation density $w_+$ diverges with $\theta$.
\end{enumerate}
The end points (i) and (ii) have been discussed above. We shall discuss (iii) in Section \ref{s:thinf}

Periodic vegetation patterns form a two-parameter family that happens to exist in a close vicinity of the Hopf curve. It is however typically not confined to the region between these two curves, as the direction of branching of periodic orbits from the Hopf curve and from the homoclinic curve change at some points. We discussed direction of branching from the Hopf curve, above. It would be interesting to obtain analytical existence results.

\subsection{The complete bifurcation diagram --- numerics and implications} \label{s:num}
We present the numerically computed bifurcation diagram and undo the scalings \eqref{e:sc}.

\paragraph{Numerically computed bifurcation diagrams.}
We computed periodic, homoclinic, and heteroclinic orbits using \textsc{AUTO07p} continuation software \cite{auto}; see Figure \ref{f:bif}. 
\begin{figure}
\includegraphics[width=0.32\textwidth]{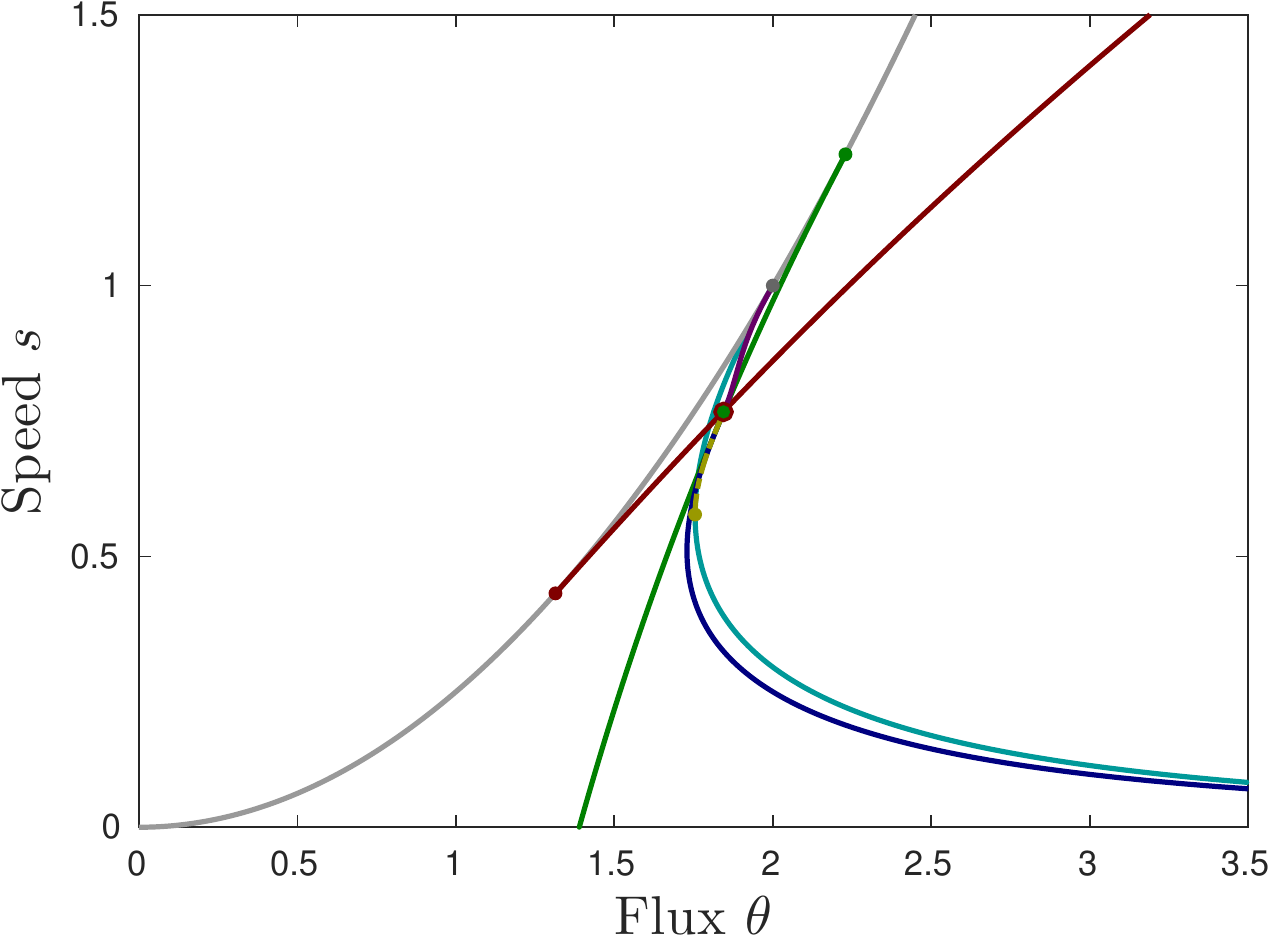}\hfill
\includegraphics[width=0.32\textwidth]{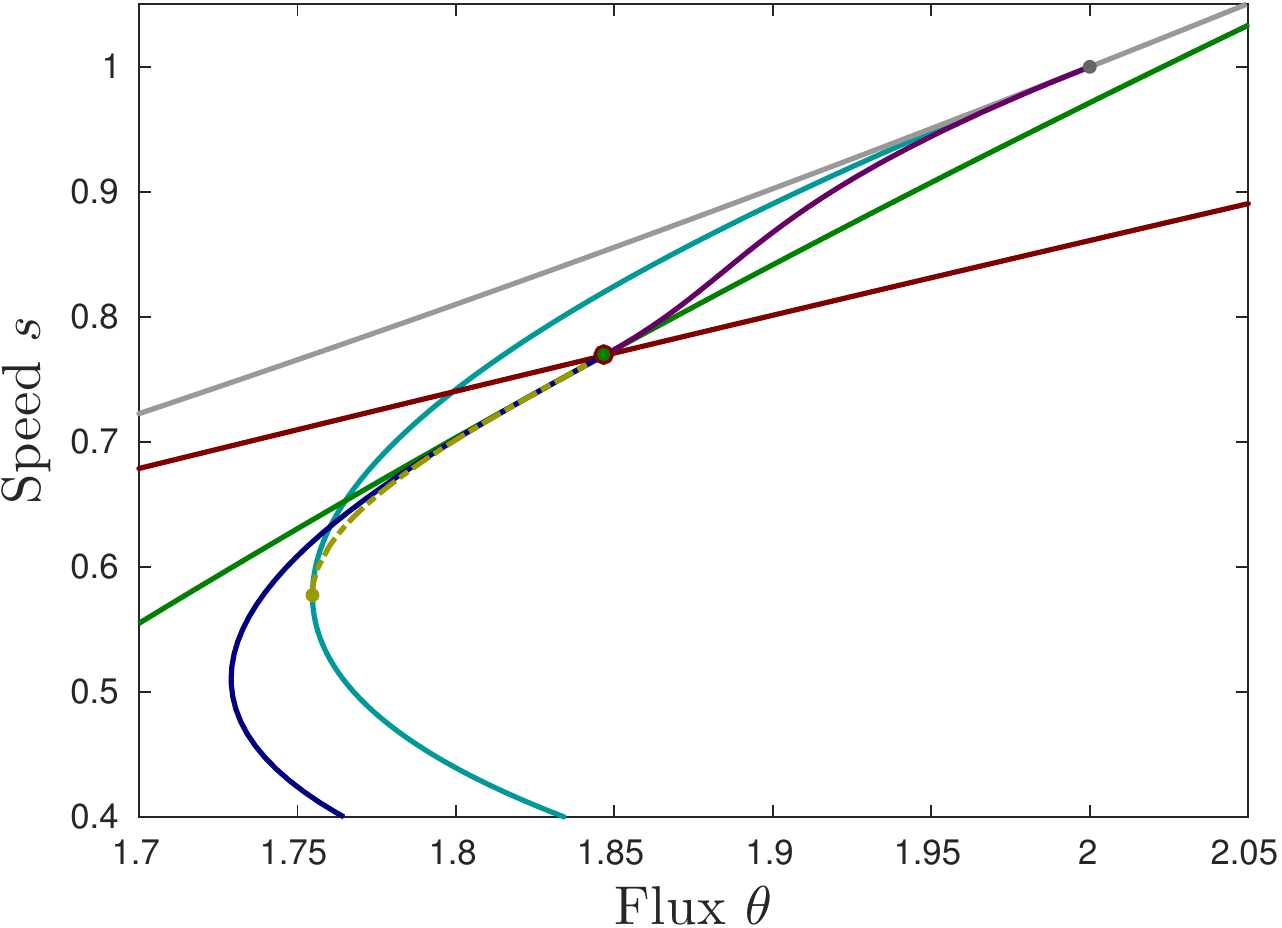}\hfill
\includegraphics[width=0.32\textwidth]{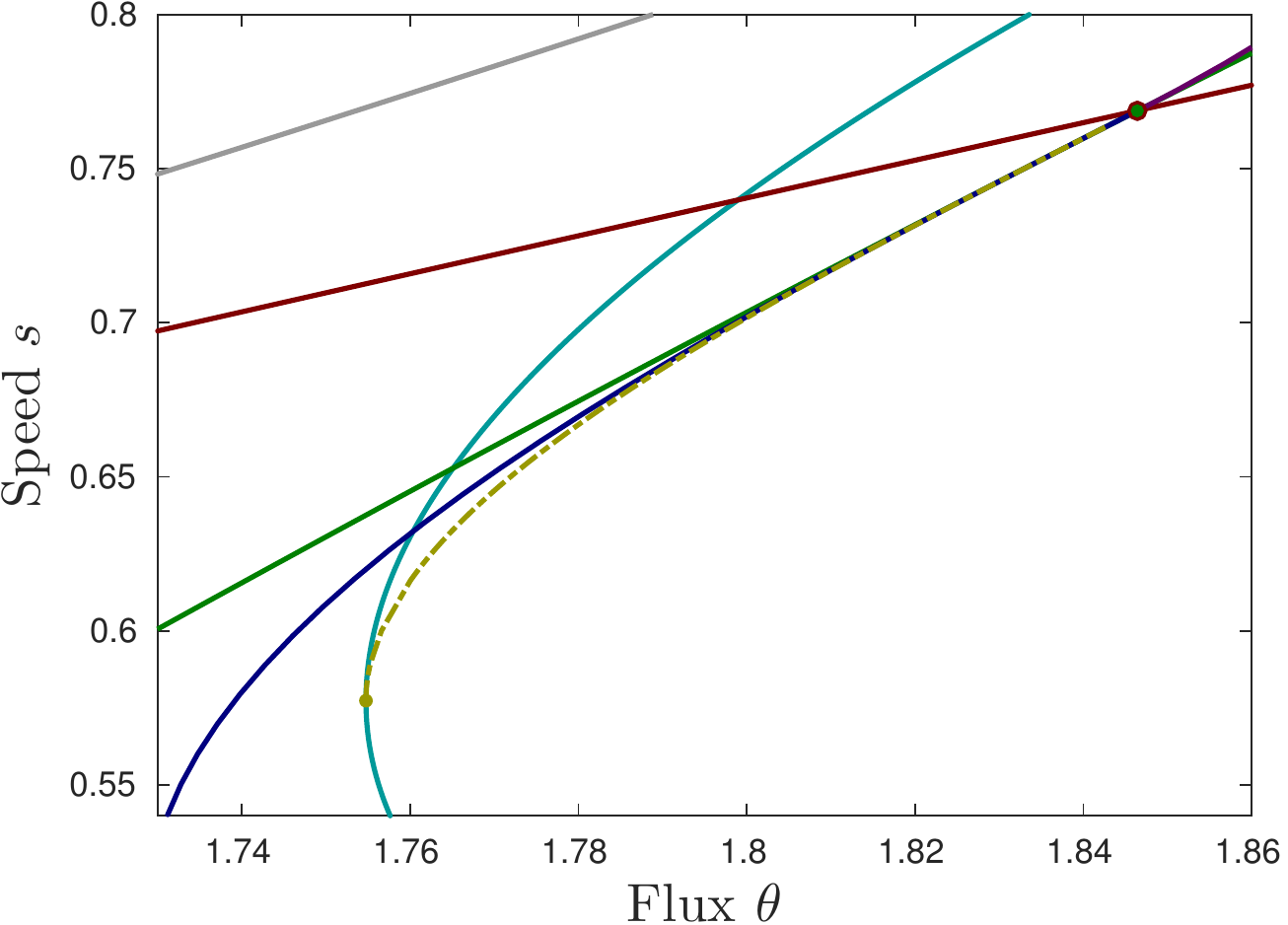}
\caption{Bifurcation diagrams with all bifurcations (left) and zooms (middle and right). Lines are saddle-node (gray) with BT point (dark gray $(2,1)$); upper edge (red) and  lower edge (green), both with endpoints on the saddle-node curve ($(1.314,0.432)$ and $(2.229,1.242)$, resp.), from where they continue on that curve; Hopf (light blue) and saddle-node of periodic orbits (yellow dashed), terminating on a degenerate Hopf point ($(1.755,0.577)$, yellow dot) and heteroclinic loop/Maxwell point ($(1.847,0.769)$, green/red circle); homoclinic bands (dark purple) and gaps (light purple), terminating at the BT point and the  heteroclinic loop. Periodic orbits exist in the area bounded by homoclinic, Hopf, and periodic saddle-node curves. The lower edge touches down at $\theta_0\sim 1.389$, the upper edge diverges $s\sim 0.9055\cdot  \theta^{2/3}$ for $\theta\to\infty$.} \label{f:bif}
\end{figure}
We find the theoretically predicted crossing of speeds of upper and lower edge heteroclinics, the Maxwell point, and the touch-down points, where heteroclinic orbits hit the saddle-node curve, that is, where the speed of the edge ceases to be larger than the group velocity of the vegetation state. Somewhat less intuitively, there exists a curve of single vegetation bands with very \emph{small} speed for large fluxes $\theta$. Interestingly, almost all homoclinic and periodic orbits are confined to a rather narrow zone in parameter space, bounded by homoclinic, Hopf, and a periodic saddle-node curve. The diagram shows a somewhat surprising complexity in the region including  BT point, Maxwell point, and a degenerate Hopf point. 

\paragraph{Undoing the scaling.}
The bifurcation diagram can be translated into the original variables, undoing the scaling \eqref{e:sc}, in a geometrically straightforward way, illustrated in Figure \ref{f:sc}, which also includes two sample-diagrams. 

\begin{figure}[h]
 \includegraphics[width=0.28\textwidth]{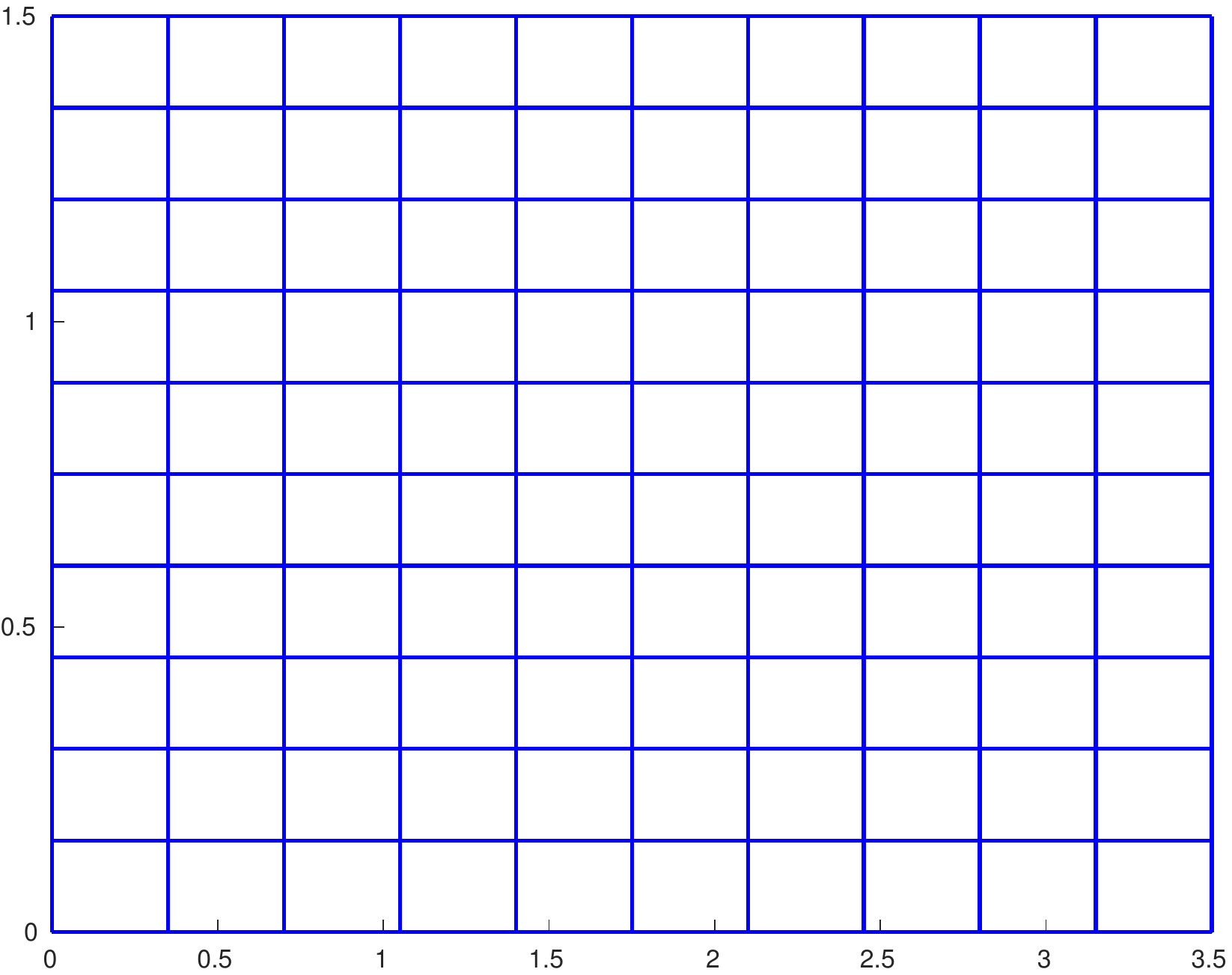}\hfill\raisebox{0.12\textwidth}{$\Longrightarrow$}\hfill\includegraphics[width=0.28\textwidth]{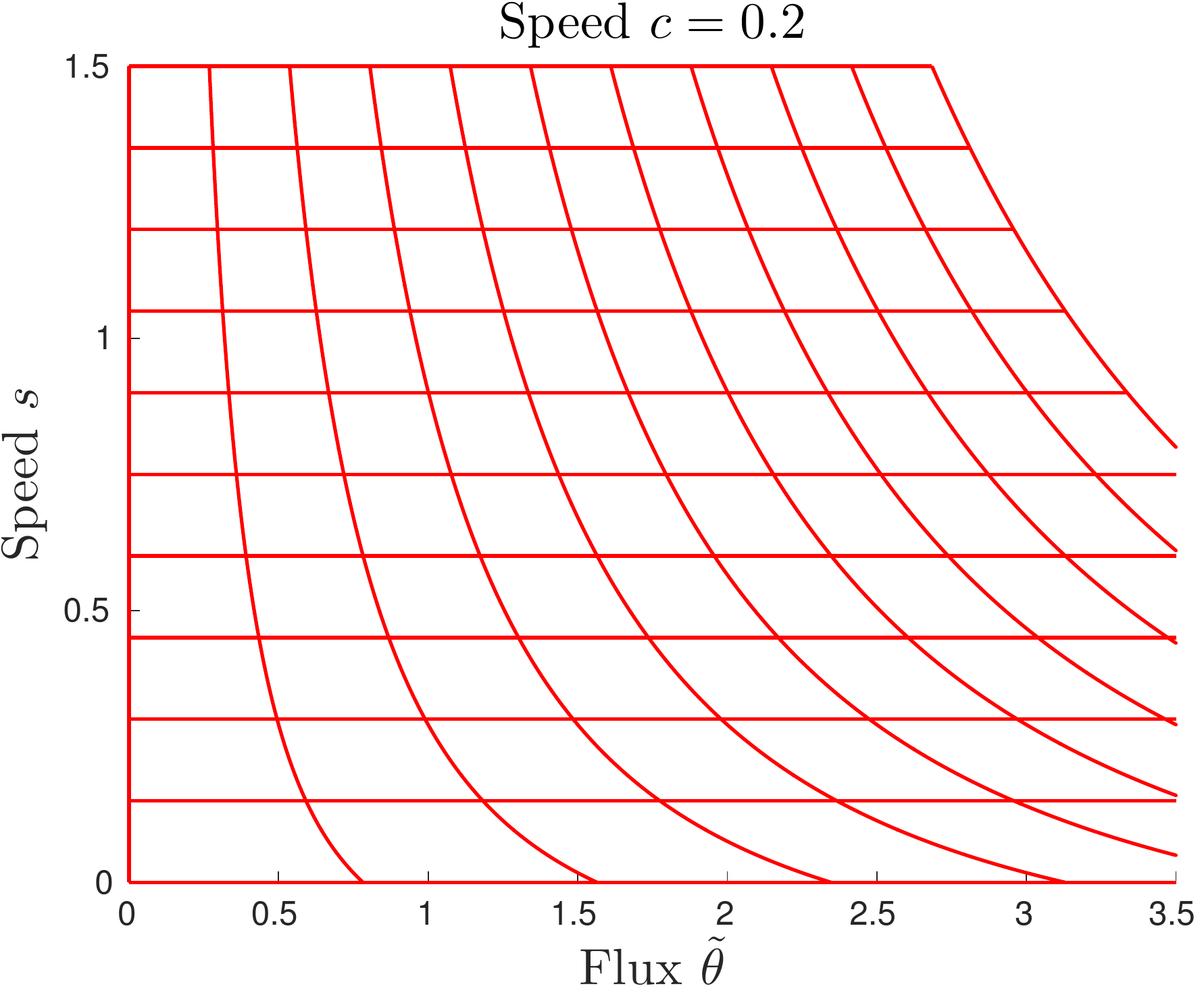}\hfill\includegraphics[width=0.28\textwidth]{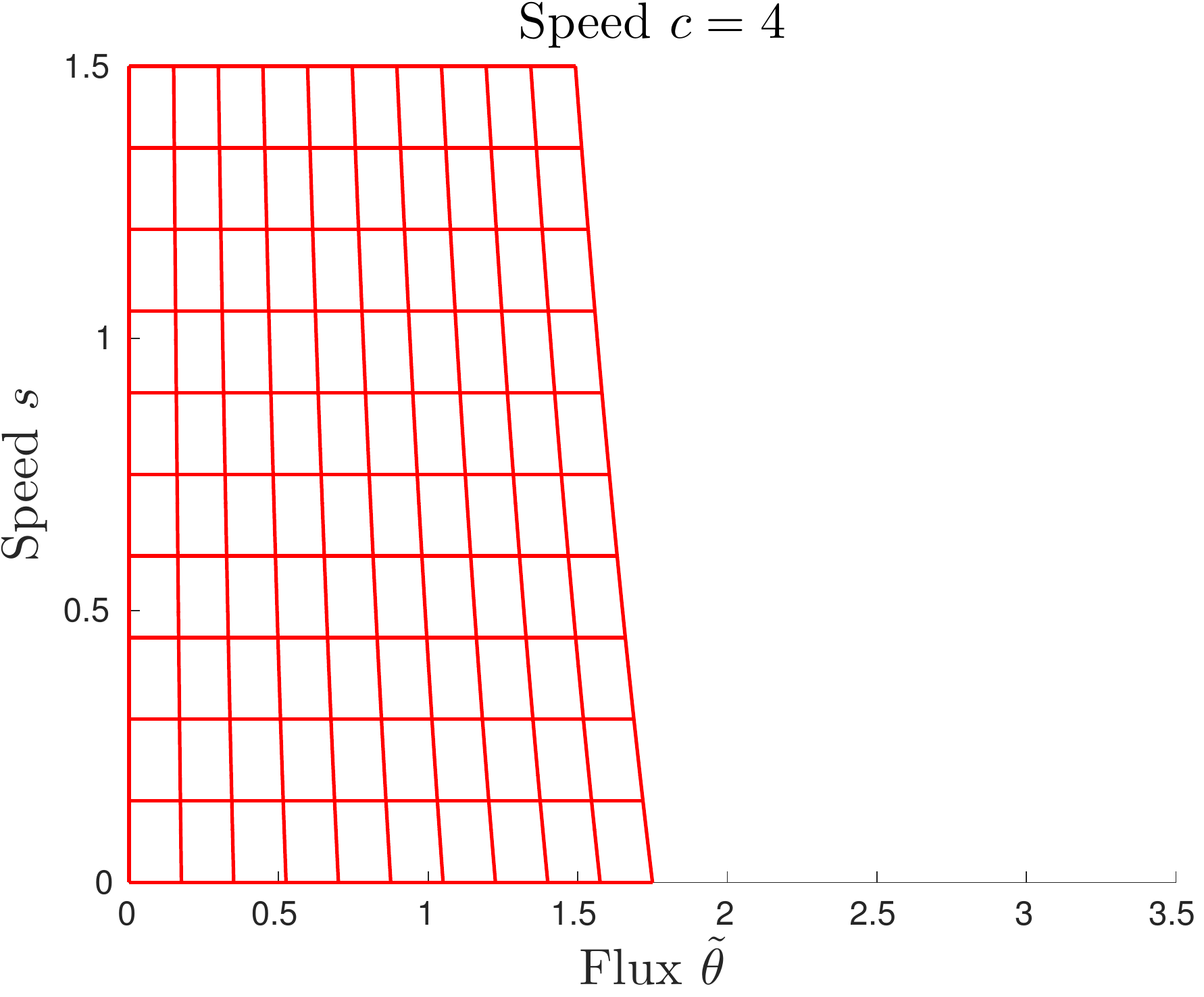}\\[0.2in]
 \includegraphics[width=0.28\textwidth]{bif_all}\hfill\raisebox{0.12\textwidth}{$\Longrightarrow$}\hfill\includegraphics[width=0.28\textwidth]{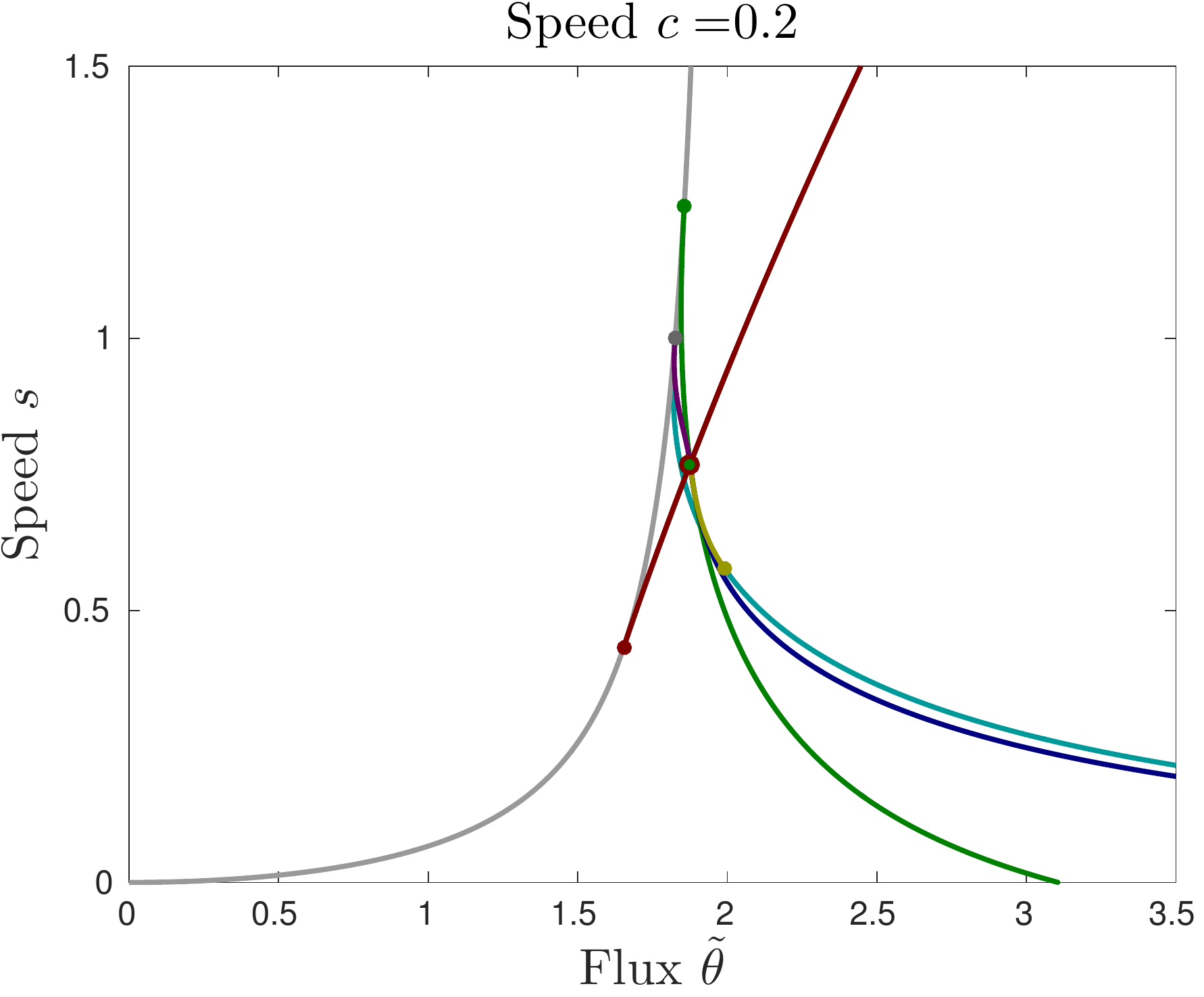}\hfill\includegraphics[width=0.28\textwidth]{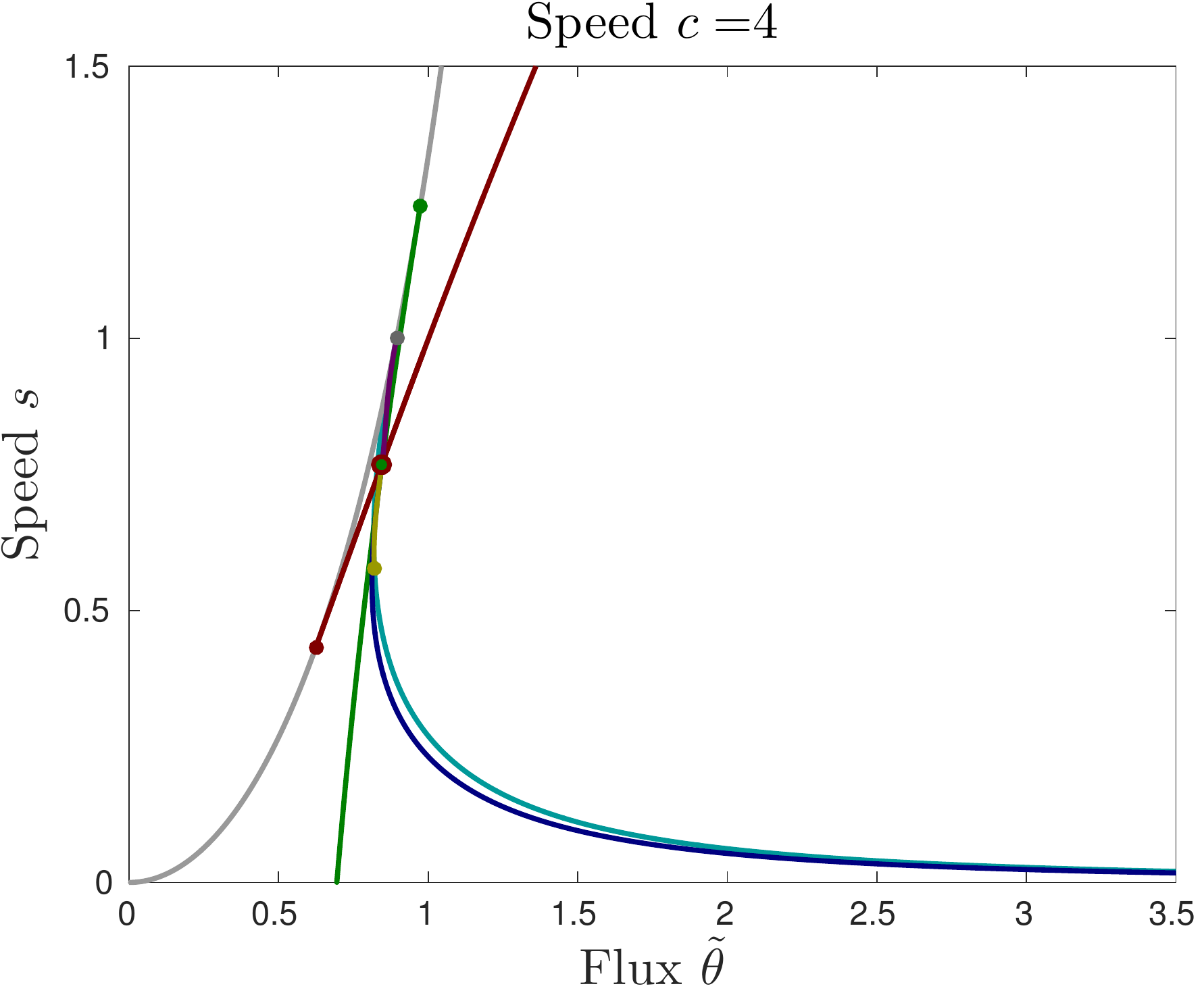} 
 \caption{Effect of scaling on a uniform grid in the $\theta-s$ plane  and on our bifurcation diagram from Figure \ref{f:bif} for speeds $c=.2$ and $c=4$.}\label{f:sc}
\end{figure}

For large advection speeds $c$, the effect of the scaling is simply a scaling of the flux by $\sqrt{c}$. For small speeds $\sqrt{c}$, however, the diagram is distorted, pushing the portion of the diagram with speed $s<1$ out to large fluxes $\scflux$. As a consequence, flux-speed relations for all periodic and homoclinic orbits will be monotone for large speeds. On the other hand, the flux-speed relation of the lower-edge heteroclinic changes monotonicity, exhibiting turning points for intermediate speeds $c\sim 0.5$. Finally, the saddle-node curve stays monotone for all speeds and flux-speed relation of the upper edge also remains monotone for arbitrarily small speeds. For very small speeds, existence of both upper and lower edges, as well as bands and periodic orbits, is guaranteed for $\scflux>2$, with speeds decreasing monotonically in $\scflux$ for all traveling waves but the upper edge; see Figure \ref{f:sc1}.

\begin{figure}
  \includegraphics[width=0.32\textwidth]{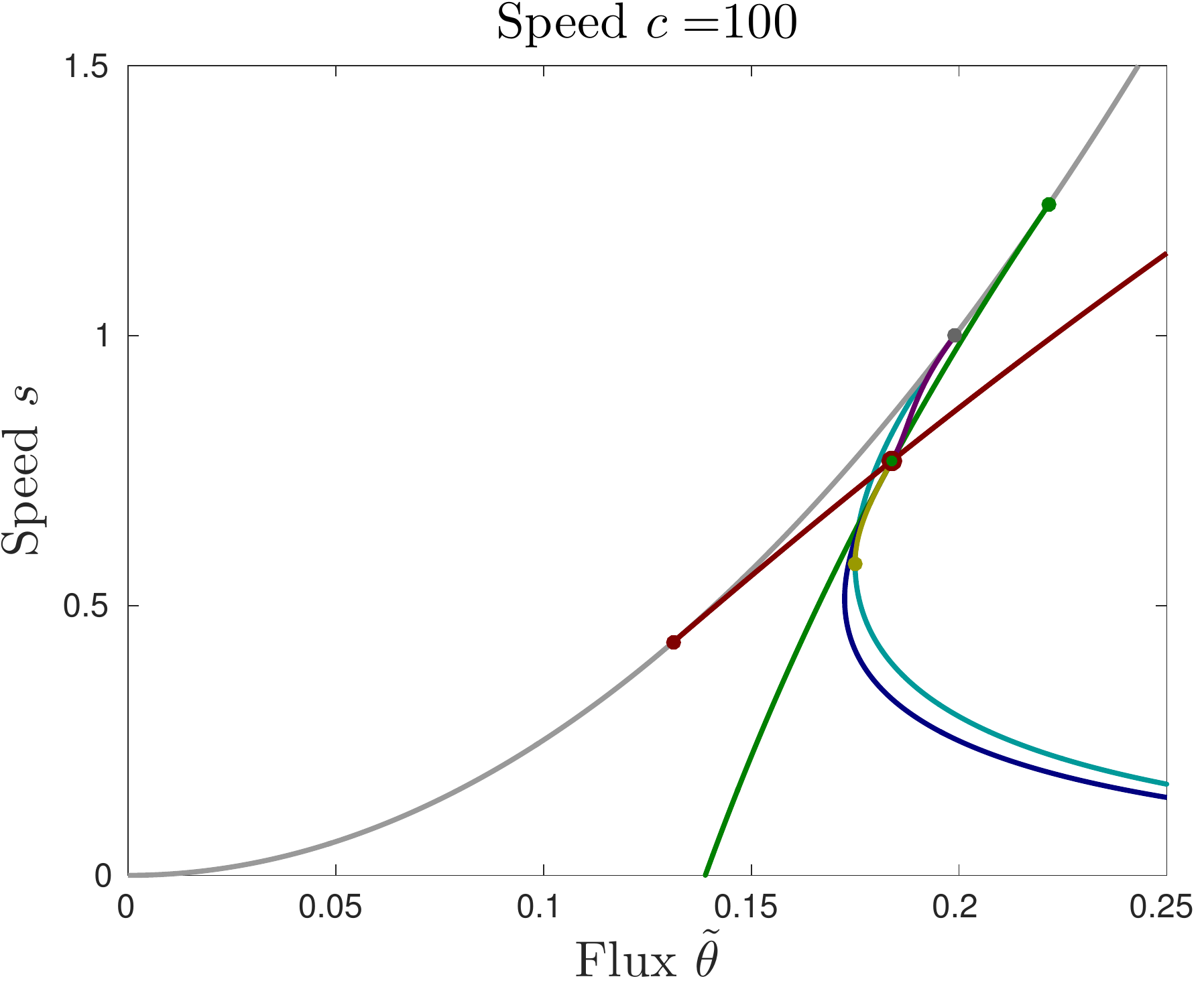}\hfill\includegraphics[width=0.32\textwidth]{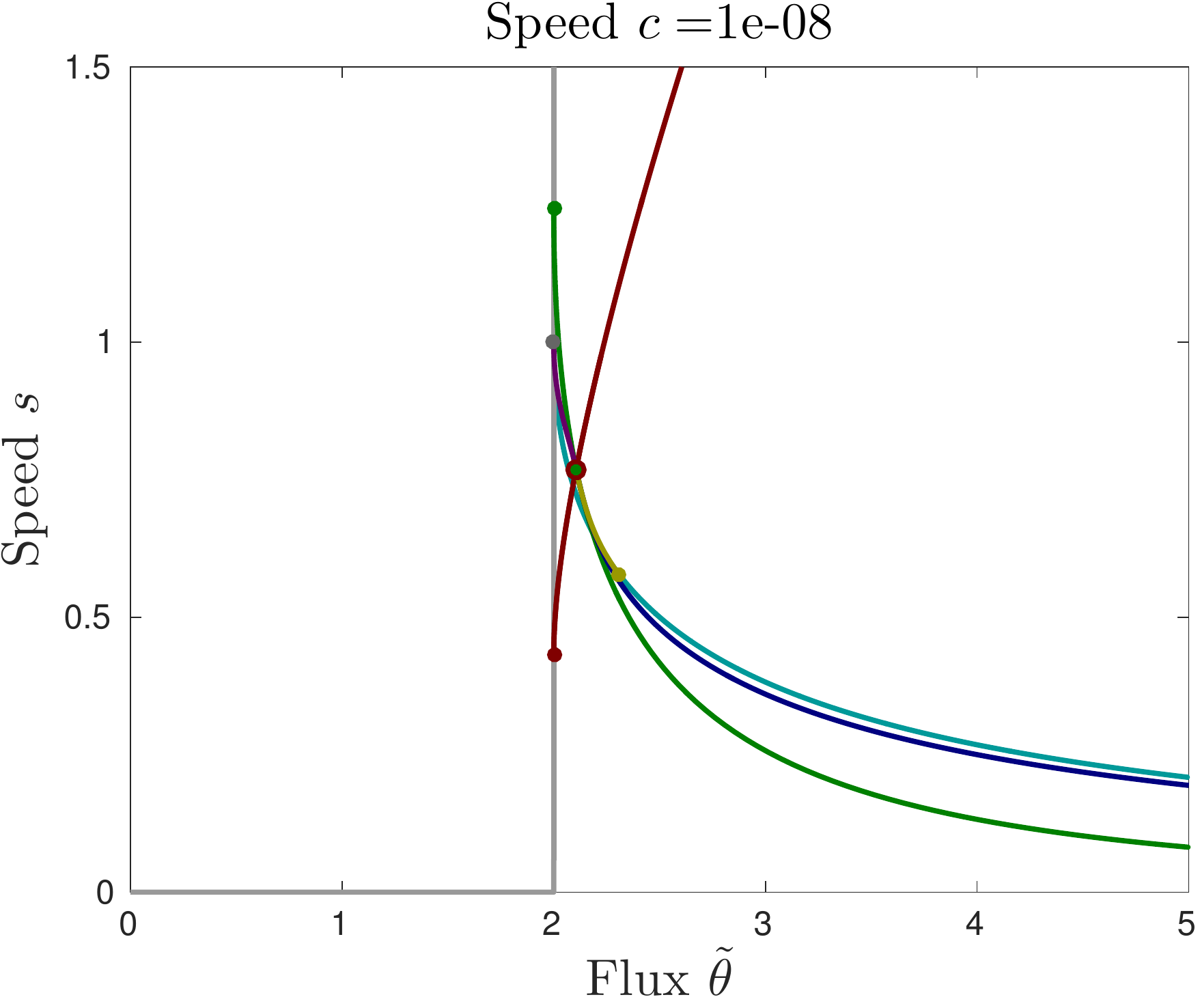}\hfill\includegraphics[width=0.32\textwidth]{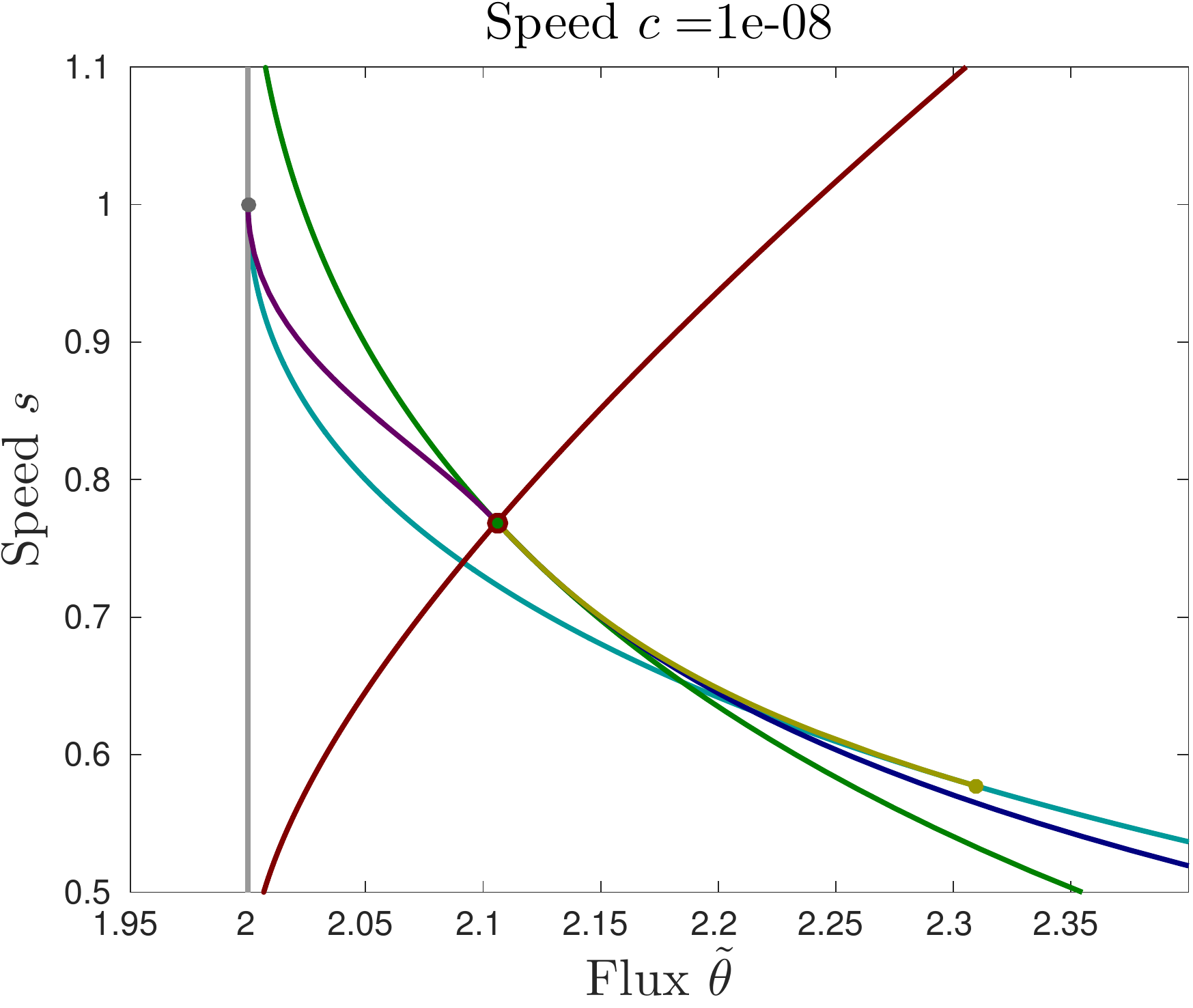} 
 \caption{Bifurcation diagrams, scaled according to large speed (left) and small speed (middle) with zoom (right); note the different scales on the axes.}\label{f:sc1}
\end{figure}

\subsection{Interpretation and implications}\label{s:int}

\paragraph{Upper and lower edge.}
We saw throughout the existence of upper and lower edges in direct simulations. Notice that, as apparent in Figure \ref{f:1b},  the profile of $b$ is monotone, apparent also from the proof, but $w$ is not necessarily monotone. In fact, $w$ is obtained from \eqref{e:twcl}, which includes the derivative of $b$ in addition to the monotone profile of $b$. Phenomenologically, the apparent peak of $w$ near the upper edge reflects an enhanced potential for growth in the leading edge of the front. 

We compared predicted speeds with speeds in direct simulations, with generally good agreement, suggesting in particular that the front solutions found here are stable as solutions of the PDE. For upper edges, we initiated a Riemann problem with prescribed concentrations $b=0,w=w_+$ at the right end and somewhat arbitrary values $b,w>0$ at the left end. We observed that the leading edge front selects the state in its wake in the following sense. Define $b_-^*=1/w_-^*$, the state in the wake selected by the front. Setting up Riemann problems with $b_-=1/w_->b_-^*$, the characteristic speed in the wake $c_\mathrm{g}(b_-)$ is smaller than the speed $s$ of the upper edge. One finds, as illustrated in Figure \ref{f:ue}, left panel, that the region between the upper edge $b\sim b_-^*$ and the wake $b\sim b_-$ is filled in by a rarefaction wave with approximately linear profile of $b$. If $b_-<1/b_-^*$ on the other hand, then the upper edge changes as a bound state between a Lax shock connecting $b_-$ and $b_-^*$ and the upper edge. These bound states are, in our traveling-wave problem, heteroclinic orbits connecting to the middle equilibrium. Reducing $b_-$ further, the left equilibrium undergoes a sideband instability and more complicated dynamics ensue; see Figure \ref{f:d}, left panel. When the heteroclinic is of saddle-node type, that is, when group velocities in the vegetation state equal speed of propagation, one notices distinctly slower spatio-temporal rates of convergence. We did not attempt to investigate those quantitatively. 

\begin{figure}[h]
\begin{minipage}{0.6\textwidth}
 \includegraphics[width=0.48\textwidth]{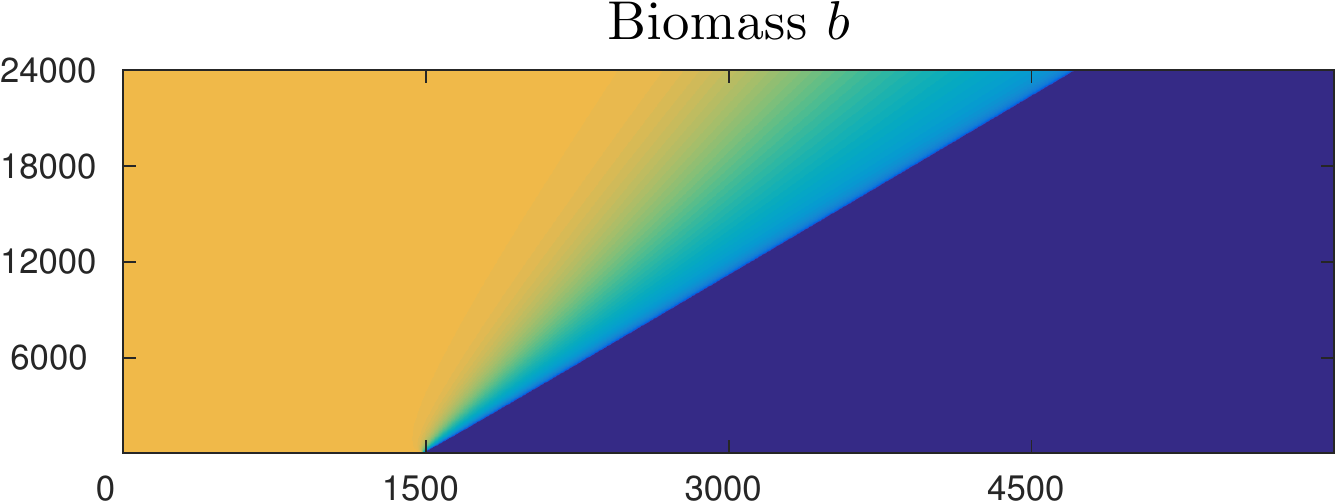}\hfill\includegraphics[width=0.48\textwidth]{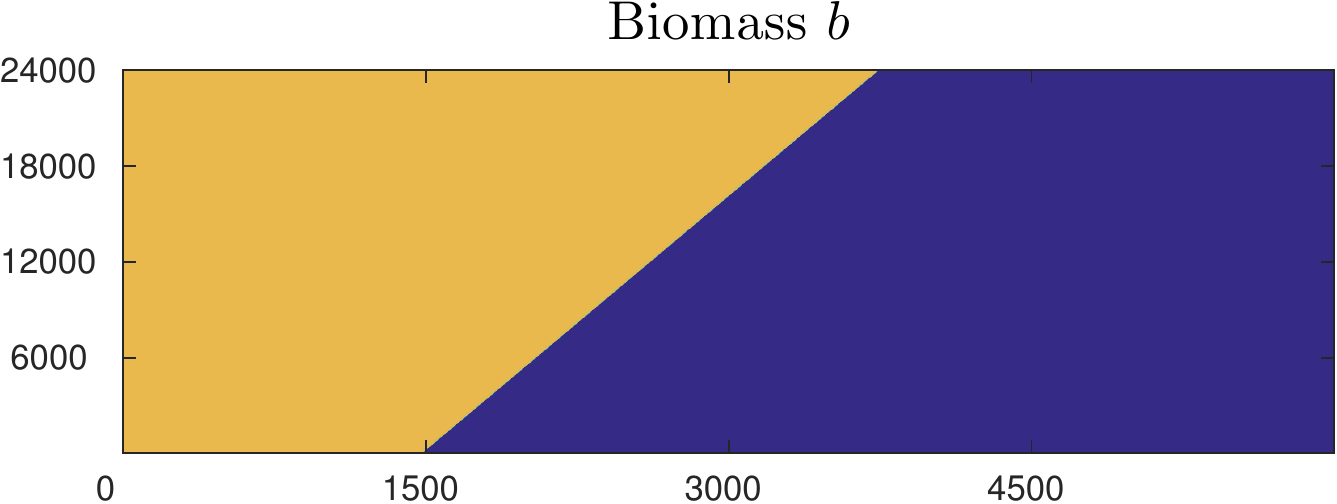}\\
 \includegraphics[width=0.48\textwidth]{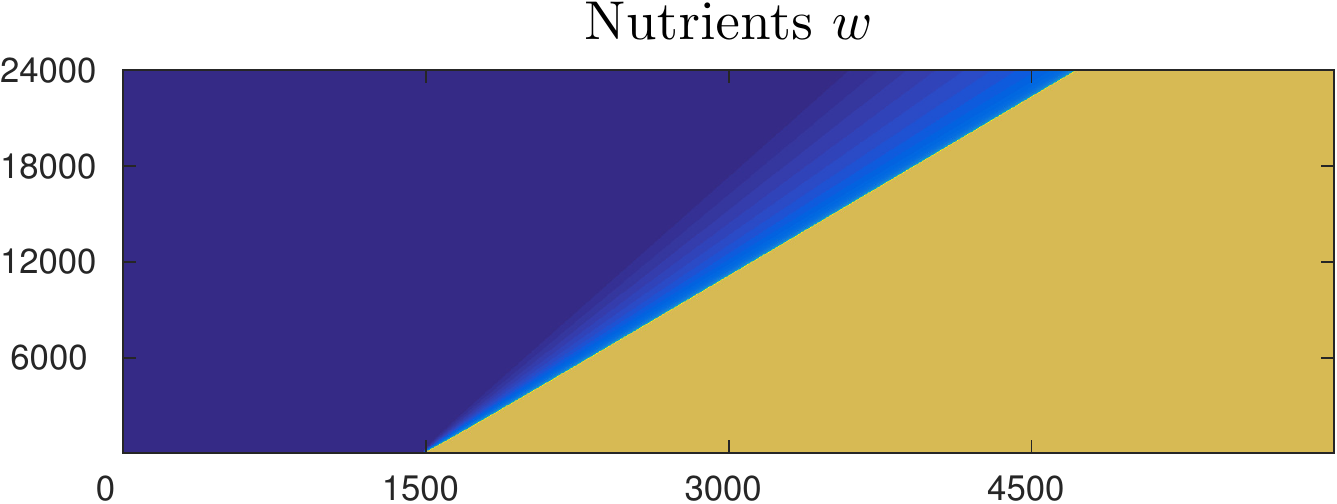}\hfill\includegraphics[width=0.48\textwidth]{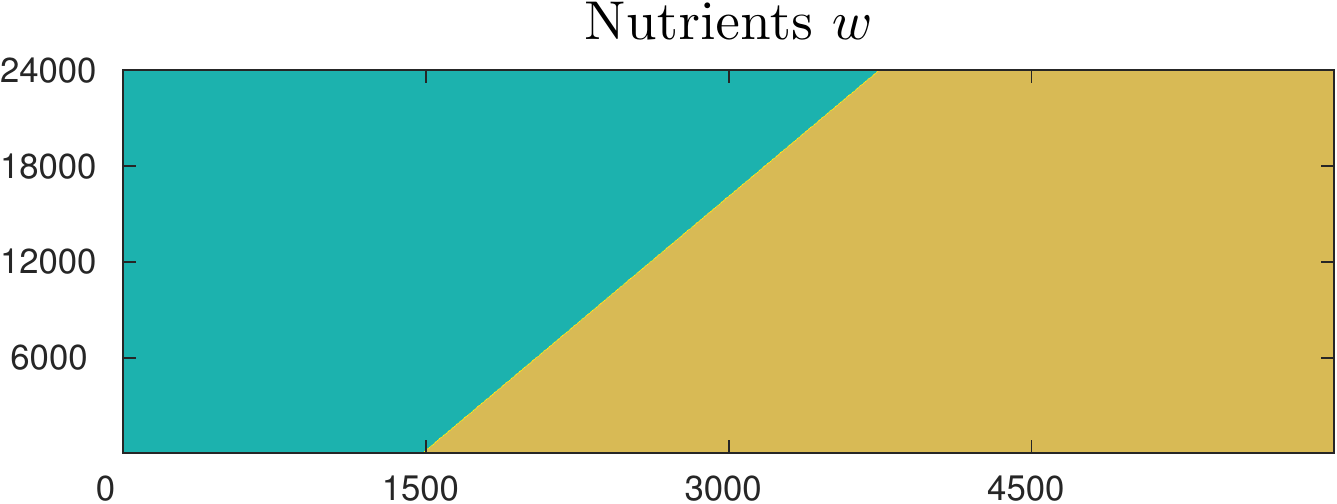}\\
 \includegraphics[width=0.48\textwidth]{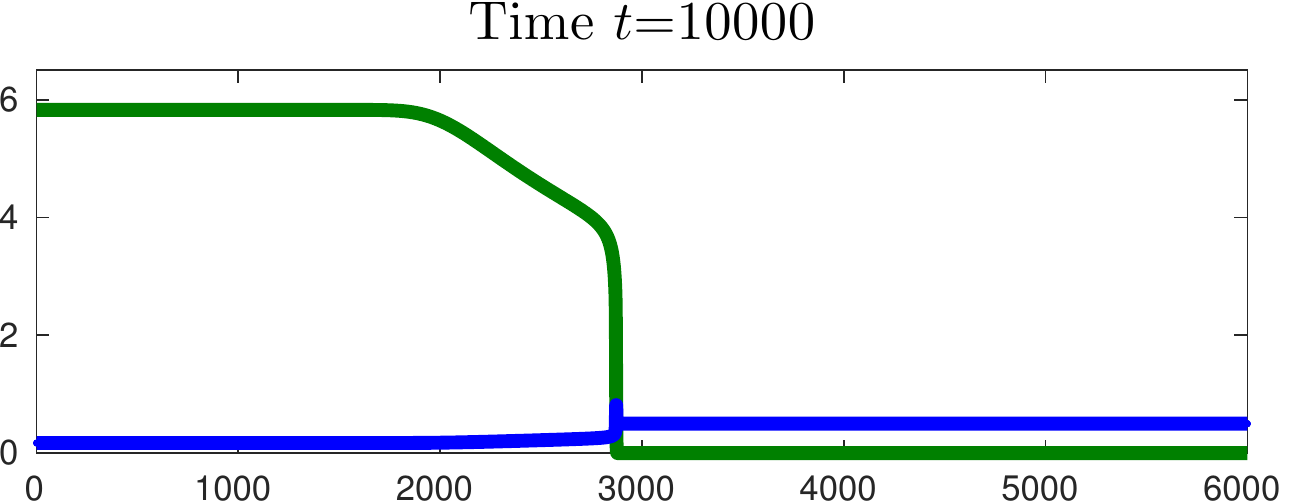}\hfill\includegraphics[width=0.48\textwidth]{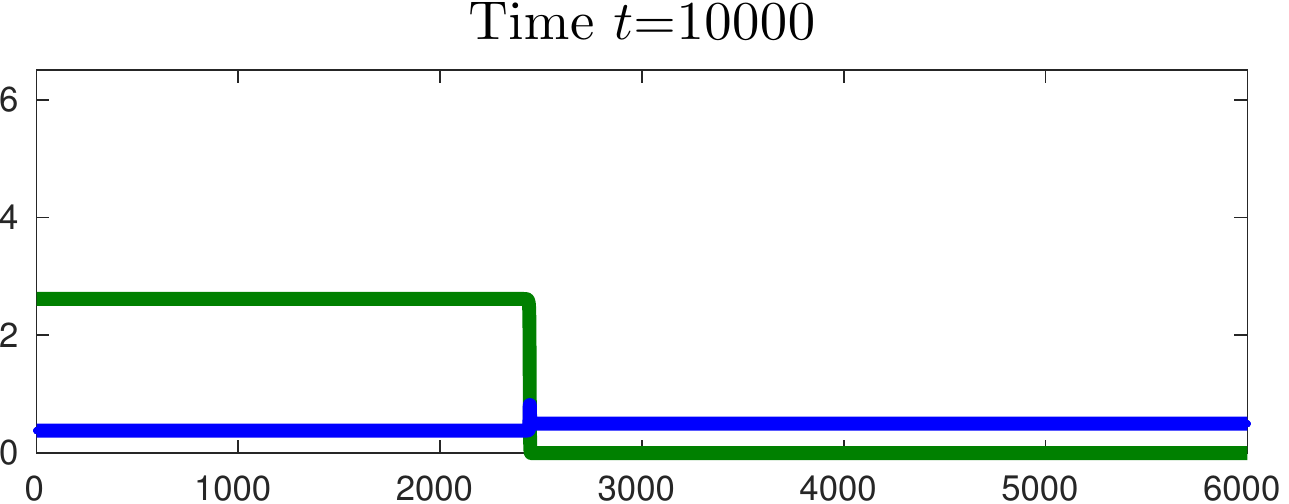}
\end{minipage}\hfill
\begin{minipage}{0.32\textwidth}
 \includegraphics[width=\textwidth]{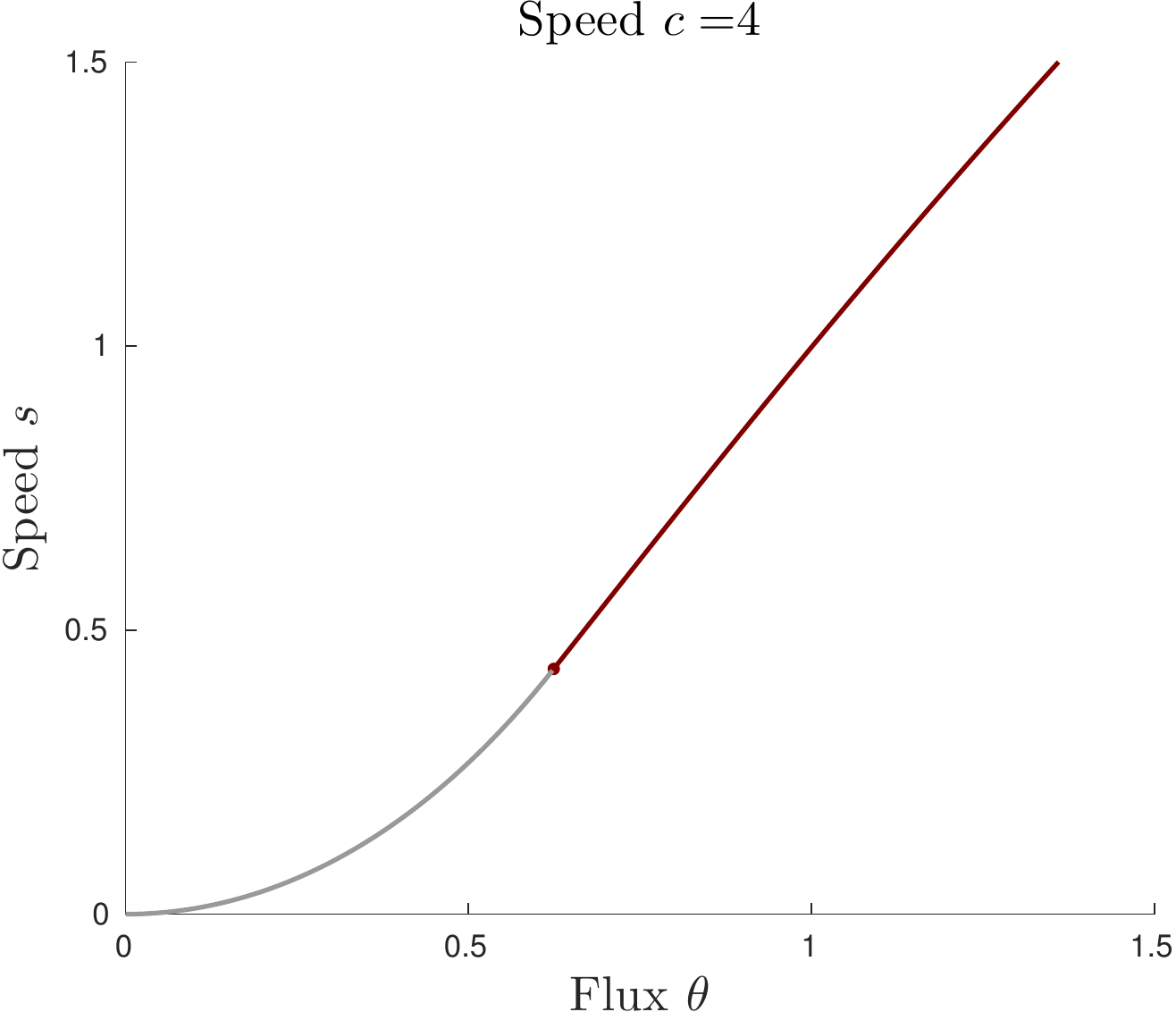}
\end{minipage}
  \caption{Upper edge dynamics for Riemann problems with smaller group velocity, generating an interpolating rarefaction wave (left), and with larger group velocity, generating a bound state of Lax shock and upper edge (middle). Comparisons of predictions from Section \ref{s:ode} with speed measurements in direct simulations for $c=4$ (right).  }\label{f:ue}
\end{figure}
For the lower edge, we set up the reflected Riemann problem. The resulting dynamics are very much equivalent and we omit detailed results, here. Group velocities in the vegetation state larger than the speed of the lower edge lead to rarefaction waves, smaller group velocities lead to bound states of lower edges and Lax shocks. In this sense, nutrient flow below the lower edge selects the vegetation state uphill. 
%
%
%
%
%

%
%
%
%
%

\paragraph{Small diffusion or large advection.}

Small diffusion $d_b=\varepsilon^2$ in \eqref{e:rd} amounts, after scaling of space and time, to a large parameter $c\mapsto c/\varepsilon$ and small effective speeds $s\mapsto \varepsilon s$. Assuming $\eps=0.01$ and $c=1$, we find the bifurcation diagram in the left panel of Figure \ref{f:sc1}, where speeds $s$ need to be multiplied by $\eps$. We find that most of the complexity is now confined to a very narrow region of fluxes $\scflux$, such that for most values of nutrient flow $w_+$, say, we find a 
upper and lower edges with different speeds, and single and periodic vegetation bands with comparatively small speeds ($\scflux\gtrsim 0.2$). 

In particular, small diffusion in the presence of an $\rmO(1)$ value of the nutrient flow $w$ in the vegetation-less state $b=0$ gives $\scflux\sim 1$, $\theta\sim \sqrt{s+1/\varepsilon}$, in \eqref{e:sc}. Therefore, $\theta\gg 1$ such that lower edges occur at $s=\theta^2$, which gives $s\sim 1/\varepsilon$, with resulting unscaled speed $s_\mathrm{eff}\sim 1$. For upper edges, $s\sim \theta^{2/3}$,  which gives $s\sim \varepsilon^{-2/3}$ and effective small speeds $s_\mathrm{eff}\sim \varepsilon^{1/3}$. Speeds of vegetation bands are yet smaller, $s\sim \theta^{-2}$, which gives $s_\mathrm{eff}\sim \varepsilon^{2}$.

\paragraph{The conservation law formalism --- undercompressive versus Lax shocks.}
In a short summary, our main results are existence results for  shocks that are \emph{not} the typical Lax shocks, but rather various types of undercompressive shocks, degenerate Lax shocks,  and spatio-temporally periodic solutions. This aspect of the traveling-wave solutions is illustrated in Figure \ref{f:char}. The existence of these shocks can in some limits be understood in relation to the sideband instability, which  causes the Bogdanov-Takens point and generates curves of homoclinic and periodic solutions. The other ingredient is the inherent difficulty with the reduction to a scalar conservation law \eqref{e:redcl}: the global elimination is ill defined since equilibrium branches $w=1/b$ and $b=0$ are not connected, separated by a region $0<b<1$ with negative effective viscosity, and a justification beyond the small-amplitude approximation in \eqref{e:lwl} seems unrealistic.
\begin{figure}[h]
 \centering\includegraphics[width=\textwidth]{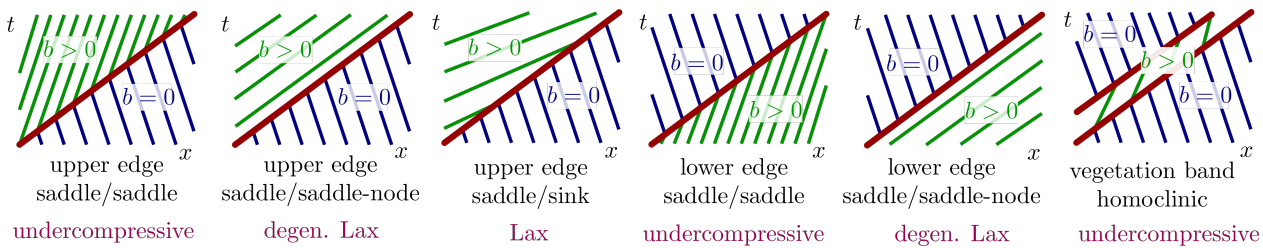}
 \caption{Schematic plot of upper edges, lower edges, and vegetation bands with direction of characteristics added in space-time plots illustrating the undercompressive nature of the traveling-wave solutions found here. Note that characteristics are always sloped to the left, negative speed, when $b=0$, and sloped to the right when $b>0$, while all traveling waves (shocks) propagate to the right, uphill. }\label{f:char}
\end{figure}

\section{Analysis --- existence proofs}\label{s:p}

We first establish existence using monotonicity in $\theta$ and $s$, Section \ref{s:mon}. We then investigate limits of lower and upper edge heteroclinics, Section \ref{s:thinf} and \ref{s:th0}.

\subsection{Monotonicity and existence}\label{s:mon}

We prove Theorem \ref{t:l} in detail. We merely outline the differences in the proof of Theorem \ref{t:u}, which is conceptually similar but requires a different set of coordinates. Both proofs rely on phase plane analysis, constructing invariant regions and using the Poincar\'e-Bendixson Theorem. The key ingredient is a monotonicity, in appropriate coordinates, with respect to $\theta$ and $s$, that implies an ordering of stable and unstable manifolds. Interestingly, the ordering property holds in a different set of coordinates for upper and lower edge heteroclinics. 

\subsubsection{Proof of Theorem \ref{t:l}}
We begin by examining the local behavior of the traveling-wave ODE \eqref{e:tw} near the equilibrium $(0,0)$. The linearization is hyperbolic and we denote the unstable manifold by $\mathcal W^u$ with subscripts $(\theta,s)$ when necessary.

\begin{Proposition}[Position of the Local Unstable Manifold $\mathcal W^u$] \label{local}
The following hold in a sufficiently small neighborhood of $(0,0)$:
\begin{enumerate}
\item $\mathcal W^u$ lies above the line $v=sb$ for any choice of $(\theta,s)$.
\item If $\theta_1<\theta_2$, then the unstable manifold $\mathcal W^u_{\theta_1}$ lies above the unstable manifold $\mathcal W^u_{\theta_2}$ for any fixed $s$.
\item If $s_1<s_2$, then the unstable manifold $\mathcal W^u_{s_2}$ lies above the unstable manifold $\mathcal W^u_{s_1}$ for any fixed $\theta$.
\end{enumerate}
\end{Proposition}

\begin{proof}
The linearization of \eqref{e:tw} at $(0,0)$ is
\[
\begin{bmatrix}
-s & 1\\
1 & 0
\end{bmatrix} \qquad \text{with unstable eigenvector} \qquad \mathbf e_u =
\begin{bmatrix}
2\\
s+\sqrt{s^2+4}
\end{bmatrix}.
\]
The slope of $\mathbf e_u$ is greater than $s$, the slope of the line, thus proving {\it(i)}.

For {\it(ii)}, note that $b'>0$ above the line $v=sb$. Thus, part {\it(i)} allows us to smoothly parametrize the local unstable manifold as a function of $b$. Let $h_i(b)$ be such a parametrization so that $\mathcal W^u_{\theta_i} = \text{graph~} h_i$, for $i = 1,2$. Assuming that $\theta_1<\theta_2$, we will show that $h_1>h_2$ in a neighborhood of $(0,0)$. 

Expanding in $b$, we write 
\[ h_i(b) = h_{i,1} b + h_{i,2} b^2 + \mathcal O(b^3), \qquad \text{for } i = 1,2, \]
where $h_{i,1} = \tfrac{1}{2}(s+\sqrt{s^2+4})$, the slope of $\mathbf e_u$. We compute the derivative $v'_i$ in two ways, first using the chain rule $v_i' = h_i'(b) b'$ and second by plugging into the ODE \eqref{e:tw}. Setting coefficients of $b^2$ equal, we find
\[
h_{i,2}= \frac{-\theta_i(s+3\sqrt{s^2+4})}{2s^2+9}.
\]
So $h_{1,2} > h_{2,2}$ and we may choose $b$ sufficiently small to guarantee that $h_1(b) > h_2(b)$.

Part {\it(iii)} follows from a similar argument, made even simpler by the fact that the monotonicity is encoded in the linear coefficient in the expansion of $h_i$.
\end{proof}

\paragraph{The phase plane.} 
For given parameter values of $(\theta, s)$, we define a region in the phase plane of \eqref{e:tw}
\[ 
\displaystyle\Sigma \defeq \left\{(b,v)\mid b\geq 0, \max\{sb, \theta-\tfrac{1}{b}\}\leq v\leq \theta\right\},
\]
bounded below by the nullclines
\begin{align*}
\gamma	&\defeq	\{b'=0\} = \{ (b, v) \mid v = sb \}\\
\alpha	&\defeq	\{ v' =0 \} = \left\{ (b, v) \mid  v = \theta - \tfrac{1}{b} \right\},
\end{align*}
and bounded above by the line $\beta \defeq \{ v= \theta \}$; see Figure \ref{f:simple_bif}. 

From Proposition \ref{local}, we know that, near $(0,0)$, $\mathcal W^u$ starts in $\Sigma$. 
By computing the direction of the vector field on $\partial \Sigma$, one sees that $\mathcal W^u$ may only exit $\Sigma$ by crossing $\alpha$ or $\beta$. We treat $\Sigma$ as a semi-invariant region and formalize a dichotomy in Proposition \ref{dichotomy}.

%
%
%
%
%
%
%

\begin{figure}
\centering\includegraphics[width=0.9\textwidth]{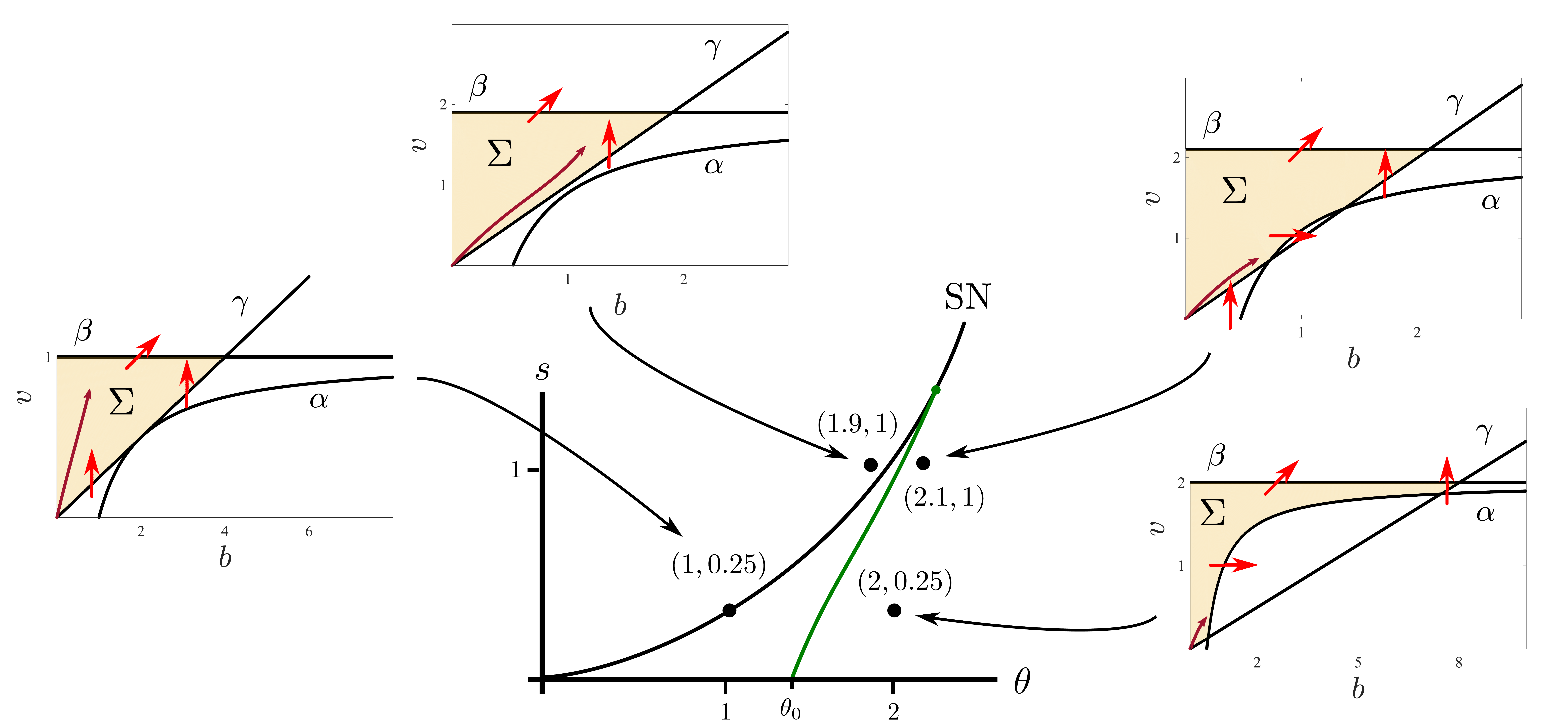}
\caption{A partial bifurcation diagram including the saddle-node curve \eqref{e:sn} and the existence curve of lower edge heteroclinics $s_\ell(\theta)$ (green). For various parameter values $(\theta,s)$ insets show: the semi-invariant region $\Sigma$ (shaded gold); its boundaries $\alpha, \beta, \gamma$ (black); selected arrows from the vector field (red); and an initial segment of the unstable manifold $\mathcal W^u$ (maroon).}\label{f:simple_bif}
\end{figure}

\paragraph{Dichotomy in the parameter plane.} We define two sets in the $(\theta,s)-$plane 
\begin{align*}
\Su	&\defeq	\{(\theta, s) \mid \mathcal W^u\cap \beta \neq \emptyset \} \\
\Sd	&\defeq	\{(\theta, s) \mid \mathcal W^u\cap \alpha  \neq \emptyset \}.
\end{align*}
We claim that both these sets are nonempty. For $\Su$, choose $\theta = 1$ and $s>\theta^2/4$. This choice lies above the saddle-node curve \eqref{e:sn} and so the curves $\alpha$ and $\gamma$ do not intersect. The vector field points only in the positive $v$ direction along the line $\gamma$, so $\mathcal W^u$ cannot cross $\gamma$ from above. Since $\mathcal W^u$ starts above $\gamma$ it must stay above $\gamma$. Thus it eventually intersects $\beta$. This type of invariant-region argument, verified by comparing slopes, will recur throughout the proof.

For $\Sd$, we choose $(\theta,s) = (4,1)$, which lies to the right of the Hopf curve \eqref{e:hopf}. By Proposition \ref{local}, near the origin we have that $\mathcal W^u$ lies beneath the line $v = 4 b$. The slope of the vector field along $v = 4b$ is less than $4$, so $\mathcal W^u$ must remain below $v = 4b$. This forces an intersection with $\alpha$ because $v = 4 b$ is tangent to $\alpha$ at $b = \tfrac{1}{2}$.

Additionally, these two sets are disjoint. If they were not, we let $(b_*,v_*)$ be the point at which $\mathcal W^u$ first exits $\Sigma$. If $(b_*,v_*)\in \beta$, then for $b>b_*$ the unstable manifold must remain in the invariant rectangle $\{ (b,v)\mid b>b_*, v>\theta \}$ and thus can never intersect $\alpha$. If $(b_*,v_*)\in \alpha$, then after $(b_*,v_*)$ the unstable manifold $\mathcal W^u$ must remain below itself for $b<b_*$ and below the line $\{(b,v)\mid v = v_*\}$ for $b>b_*$ 
This nearly proves the next proposition.

\begin{Proposition}[The $\Su,\Sd$ Dichotomy] \label{dichotomy}
The interior of the first quadrant of the $(\theta,s)$ parameter plane is partitioned into the two nonempty sets $\Su$ and $\Sd$.
\end{Proposition}

\begin{proof}
We have already shown $\Su$ and $\Sd$ are nonempty and disjoint. We only need to show that the union $\Su\cup \Sd$ covers the first quadrant. For any $(\theta,s)$, Proposition \ref{local} implies that $\mathcal W^u$ begins in $\Sigma$. For convenience, let $s>0$ so that the region $\Sigma$ is bounded with no equilibria in its interior. Thus, by the Poincar\'e-Bendixson Theorem the unstable manifold must either exit $\Sigma$ or converge to an equilibrium on $\partial \Sigma$. If $\mathcal W^u$ exits $\Sigma$, it must intersect either $\alpha$ or $\beta$ to do so, as previously stated.\footnote{This dichotomy is central to the proof and the reader may verify the statement by computing the vector field on $\partial \Sigma$.} 
If $\mathcal W^u$ does not exit $\Sigma$, it must converge to one of the three equilibria $(0,0), (b_-,v_-),(b_+,v_+)\in \partial \Sigma$. Since $b'>0$ in $\Sigma$, the unstable manifold cannot return to $(0,0)$ without exiting $\Sigma$. Meanwhile $(b_\pm, v_\pm)$ are both contained in $\alpha$.
\end{proof}

\begin{Remark}
If $s = 0$, the region $\Sigma$ is infinite and the equilibrium $(b_+,v_+)$ does not exist. Assuming $\mathcal W^u$ does not exit $\Sigma$, it is squeezed between $\alpha$ and $\beta$. So it must converge to the line $v = \theta$. In a sense, one may think of this situation as consistent with the proposition since the equilibrium $(b_+,v_+) \to \{v = \theta\}$ as $s\to 0$. For more detail, see the argument in Section \ref{s:th0}.
\end{Remark}

The local Proposition \ref{local} has consequences in the whole region $\Sigma$.

\begin{Proposition}[Relative Positioning of $\mathcal W^u$] \label{above}
Within the region $\Sigma$, the following hold:
\begin{enumerate}
\item For any fixed $s$ and any $\theta_1<\theta_2$, the global unstable manifold $\mathcal W^u_{\theta_1}$ lies above $\mathcal W^u_{\theta_2}$. \label{abovetheta}
\item For any fixed $\theta$ and any $s_1<s_2$, the global unstable manifold $\mathcal W^u_{s_1}$ lies below $\mathcal W^u_{s_2}$. \label{aboves}
\end{enumerate}
\end{Proposition}

\begin{proof}
We prove part {\it(i)} and omit the similar proof of part \textit{(ii)}.

Since $b'>0$ in $\Sigma$, both unstable manifolds are functions of $b$ and so the only way for them to switch their relative positioning is by intersecting.
For the sake of contradiction, suppose there is an intersection $(b_*,v_*)\in \Sigma$.  Without loss of generality, we may assume that $(b_*,v_*)$ is the first such intersection. We know by Proposition \ref{local} that  $\mathcal W^u_{\theta_1}$ lies above $\mathcal W^u_{\theta_2}$ for all $b<b_*$. However, computing the vector fields at $(b_*,v_*)$ with each $\theta_i$ reveals that the one for $\theta_1$ has a larger slope. This contradicts the fact that, due to their relative positions, the $v$ value of $\mathcal W^u_{\theta_2}$ must be increasing at least as fast as that of $\mathcal W^u_{\theta_1}$ in order to have an intersection.
\end{proof}

\begin{Corollary}[Rectangular Subsets]\label{rect} \quad \\
\vspace{-.2 in}
\begin{enumerate}
\item If $(\theta_*,s_*)\in \Su$, then $(\theta,s)\in \Su$ for all $\theta < \theta_*$ and $s>s_*$.
\item If $(\theta_*,s_*)\in \Sd$, then $(\theta,s)\in \Su$ for all $\theta > \theta_*$ and $s<s_*$.
\end{enumerate}
\end{Corollary}

\begin{proof}
Again, we prove part {\it(i)} and omit the similar proof of part \textit{(ii)}.

Suppose that $(\theta_*,s_*)\in \Su$ and that $\theta < \theta_*$. Fixing $s= s_*$ the proposition implies that $\mathcal W^u_{\theta}$ lies above $\mathcal W^u_{\theta_*}$ (even in the larger region $\Sigma_{\theta_*}$). Thus $\mathcal W^u_\theta$ intersects $\beta_{\theta_*} = \{v = \theta_*\}$. But since $\beta_\theta =  \{v = \theta\}$ is strictly below $\beta_{\theta_*}$, $\mathcal W^u_{\theta}$ must have also intersected $\beta_\theta$ and so $(\theta, s_*)\in \Su$. Now suppose that $s>s_*$ and fix $\theta$. The second part of Proposition \ref{above} implies that $\mathcal W^u_{s_*}$ lies below $\mathcal W^u_{s}$. Thus $\mathcal W^u_{s}$ must intersect $\beta_\theta$, so $(\theta,s)\in\Su$.
\end{proof}

This powerful statement immediately guarantees that $\Su, \Sd$ are path connected (use paths along edges of rectangles) and that the common boundary $\partial \mathcal S \defeq \partial \Su \cap \partial \Sd$ is non-decreasing.

\paragraph{The curve of lower-edge heteroclinics $s_{\lindex}(\theta)$.}
The next proposition implies that the common boundary $\partial \mathcal S = \partial \Su \cap \partial \Sd$ forms a curve of parameter values for which the system \eqref{e:tw} has a heteroclinic orbit that corresponds to the desired traveling wave of equation \eqref{e:rds}.
\begin{Proposition}\label{sup}
Fix an arbitrary $\theta_*$ and suppose that $\Sd\cap \{\theta = \theta_*\}$ is nonempty and bounded above. Let
\[ s_* = \sup_{s\geq 0} \left(\Sd\cap \{\theta = \theta_*\}\right). \]
Then $(b_+,v_+)\in \mathcal W^u_{(\theta_*,s_*)}$ and $s_*$ is the only $s$ value in $\Sd\cap \{\theta = \theta_*\}$ with this property.
\end{Proposition}

\begin{proof}
With $\theta_*$ and $s_*$ as above, suppose that $(b_+,v_+)\notin \mathcal W^u_{(\theta_*,s_*)}$. We treat two cases.

First, suppose that $(\theta_*,s_*)\in \Su$. There exists $ s < s_*$ arbitrarily close to $s_*$ with $(\theta_*,s)\in \Sd$. But this is impossible because the unstable manifold $\mathcal W^u_{s}$ is continuous in the parameter $s$.

Second, suppose that $(\theta_*,s_*)\in \Sd$. Then  $\mathcal W^u_{s_*}$ must intersect $\alpha$ at a point below $(b_+,v_+)$. Since the unstable manifold is continuous in parameters, there must be a $ s > s_*$ such that $\mathcal W^u_{s}$ intersects $\alpha$ near the intersection $\mathcal W^u_{s_*}\cap \alpha$. But now $(\theta_*,s) \in \Sd$, contradicting the maximality of $s_*$.

Now we show that $s_*$ is unique. Suppose that there is an $s\neq s_*$ such $(b_+,v_+)_s\in \mathcal W^u_{(\theta_*,s)}$. Clearly, $ (\theta_*,s)\in\Sd$, so we must know that $ s< s_*$. Then, by Proposition \ref{above}, $\mathcal W^u_{s}$ lies below $\mathcal W^u_{s_*}$. However, $(b_+,v_+)_s$ is above $(b_+,v_+)_{s_*}$, so $\mathcal W^u_{s}$ could never reach $(b_+,v_+)_s$.
\end{proof}

Corollary \ref{rect} implies that $\partial\mathcal S$ is non-decreasing. We parameterize $\partial \mathcal S$ by a function $s_{\lindex}(\theta)$ so that $\partial \mathcal S = \text{graph } s_{\lindex}$. The continuity of $\mathcal W^u$ with respect to parameters implies that $s_{\lindex}(\theta)$ is continuous. All that remains is to characterize the location and some properties of this boundary curve by describing the sets $\Su$ and $\Sd$.

\begin{Proposition}[Location and Shape of $\partial \mathcal S$] \label{curve}
In the first quadrant of the $(\theta,s)-$plane, we have:
\begin{enumerate}
\item The subset $\{s>\theta^2/4\}\subseteq \Su$. So the intersection of Proposition \ref{sup} is always bounded above.
\item The subset $\{ s = 0, \theta\leq1\}\subseteq \Su$. So the intersection of Proposition \ref{sup} is empty for $\theta_*\leq 1$.
\item The subset $\{ s = 0, \theta\geq2\}\subseteq \Sd$. So the intersection of Proposition \ref{sup} is nonempty for $\theta_*\geq2$.
\item The subset $\{ s = \theta^2/4, \theta>2^{5/4}\}\subseteq \Sd$. In fact, the saddle-node $(b_+,v_+)\in \mathcal W^u$ for these parameter values.
\end{enumerate}
\end{Proposition}

\begin{proof}
We use the same construction for the proof of each statement, which are all similar to the earlier proofs of the nonemptiness of $\Sd, \Su$. Define
\[
\begin{array}{r l c l}
\mathcal L & \defeq \{ v = cb\} & & \text{for any slope } 0<c\leq 1,\\
M & \defeq \displaystyle\frac{-b^2(\theta-cb)+b}{cb-sb} & & \text{the slope of the vector field \eqref{e:tw} on } \mathcal L.
\end{array}
\]
For given parameter values $(\theta,s)$, we choose $c$ such that $\mathcal W^u$ is above (or below, as appropriate) $\mathcal L$ in a neighborhood of $(b,v) = (0,0)$. Such choices of $c$ may be determined by examining the expansion of $\mathcal W^u$ as computed in Proposition \ref{local}. Then, we may compare $c$ and $M$ to ensure that $\mathcal W^u$ stays above (or below) $\mathcal L$.

For {\it (i)}, let $\theta>0$ and $s>\theta^2/4$. Choose $c = s$, so $\mathcal L = \gamma$, and note that $\mathcal W^u$ is above $\mathcal L$ in a neighborhood of $(0,0)$. Now $M$ is undefined because near $\mathcal L$ the slope of the vector field increases to $\infty$. Thus $\mathcal W^u$ is trapped above $\mathcal L$. Since $(\theta,s)$ lies above the saddle-node curve, $\alpha$ lies strictly below $\gamma = \mathcal L$ and so $(\theta,s)\in \Su$.

For {\it (ii)}, let $s = 0$ and $\theta = 1$. Choose $c = \tfrac{1}{2}$ and note that $\mathcal W^u$ is above $\mathcal L$ a neighborhood of $(0,0)$. For all $b>0$, we have $M > c$, so $\mathcal W^u$ cannot cross $\mathcal L$ from above. Note that $\alpha$ is below $\mathcal L$ and so $(1,0)\in \Su$. Corollary \ref{rect} finishes the proof.

For {\it (iii)}, let $s=0$ and $\theta=2$. Choose $c = 1$ and note that $\mathcal W^u$ is below $\mathcal L$ in a neighborhood of $(0,0)$. Now, for $b<1$ we have $M<1=c$, so $\mathcal W^u$ cannot cross $\mathcal L$ from below. Note that $\alpha$ is below $\mathcal L$, except for a single intersection $(b,v) = (1,1)$. Thus $\mathcal W^u$ must intersect $\alpha$ and so $(2,0)\in \Sd$. Corollary \ref{rect} finishes the proof.

For {\it (iv)}, let $\theta > 2^{5/4}$ and $s = \theta^2/4$. Choose $c = s = \theta^2/4$. As in the proof of {\it (i)}, $\mathcal W^u$ is trapped above $\mathcal L$. Now let
$\mathcal L_2 \defeq \{ v = \tfrac{\theta^2}{8}b+\tfrac{\theta}{4} \}$ and $M$ be the slope of the vector field on $\mathcal L_2$. Clearly, $\mathcal W^u$ is below $\mathcal L_2$ in a neighborhood of $(0,0)$ and a brief computation shows that $M<\tfrac{\theta^2}{8}$, the slope of $\mathcal L_2$, for $b<\tfrac{2}{\theta}=b_+$. Thus, we've shown that $\mathcal W^u$ is trapped between $\mathcal L$ and $\mathcal L_2$ and so it must converge to the only equilibrium in the region, which is $(b_+,v_+)$.
\end{proof}

To summarize, we have shown that the curve starts on the line $\{s=0\}$ at some $1<\theta_0\leq 2$. We've shown that the curve is defined for all $\theta\geq 2$, and it is well-defined by construction. Finally, we showed that the curve lands on and joins the saddle-node starting at some $\theta < 2^{5/4}$; see Figure \ref{f:simple_bif}.

\subsubsection{Outline of Proof of Theorem \ref{t:u}}

Our proof of the existence of the curve $s_u(\theta)$ of upper edge traveling waves uses the same ideas as the proof above. These traveling waves correspond to heteroclinic orbits $(b_+,v_+) \to (0,0)$ of the ODE \eqref{e:tw}. The first difference is that we choose new coordinates in which to do the phase plane analysis. We apply the coordinate transformation $(b,v) \mapsto (b, u) = (b, v-sb)$ to arrive at the ODE
\begin{align}
 b'&=u,\notag\\
 u'&=-b^2(\theta-u-sb)+b-su.\label{e:twu}
\end{align}
Note that this transformation shifts the important nullcline $\gamma$ to the line $\{u=0\}$, so the rest of the analysis occurs in the lower half plane $\{u<0\}$. By examining the local stable manifold $\mathcal W^s$ of $(b,u) = (0,0)$, we obtain local results on the positioning of $\mathcal W^s$ with respect to the parameters $(\theta,s)$, which are precisely reversed from the lower edge proof. We extend this local monotonicity to a statement about the relative positions of the global unstable manifolds $\mathcal W^s_{(\theta,s)}$, for various $(\theta, s)$, within a large semi-invariant region 
\[ 
\Sigma' \defeq \left \{ (b,u) \mid 0\leq b\leq b_+, \tfrac{-1}{\eps}b< u < 0\right \},  \quad \text{for small } \eps>0 \text{ depending on } s.
\]
The region $\Sigma'$ has only two entrance sets that, when considered in backwards time, correspond to the exit sets $\alpha$ and $\beta$ in the proof above. We use these entrance sets to characterize a dichotomy in the $(\theta,s)-$plane,
\begin{align*}
\mathcal S_R	&\defeq	\left\{(\theta, s) \mid \mathcal W^s\cap \{b_-\leq b\leq b_+, u=0\} \neq \emptyset\right \} \\
\mathcal S_U    &\defeq	\left\{(\theta, s) \mid \mathcal W^s\cap \{b=b_+,\tfrac{-1}{\eps}b_+< u< 0\} \neq \emptyset \right\}.
\end{align*}
Here, the set $\mathcal S_R$ plays the role of $\Su$ above while the set $\mathcal S_U$ plays the role of $\Sd$ above. Just as above, these sets have an analogous ``rectangular subsets" property due to the relative positioning of unstable manifolds within $\Sigma'$. Again, one has $\{s>\theta^2/4\}\subseteq \mathcal S_R$ but this time, for any $\theta$, one shows that $(\theta,s)\in \mathcal S_U$ for arbitrarily small $s$. Thus both sets are nonempty and $\mathcal S_U$ is bounded above for each fixed $\theta_*$. Setting 
\[
s_* = \sup_{s\geq 0} \left (\Su\cap \{\theta = \theta_*\} \right),
\]
we have $(b_+,0)\in \mathcal W^s_{(\theta_*,s_*)}$. The asymptotic results appear in the next section.

\subsection{Large $\theta$}\label{s:thinf}
We study the existence of upper and lower edge heteroclinics, and of homoclinics, in the limit $\theta\to\infty$. In each case, we introduce a suitable scaling which allow us to construct heteroclinic and homoclinic orbits based on simple transversality arguments. 

\paragraph{Vegetation bands.} 
We scale $\theta=1/\eps$, $v=\eps\tilde{v}$, $b=\beta$, which gives 
\begin{align}
 \beta'&=\eta - s\beta,\nonumber\\
 \eta'&=-(1-\eps^2\eta)\beta^2+\beta.\label{e:scal1}
\end{align}
At $s=\eps=0$, we find a homoclinic $\beta_*(x)=\frac{3}{2}\mathrm{sech}\,^2(x)$ to the origin from explicitly solving $\beta''-\beta+\beta^2=0$. We now follow standard Melnikov theory along homoclinics, as laid out for instance in \cite{chh}. One writes \eqref{e:scal1} as an equation $F(\beta,\eta;s,\eps)=0$, defined as a smooth map from $H^1(\R,\R^2)\times\R^2$ into $L^2(\R,\R^2)$. The linearization at the homoclinic, $s=\eps=0$, is Fredholm of index 0 and possesses a one-dimensional cokernel, given through the unique (up to scalar multiples) solution to the adjoint equation, $\psi=(-\beta_*'',\beta_*')^T$. One also computes the derivative of $F$ with respect to $s$ and $\eps^2$ at the homoclinic, which gives,
\[
 \partial_s F=(-\beta_*,0)^T,\qquad \partial_{\eps^2}F=(0,\beta_*'\beta_*^2)^T.
\]
After Lyapunov-Schmidt reduction, we find the leading-order equation
\[
 \langle \partial_s F,\psi\rangle s+\langle\partial_{\eps^2}F,\psi\rangle \eps^2+\rmO\left(|s|^2+\eps^4))\right)=0,
\]
where $\langle \cdot,\cdot\rangle $ denotes the $L^2$-inner product. Evaluating the relevant integrals, one readily finds at leading order
$-\frac{6}{5}s+\frac{36}{35}\eps^2=0,$ hence 
\[
  s=\frac{6}{7}\theta^{-2}+\rmO(\theta^{-4}).                                              
\]

\paragraph{Upper edge.}
We set $\theta=\eps^{-3}$, $s=\sigma \eps^{-2}$, $b=\beta \eps^{-1}$,  $v=\eta\eps^{-3}$, and $x=\eps^2 y$, and  obtain 
\begin{align}
 \beta_y&=\eta - \sigma \beta\\
 \eta_y&=-(1-\eta)\beta^2 + \eps^4\beta.
\end{align}
This planar system possesses equilibria $\eta=\beta=0$ and $\eta=1,\beta=\sigma^{-1}$. Linearizing at the origin reveals a one-dimensional stable subspace spanned by $(1,0)^T$ and a one-dimensional center subspace spanned by $(1,\sigma)^T$. In the one-dimensional corresponding center manifold, the flow is given to leading order by $\eta_y=-\sigma^{-2}\eta^2$, such that the equilibrium is stable in the positive quadrant. The other equilibrium is a saddle. Elementary phase plane analysis similar to the analysis for finite speed, presented above, reveals the existence of a unique value $\sigma_*>0$ for which the system possesses a heteroclinic orbit connecting the unstable manifold of the non-trivial equilibrium and the strong stable manifold of the origin. Moreover, this intersection is transverse in the parameter $\sigma$. Perturbations in $\eps$ unfold the saddle-node in a transcritical bifurcation with an equilibrium bifurcating into the positive quadrant. There hence exists $\sigma=\sigma_*+\rmO(\eps^4)$ for which a connecting orbit between the equilibrium $\beta=\sigma^{-1}+\rmO(\eps^4)$ and the origin exists. Scaling back gives 
\[
 s=s_\infty\theta^{2/3}+\rmO(\theta^{-2/3}).
\]
Numerically, we find $s_\infty\sim 0.9055$. 

\paragraph{Lower edge.}
Here, we scale $\theta=\eps^{-1}$, $s=\sigma \eps^{-2}$, $v=\eta \eps^{-1}$, $b=\beta\eps$, $x=\eps^2 y$, to find 
\begin{align}
 \beta_y&=\eta-\sigma\beta\nonumber\\
 \eta_y&=\eps^4\left(-(1-\eta)\beta^2+\beta\right).
\end{align}
This slow-fast system possesses a one-dimensional slow manifold \cite{fen} $\beta=\sigma \eta+\rmO(\eps^4)$, with reduced slow flow, projected on the $\eta$-axis, 
\[
 \eta_y=\eps^4\left(-\sigma^{-2}(1-\eta)\eta^2+\sigma^{-1}\eta\right).
\]
In this cubic nonlinearity, a heteroclinic connecting the left-most equilibrium $\eta=0$ and the right-most equilibrium exists precisely when the right-most equilibrium is double, for $\sigma=1/4$, that is, on the saddle-node curve where $s=\theta^2/4$.

\subsection{Heteroclinic limits --- small $s$}\label{s:th0}
We study small-speed limits of upper and lower edge homoclinics.
\paragraph{Upper edge.} We scale $\theta=\eps$ and set $s=\sigma \eps^2$, which gives 
\begin{align}
 b'&=v-\sigma\eps^2 b,\nonumber\\
 v'&=-(\eps-v)b^2+b.
\end{align}
At $\eps=0$, we find the simple system 
\begin{align}
 b'&=v,\nonumber\\
 v'&=vb^2+b.\label{e:bch}
\end{align}
Setting $b=1/\beta$, thus compactifying the plane in the $b$-direction, we find
\begin{align}
 \beta_y&=-\beta^4v,\nonumber\\
 v_y&=v+\beta,\label{e:betch}
\end{align}
where we used a nonlinear rescaling of time $\beta^2 \partial_x=\partial_y$. Equations \eqref{e:bch} and \eqref{e:betch} together define a flow on $v\in\R$, $b\leq1$, $\beta\geq 1$, patching continuously at $b=\beta=1$. There are precisely two equilibria $b=v=0$, a saddle, and $\beta=v=0$, with a one-dimensional strong unstable manifold and a one-dimensional center manifold, tangent to $v=-\beta$, with leading-order flow $\beta_y=\beta^3$. One readily establishes the existence of a connecting orbit by continuing the stable manifold of the origin into $b>0$ by flowing backward, and exploiting that $\beta=v=0$ is asymptotically stable in backward time within $\beta\geq 0$. Perturbing in $\eps$, two equilibria bifurcate within the center manifold. By continuity, the stable manifold always connects to the left-most non-trivial equilibrium, such that we can have connecting orbits to the right-most equilibrium only when it coincides with the middle equilibrium, for $s=4\theta^2$. 

\paragraph{Lower edge.} We complement the system 
\begin{align}
 b'&=v-sb\nonumber\\
 v'&=-(\theta-v)b^2+b,\label{e:1}
\end{align}
with the compactification $\beta=1/b$, 
\begin{align}
 \beta_y&=-\beta^4v+s\beta^2\nonumber\\
 v_y&=-(\theta -v)+\beta,\label{e:2}
\end{align}
after a nonlinear rescaling of time $\beta^2\partial_x=\partial_y$. Similar to the previous limit, \eqref{e:1} and \eqref{e:2} together define a flow in $b\leq 1$, $\beta\geq 1$, with trivial gluing at $b=\beta=1$. At $s=0$, the system possesses three equilibria. The origin is a saddle, $\beta=\theta,v=0$ is totally unstable, and $\beta=0$, $v=\theta$ possesses a one-dimensional strong unstable manifold within $\beta=0$, and a center manifold tangent to $v=\theta-\beta$, with local flow 
\[
 \beta_y=-\beta^4\theta+s\beta^2+\rmO(\beta^5)+\rmO(s^2).
\]
For $\theta>0$, the equilibrium $\beta=0$ therefore possesses a one-dimensional stable manifold in $\beta\leq 0$, given by the center manifold. Using the comparison techniques used above, one establishes the existence of a heteroclinic orbit connecting the origin and the equilibrium in $\beta=0$ for some $\theta_0>0$. One also shows that the heteroclinic is transversely unfolded in $\theta$. For $s>0$, small, an equilibrium $\beta=\sqrt{s}$ bifurcates within the center manifold, with one-dimensional stable manifold given by the stable manifold, hence smoothly depending on $s$ and $\theta$. As a consequence, the heteroclinic orbit persists for finite $s$ as a heteroclinic to the finite right-most equilibrium.

\section{Discussion}\label{s:d}

We presented a simplistic model for the conversion of nutrients to biomass in the presence of advection. We analyzed traveling-wave solutions and explained an analogy with viscous scalar conservation laws. We now briefly discuss generalizations and possible extensions. 

\paragraph{Threshold conversion.} One can easily envision other applications, where an ingredient $w$ is converted into a product $b$, with rate function $f(b,w)$. The ingredient  $w$ is ``supplied'' through a constant-speed mean flow, and the product $b$ simply diffuses. Our rate function illustrates a threshold behavior in $f$, where curves $\Gamma=\{b(s),w(s)|s\in]\R\}$ of equilibrium concentrations, $f(b(s),w(s))=0$, are not monotone in $b$. In other words, increasing the concentration of the ingredient $w$ at equilibrium may not result in a continuous change of $b$. This lack of monotonicity is at the origin of both sideband instabilities and the existence of undercompressive shocks. We believe that the methods here generalize to a much larger class of rate functions $f$, while details of the bifurcation diagram will of course vary. It is worth noticing that much of the information on group velocities and instabilities is contained in the geometry of the equilibrium curve, as made explicit in Remark \ref{r:geom}. 

\paragraph{Wavenumber selection.} Our model does not contain Turing instabilities in the sense that at onset of an instability, the fastest growing Fourier mode of the linearization would be nonzero. The only instabilities present are sideband instabilities, leading to notoriously complex dynamics. We notice however that, similar to the case of diffusive transport of water discussed below, invasion processes do select distinct wavenumbers in their wake; see Figure \ref{f:d}. For most parameter values, those wavenumbers can be predicted from a linear analysis \cite{vansaarloos,holzerscheel}. We did not perform a systematic study here of this selection mechanism, but refer to \cite{sherratt1} for a scenario where wavenumber selection in the presence of invasion can possibly yield information on the origin of banded patterns.

\paragraph{Advective transport versus diffusion.} The case when the ingredient $w$ is diffusing rather than advected is in many ways much simpler,
\begin{equation}\label{e:rdd}
 b_t=b_{xx}+f(b,w),\qquad w_t=dw_{xx}-f(b,w).
\end{equation}
Stationary solutions can be obtained by noticing that $b+dw\equiv \theta$, and then solving the scalar equation $b_{xx}+f(b,\theta-dw)=0$. Equations of this type have been studied in many contexts 
\cite{gohmesuroscheel2,gohmesuroscheel,morita,morikeshet,poganspikes}, showing that in many cases stationary layers are the key ingredient. In particular, for threshold-type kinetics as described above, with for instance $f(b,w)=b(1-b)(b-a)-\gamma w$, the system is equivalent to the phase-field system and possesses a Lyapunov function, provided $d>1$ \cite{gohmesuroscheel2}. When $d<1$, traveling fronts bifurcate from the stationary layer solutions \cite{poganbif}. For $d>1$, one observes slow coarsening of layers and only stable solutions (energy minimizers) are single layer solutions or constants \cite{poganscheelzumbrun}. Quite analogous to our system, \eqref{e:rdd} does not possess an inherent wavenumber selection mechanism, that is, a fastest-growing mode analysis shows selected wavenumbers close to zero near onset \cite{kotzagiannidis}. Invasion fronts do however exhibit predictable wavenumber selection mechanisms \cite{gohmesuroscheel2,kotzagiannidis,scheelstevens}, yielding phenomena similar to Figure \ref{f:d}. It is worth noticing that linear stability information is qualitatively contained in information on the geometry of the curves of equilibria, in analogy to Remark \ref{r:geom} on the sign of group velocities, here; see \cite{gohmesuroscheel2}.

\paragraph{Sideband instabilities, breakup, and scale-free patterns.} 
One of the robust predictions here is the occurrence of sideband instabilities prior to a saddle-node bifurcation in which the vegetation state disappears; see Section \ref{s:ode}. Near sideband instabilities, one expects a description of spatio-temporal dynamics in terms of Kuramoto-Sivashinsky-like equations, with additional third-order dispersion, which in turn tend to exhibit spatio-temporally chaotic, sustained dynamics. This correlates well with the observation of disorganized, ``scalefree'' patterns near the edges of banded zones; see for instance \cite{scalefree} for a discussion of observations and mechanisms for scalefree vegetation patterns.

\paragraph{Stability.}
The next most natural question would appear to be for stability of the traveling waves found here. While the elementary phase-plane analysis insinuates that stability questions might be accessible, we did not attempt such a study. Numerically, we did not see instabilities of upper and lower edges (undercompressive shocks), except for instabilities in the Lax case, when one of the asymptotic states undergoes a sideband instability. Beyond analytical results, it would be interesting to add, even numerically, stability information to the rather comprehensive numerical diagram established here. We suspect that stability of periodic traveling waves will reveal a plethora of instability mechanisms that might guide through qualitative transitions between patterns. 

\paragraph{Beyond mass conservation.} As we emphasized early on, the present model is to be understood as a building block for more realistic and complex models. Adding source terms, such as evaporation and rain fall promises yet more complexity. In the case of simple diffusion, the effect of source terms on models with conservation laws was studied in \cite{morita2,champneys,champ2}, revealing in particular the presence of localized patches of periodic structures. 

\paragraph{Two space-dimensions and topography.} Comparing with natural patterns, an important next step will be an analysis in two dimensions. For instance, including the stability of patterns found here with respect to two-dimensional perturbations and the existence and properties of banded patterns not aligned with level sets, functions of $n_x x + n_y y$, with $|(n_x,n_y)|=1$. In fact, much of the analysis here can be adapted to this situation in a straightforward fashion. One would then wish to explain phenomena such as the alignment of bands perpendicular to the slope, or the deformation of bands in non-uniform slopes. 

\small


%
%
%
%

\end{document}